\documentclass[11pt]{article}

\usepackage{amsmath, amsfonts, amsthm, amssymb,verbatim}
\usepackage[mathscr]{euscript}
\usepackage[a4paper,nohead,twoside]{geometry}
\usepackage[english]{babel}
\usepackage{xspace}
\usepackage{enumitem}
\usepackage{appendix}
 
\usepackage{perpage} 
\MakePerPage{footnote}

\newtheorem{theorem}{Theorem}[section] 
\newtheorem{definition}[theorem]{Definition}
\newtheorem{lemma}[theorem]{Lemma}
\newtheorem{corollary}[theorem]{Corollary}

\newtheorem{proposition}[theorem]{Proposition}

\numberwithin{equation}{section}
 
\usepackage[bookmarks=true,
            bookmarksopen=true,colorlinks=true,breaklinks=true,linkcolor=blue,
            citecolor=blue]{hyperref}

\newcommand \be {\begin{equation}}
\newcommand \ee {\end{equation}}

\newcommand \N {\mathbb N}
\newcommand \R {\mathbb R} 
\newcommand{\scalarpr}[2]{\left<#1,#2\right>}
\newcommand{\pairing}[2]{\left(#1,#2\right)}
\newcommand{\diag}{\text{Diag}\xspace}
\newcommand{\Eqref}[1]{Eq.~\eqref{#1}}
\newcommand{\Eqsref}[1]{Eqs.~\eqref{#1}}
\newcommand{\Sectionref}[1]{Section~\ref{#1}}  
\newcommand{\Sectionsref}[1]{Sections~\ref{#1}}  
\newcommand{\Defref}[1]{Definition~\ref{#1}}
\newcommand{\Lemref}[1]{Lemma~\ref{#1}}
\newcommand{\Propref}[1]{Proposition~\ref{#1}}
\newcommand{\Theoremref}[1]{Theorem~\ref{#1}}
\newcommand{\Corref}[1]{Corollary~\ref{#1}}

\newcommand{\Conditionref}[1]{Condition~\ref{#1}}
\newcommand{\Conditionsref}[1]{Conditions~\ref{#1}}
\newcommand \del 	\partial

\newcommand{\keyword}[1]{\textbf{#1}}

\newcommand{\LPDE}[2]{\ensuremath{\widehat L(#1)[#2]}\xspace}
\newcommand{\LPDEuOp}[1]{\ensuremath{\widehat L(u_0+#1)}\xspace}
\newcommand{\LPDEu}[2]{\LPDE{u_0+#1}{#2}}
\newcommand{\LPDEun}[2]{\LPDE{u_n+#1}{#2}}
\newcommand{\HPDELinOp}{\ensuremath{\mathbb H}\xspace}
\newcommand{\HPDEuOp}[1]{\ensuremath{\mathbb H(u_0+#1)}\xspace}
\newcommand{\HPDEu}[2]{\ensuremath{\mathbb H(u_0+#1)[#2]}\xspace}
\newcommand{\GuOp}{\ensuremath{\mathbb G(u_0)}\xspace}
\newcommand{\Gu}[1]{\ensuremath{\mathbb G(u_0)[#1]}\xspace}
\newcommand{\LODE}[2]{\ensuremath{L_{\text{ODE}}(#1)[#2]}\xspace}
\newcommand{\LODEu}[1]{\LODE{u_0}{#1}}
\newcommand{\HODEu}[1]{\ensuremath{\mathbb H_{\text{ODE}}(u_0)[#1]}\xspace}
\newcommand{\SO}[1]{\ensuremath{S_1(#1)}\xspace}
\newcommand{\SOu}[1]{\SO{u_0+#1}}
\newcommand{\SOL}[1]{\ensuremath{S_{1,0}(#1)}\xspace}
\newcommand{\SOLu}{\SOL{u_0}}
\newcommand{\SOLus}{\ensuremath{S_{1,0}}\xspace}
\newcommand{\SOH}[1]{\ensuremath{S_{1,1}(#1)}\xspace}
\newcommand{\SOHu}[1]{\SOH{u_0+#1}}
\newcommand{\SOHus}[1]{\SOH{#1}}
\newcommand{\SOInv}[1]{\ensuremath{S_{1}^{-1}(#1)}\xspace}
\newcommand{\SOInvu}[1]{\SOInv{u_0+#1}}
\newcommand{\SOLInv}[1]{\ensuremath{S_{1,0}^{-1}(#1)}\xspace}
\newcommand{\SOLInvu}{\SOLInv{u_0}}

\newcommand{\ST}[1]{\ensuremath{S_2(#1)}\xspace}
\newcommand{\STu}[1]{\ST{u_0+#1}}
\newcommand{\STL}[1]{\ensuremath{S_{2,0}(#1)}\xspace}
\newcommand{\STLu}{\STL{u_0}}
\newcommand{\STLus}{\ensuremath{S_{2,0}}\xspace}
\newcommand{\STH}[1]{\ensuremath{S_{2,1}(#1)}\xspace}
\newcommand{\STHu}[1]{\STH{u_0+#1}}
\newcommand{\STHus}[1]{\STH{#1}}
\newcommand{\NN}[1]{\ensuremath{N(#1)}\xspace}
\newcommand{\NNu}[1]{\NN{u_0+#1}}
\newcommand{\NL}[1]{\ensuremath{N_{0}(#1)}\xspace}
\newcommand{\NLu}{\NL{u_0}}
\newcommand{\NLus}{\ensuremath{N_{0}}\xspace}
\newcommand{\NH}[1]{\ensuremath{N_{1}(#1)}\xspace}
\newcommand{\NHu}[1]{\NH{u_0+#1}}
\newcommand{\NHus}[1]{\NH{#1}}
\newcommand{\f}[1]{\ensuremath{f(#1)}\xspace}
\newcommand{\FPDE}[2]{\ensuremath{F(#1)[#2]}\xspace}
\newcommand{\FPDEu}[1]{\FPDE{u_0}{#1}}

\newcommand{\FPDEus}[1]{\ensuremath{F[#1]}\xspace}
\newcommand{\FPDEun}[1]{\FPDE{u_n}{#1}}
\newcommand{\FODE}[2]{\ensuremath{F_{\text{ODE}}(#1)[#2]}\xspace}
\newcommand{\FODEu}[1]{\FODE{u_0}{#1}}
\newcommand{\FODEuOp}{\ensuremath{F_{\text{ODE}}(u_0)}\xspace}
\newcommand{\FL}[2]{\ensuremath{\mathscr{F}(#1)[#2]}\xspace}
\newcommand{\FLuOp}{\ensuremath{\mathscr{F}}(u_0)\xspace}
\newcommand{\FLu}[1]{\ensuremath{\mathscr{F}}(u_0)[#1]\xspace}
\newcommand{\FLus}[1]{\ensuremath{\mathscr{F}}[#1]\xspace}
\newcommand{\FLun}[1]{\ensuremath{\mathscr{F}}(u_n)[#1]\xspace}
\newcommand{\FLunOp}{\ensuremath{\mathscr{F}}(u_n)\xspace}
\newcommand{\Resu}[1]{\ensuremath{\mathrm{Res}[#1]}\xspace}
\newcommand{\RR}[1]{\ensuremath{\mathcal{R}[#1]}\xspace}

\begin{document}

\title{Quasilinear hyperbolic Fuchsian systems 
and AVTD behavior in $T^2$--symmetric vacuum spacetimes}

\author{Ellery Ames\footnote{Department of Physics, University of
    Oregon, Eugene, OR 97403, USA. Email: ellery@uoregon.edu} 
\hskip0.cm, Florian Beyer\footnote{
Department
  of Mathematics and Statistics, University of Otago, P.O. Box 56, Dunedin 9054, New Zealand. Email: fbeyer@maths.otago.ac.nz.}, 
\\
James Isenberg\footnote{Department of Mathematics, University of Oregon, Eugene, OR 97403, USA. Email: isenberg@uoregon.edu.},
and Philippe G. LeFloch\footnote{ Laboratoire Jacques--Louis Lions \& Centre National de la Recherche Scientifique, Universit\'e Pierre et Marie Curie (Paris 6), 4 Place Jussieu, 75252 Paris, France. Email: contact@philippelefloch.org.}
}

\date{February 2013 (final version)}

\maketitle

\begin{abstract}
We set up the singular initial value problem for quasilinear hyperbolic Fuchsian systems of first order and establish an existence and uniqueness theory for this problem with smooth data and smooth coefficients (and with even lower regularity). We apply this theory in order to show the existence of smooth (generally not analytic) $T^2$--symmetric solutions to the vacuum Einstein equations, which exhibit AVTD (asymptotically velocity term dominated) behavior in the neighborhood of their singularities and are polarized or half--polarized. 
\end{abstract}

 \tableofcontents

\vspace{2ex}

\section{Introduction}

Fuchsian formulations have proven to be very useful for studying the behavior of cosmological spacetimes in the neighborhood of their singularities. Introduced into general relativity almost fifteen years ago by Kichenassamy and Rendall \cite{Kichenassamy:1999kg}, these formulations have been used primarily as a tool for showing that within certain families of solutions of the Einstein equations (defined primarily by the invariance of each member of the family under a fixed isometry group), there is a large collection of solutions which exhibit AVTD (asymptotically velocity term dominated) behavior. Roughly speaking, a spacetime shows AVTD behavior if, in a neighborhood of its cosmological singularity, the evolution of the spacetime metric field of the solution approaches the evolution of a model metric field which (relative to some choice of spacetime coordinates)  satisfies a system of ordinary differential equations (ODEs) deduced from the Einstein equations by suppressing spatial derivatives relative to time derivatives. The detection of AVTD behavior has proven to be a very useful step towards verifying that the strong cosmic censorship conjecture holds for certain families of solutions of the Einstein equations \cite{Chrusciel:1999dk,Ringstrom:2009ji}. 

Fuchsian formulations are effective in studying the possible presence of AVTD behavior since they are designed specifically to handle singular systems of partial differential equations (PDEs) or, equivalently, PDE systems in the neighborhood of their singularities. While it is not easy to identify the location of singularities in generic spacetime solutions of Einstein's equations, for certain isometry-defined families of solutions ---e.g. spatially homogeneous solutions, Gowdy solutions, $T^2$--symmetric solutions, and many families of $U(1)$--symmetric solutions--- one can use areal coordinates or special forms of harmonic coordinates to locate the singularities. If the Einstein equations are then reduced relative to these symmetries and expressed in terms of these coordinates, then the resulting PDE system takes a singular form in the neighborhood of the singularity which is amenable to a Fuchsian formulation and analysis in the form of a singular initial value problem ---presuming that various further conditions are met. 

Most of the earlier applications of Fuchsian formulations to families of solutions of Einstein's equations have presumed that the spacetimes are analytic \cite{Kichenassamy:1999kg,Isenberg:1999ba,Isenberg:2002ku, Andersson:2001fa}. This is not surprising, since Fuchsian formulations for generic systems of PDEs were initially developed with analytic PDE systems in mind \cite{Kichenassamy:1996wy,Kichenassamy:2007tr}. It is important, however, to extend studies of AVTD behavior and strong cosmic censorship beyond analytic spacetimes and to consider if they also hold for spacetime solutions which are smooth, but not necessarily analytic. 

There are two sets of results (prior to our work) known to the authors concerning the existence of solutions to quasilinear Fuchsian PDEs in smooth or finite differentiability regularity classes.\footnote{In addition to the results we cite here involving quasilinear Fuchsian PDEs with non-analytic regularity, there are a number of treatments of \emph{linear} Fuchsian PDEs with non-analytic regularity; these include the works of Tahara \cite{Tahara:1986vz,Tahara:1984tg} and Kichenassamy \cite{Kichenassamy:2007tr,Kichenassamy:1996hr}.}
In the first of these, proposed by Claudel and Newman \cite{Claudel:1998tt}, the main result is that if a number of quite restrictive technical conditions are satisfied by the PDE system, then the Cauchy problem is well-posed for data specified at the singular time. As noted in \cite{Rendall:2000ki}, these restrictive conditions are not generally satisfied by the PDE systems corresponding to the Einstein equations for the Gowdy spacetimes, for the $T^2$-symmetric spacetimes, or for most other families of spacetimes under consideration; hence the Claudel and Newman results are not useful for our present purposes. 

The second set of results concerning smooth solutions of Fuchsian systems are those proven by Rendall. 
In \cite{Rendall:2000ki}, he develops a Fuchsian--based approach that is applicable to both semilinear and quasilinear equations,  and he uses it to establish the existence of a class of smooth $T^3$ Gowdy spacetimes which exhibit AVTD behavior. In Rendall's approach,  one performs a series of reduction steps in order to obtain a symmetric hyperbolic system, and one then proves the existence of smooth solutions using a sequence of analytic solutions to a sequence of analytic ``approximate equations.'' Although this method has successfully been applied by Clausen \cite{Clausen:2007vq} 
to the family of polarized $T^2$-symmetric spacetimes, and has also been used by St{\aa}hl \cite{Stahl:2002bv} in studying $S^{3}$ and $S^{2}\times S^{1} $ Gowdy spacetimes, it has proved difficult to apply in more general cases, such as  for spacetimes with only one Killing vector field \cite{Isenberg:2002ku}, or with no symmetries \cite{Andersson:2001fa}.

Our goal in the present work is to develop a general Fuchsian formulation for analyzing smooth (but not analytic) solutions to quasilinear hyperbolic PDEs which can be fairly directly applied to polarized and half--polarized $T^2$--symmetric solutions to the Einstein equations and can be applied  to polarized and half--polarized $U(1)$-symmetric solutions as well. Two of the authors of this work, Beyer and LeFloch in \cite{Beyer:2010fo}, have carried out this program for semilinear hyperbolic systems  and have applied  their formulation to $T^3$ Gowdy solutions. Therein, Beyer and LeFloch set up a second--order Fuchsian formulation for smooth semilinear systems. In the present paper, in addition to generalizing to smooth \emph{quasilinear} hyperbolic systems, we also work with Fuchsian systems in first--order form. 

One of the motivations for the semilinear work and its application to the $T^3$ Gowdy spacetimes was that the  approximation scheme which plays a key role in the existence proof can also be used as the basis for a robust method for numerical simulations. This numerical approach has been developed and implemented in \cite{Beyer:2010tb,Beyer:2011ce} (see also \cite{Amorim:2009ka}) as a tool for the  numerical exploration of Gowdy solutions. Since our analysis in the present paper involves a similar approximation scheme, we expect to be able to carry out  numerical investigations of singular initial value problems in more general classes of equations in future work.

An outline of this paper is as follows. 
We begin the discussion of our results in Subsection \ref{sec:firstordertheory}, where we consider a general class of first--order quasilinear Fuchsian systems and, then, formulate the singular initial value problem for such systems. Next, in the same subsection, we introduce the class of first--order quasilinear Fuchsian systems in \textit{symmetric hyperbolic} form and, for such systems, state an existence and uniqueness result (in Theorem \ref {th:Wellposedness1stOrderFiniteDiff}, below) which holds in, both, infinite differentiability and finite differentiability classes. In Subsection \ref{section:proofs}, we carry out the details of the proof of this result. 

This theorem holds for a broad class of asymptotic data specified on the singularity. Next, in Subsection \ref{sec:expansionsHO}, we discuss special choices for this data, which we call ``ODE leading-order term", and then state and establish an alternative existence and uniqueness result (Theorem \ref{th:Wellposedness1stOrderHigherOrder}) for the singular initial value problem. This alternative theorem is useful in applications, as we show later in Subsections \ref{sec:EPD} and  \ref{sec:optimalexistence}.  Interestingly enough, ODE--leading--order asymptotic data also play a useful role as approximate solutions. 

In the second part of this paper,  in Section \ref{application}, we apply our theoretical results and study polarized and half-polarized $T^2$--symmetric solutions of Einstein's vacuum field equations. To this end, in Subsection \ref{spacetimes}, we define this family of spacetimes, and then, in the same subsection, we write out the Einstein equations in terms of areal coordinates.  Next, in Subsection \ref{sec:defnAVTD}, we discuss the concept of AVTD behavior and discuss what one needs to do in order to check whether a set of polarized or half-polarized $T^2$--symmetric solutions do exhibit AVTD behavior. In Subsections \ref{sec:firstexistence} and \ref{sec:optimalexistence}, we are in a position to rely on our results in Section \ref{theory} and we establish that our conditions therein indeed hold true for the spacetimes under consideration. On one hand, we rely on Theorem \ref {th:Wellposedness1stOrderFiniteDiff} and establish the existence of a parametrized family of polarized and half--polarized $T^2$--symmetric solutions with, both, finite or infinite order of differentiability and with the expected AVTD behavior. On the other hand, we rely on Theorem \ref{th:Wellposedness1stOrderHigherOrder} and show that, provided attention is restricted to smooth ($C^{\infty}$) solutions, then the family of $T^2$--symmetric with AVTD behavior can be extended to include those with a wider (``optimal") range of the ``asymptotic velocities" (labeled by $k$, below, as defined and discussed in Section \ref{sec:avtdexistence}).

In Section \ref{conclusion}, we conclude and discuss the relevance of Fuchsian formulations for numerical simulations, as well as the application of the proposed formulation in order to tackle more general families of spacetimes. 


\section{The singular initial value problem}
\label{theory}

\subsection{Objective of this section} 

For a given (say, first-order) PDE system $\mathcal{P}[\psi]=0$, the (regular) Cauchy problem involves finding a solution $\psi=\psi(t,x)$ to this system such that, at some chosen value $t_0$ of the time, the solution satisfies the initial condition $\psi(t,x_0) = \phi(x)$, for some specified initial data function $\phi=\phi(x)$. If the Cauchy problem is well-posed, then for any appropriate choice of $\phi$, the solution $\psi$ exists for some open interval $I$ containing $t_0$. 

Here, we are interested in the \emph{singular} initial value problem rather than the regular one. That is, rather than seeking solutions to $\mathcal{P}[\psi]=0$ which agree with specified initial data at a chosen time, we seek for solutions which become (in general) singular as $t$ approaches some fixed value $t_{\infty}$, and which agree with some specified fall-off data as one approaches the singularity at $t_{\infty}$. As for the Cauchy problem, one can prove existence and uniqueness theorems for the singular initial value problem; these theorems guarantee that for any appropriate choice of the ``asymptotic data", there is a solution to $\mathcal{P}[\psi]=0$ which exists for $t$ approaching $t_{\infty}$, and which asymptotically matches the prescribed asymptotic data


 \subsection{Quasilinear first-order symmetric hyperbolic
  Fuchsian systems}
\label{sec:firstordertheory}

Before making this notion of singular initial value problem precise, and proving existence and uniqueness results, we carefully define the class of PDEs we shall consider here.  While the theory we develop here can be generalized to a much wider
class of background spacetime manifolds\footnote{We work with one
  spatial dimension, primarily since in addition to simplifying the
  discussion, this case is sufficient for handling our primary
  application: $3+1$ dimensional spacetimes with a spatially acting $T^2$ isometry group.} (see, for example \cite{Ames:2012tm}, in which we generalize our results to spacetime manifolds $(0,\delta] \times T^n$),
let us presume for now
that we work on the cylinder spacetime 
$(0,\delta] \times T^1$, for some small parameter $\delta$, with coordinates $t \in (0,\delta]$ and  $x \in T^1$. The singularity is presumed to occur at $t=0$ (hence we set $t_{\infty} =0$ in the earlier notation); correspondingly, it is useful to work with the ``singular time differential operator"
$$
D:=t\partial_t.
$$

The general form of the first-order PDE systems under consideration reads 
\begin{equation}
  \label{eq:1stordersystem}
  S_1(t,x,u)Du(t,x)+S_2(t,x,u) t \partial_x u(t,x)+N(t,x,u)
  u(t,x)=f(t,x,u),
\end{equation}
in which the unknown is a vector-valued spacetime function $u:(0,\delta] \times T^1 \rightarrow \R^d$
for some integer $d \geq 1$ and some real $\delta>0$. Here $S_1=S_1(t,x,u)$, $S_2=S_2(t,x,u)$, and $N=N(t,x,u)$ are specified
$d\times d$--matrix--valued maps of the spacetime coordinates
$(t,x)$ and the unknown $u$ (but is independent of its derivatives), while
$f=f(t,x,u)$ is a specified $\R^d$--valued map of $(t,x)$ and the unknown $u$ (but again is independent of its derivatives). 
The specific requirements for the functions $S_1$, $S_2$, $N$ and $f$ are fixed precisely below; see in particular \Defref{def:quasilinearlimit}.
For notational convenience, we often leave out the arguments $(t,x)$, instead we use the short--hand notation $\SO{u}$, $\ST{u}$, $\NN{u}$, and $\f{u}$. Notationally, an object such as $\SO{u}$ may be interpreted as a map $u\mapsto \SO{u}$  between two function spaces (as further discussed below). In this context we often write $\SO{u}(t,x)$, and we do the same for $S_2(u)$ and $N(u)$.

Observe that, in principle, one could absorb the term $\NN{u}u$ into the source $\f{u}$;
however in view of the conditions on these terms that we will introduce below, it is important to keep these two terms separate. A system of this form \eqref{eq:1stordersystem} (noting especially the use of the singular operator $D$) will be referred to as 
a \keyword{quasilinear first--order Fuchsian system}.

Before 
formulating the singular initial value problem for such systems, we wish to define functional norms and the corresponding function spaces which include  built-in specifications of the asymptotic behavior of the functions in time. We state these definitions first (here) for vector-valued functions, and then (below) for matrix-valued functions. The definitions are parametric: For the vector-valued function case, we specify as parameters  i) a non-negative integer $q$, and ii) a fixed smooth\footnote{One could choose $\mu$ (and other parameter functions below) to have less regularity, but this is an inconsequential generality for the purpose of this paper.} 
vector-valued function $\mu:T^1\rightarrow \R^d$. Then using $\mu$ (which we label as an \keyword{exponent vector}) to construct the corresponding diagonal matrix  
\begin{equation}
 \label{eq:defR}
 \RR{\mu}(t,x):=\text{diag}\, \bigl(t^{-\mu_1(x)},\ldots,t^{-\mu_d(x)}\bigr), 
\end{equation}
we  define the norm
\begin{equation}
  \label{eq:norm}
  \|w\|_{\delta,\mu,q}:=\sup_{t\in (0,\delta]}
    \|\RR{\mu}(t,\cdot) w(t,\cdot)
    \|_{H^q(T^1)} 
\end{equation}
for vector-valued functions $w=w(t,x)$. Here, $||\cdot||_{H^q(T^1)}$ denotes the standard $q$--order Sobolev norm on $T^1$. Based on \eqref{eq:norm}, we define the Banach space $X_{\delta, \mu,q}(T^1)$ ---also simply written as $X_{\delta, \mu,q}$--- as the completion of the set of all functions $w\in C^\infty((0,\delta]\times T^1)$ for which this norm is finite, and we denote by $B_{\delta,\mu,q,r} \subset X_{\delta,\mu,q}$ the closed ball of radius $r>0$ (measured with the given norm) 
and center $0$. To handle the class of functions which are infinitely differentiable, we define the space 
\begin{equation*}
X_{\delta,\mu,\infty}:=\bigcap_{q=0}^\infty X_{\delta,\mu,q}.
\end{equation*}
In order to compare two function spaces $X_{\delta,\mu,q}$ and $X_{\delta, \nu,q}$, we write $\nu>\mu$ if, for each index $i=1,\ldots,d$ and for all $x\in
T^1$, the components of $\nu$ and $\mu$ satisfy the inequality
$\nu_{i}(x)>\mu_i(x)$. Clearly, $X_{\delta,\nu,q}
\subset X_{\delta,\mu,q}$ if $\nu>\mu$. 

We use analogously defined norms and functions spaces in order to control $d \times d$ matrix-valued functions such as $S_{1}, S_{2}$, and $N$. More specifically, in this case we choose 
as a parameter a  fixed $d\times d$ matrix-valued valued function $\zeta$ (labeled as an \keyword{exponent matrix})  which depends smoothly $x \in T^{1}$, and we define the corresponding norm as 
\begin{equation}
\|S\|_{\delta,\zeta,q} := \sup_{t\in (0,\delta]}\sum_{i,j=1}^d \| t^{-\zeta_{ij}(\cdot)} S_{ij}(t,\cdot)\|_{H^q},
\end{equation}
for matrix-valued functions $S= S(t,x)$. We denote the corresponding Banach space by $X_{\delta,\zeta, q}.$
Based on these function spaces of matrix-valued functions, we define (for $r>0$) $B_{\delta,\zeta,q,r} \subset X_{\delta,\zeta,q}$ and $X_{\delta, \zeta, \infty}$ (as above), and we note that $X_{\delta,\xi,q}
\subset X_{\delta,\zeta,q}$ if $\xi>\zeta.$

As noted above, the singular initial value problem associated with a system such as \Eqref{eq:1stordersystem} consists of choosing a set of ``asymptotic data", and seeking for solutions which asymptotically approach that data. Using the function spaces just defined, we make this idea precise as follows. 

 \begin{definition}
  \label{def:SIVP}
Given the parameters $\delta,\mu$ and $q$ as above, and a chosen function $u_0:(0,\delta]\times T^1\rightarrow\R^d$, the \keyword{singular initial value problem} consists of seeking a solution $u=u_0+w$ to
  \Eqref{eq:1stordersystem} whose \keyword{remainder} $w$ belongs to
 $X_{\delta,\mu,q}(T^1)$. 
\end{definition}

The function $u_0$, which we refer to as the \keyword{leading-order term}, constitutes the asymptotic data, and is (a priori) of unspecified regularity. Regarding the solution function $w$, if it is to be considered a ``remainder", then in comparison with $u_0$ it should be of higher order in $t$ as one approaches $t=0$. 
Observe that the exponent vector $\mu$, which parametrizes the function space  $X_{\delta,\mu,q}(T^1)$, 
controls the order of the singularity of the remainder  $w$ at
$t=0$; roughly speaking, each component of $w$ is of corresponding
component order $O(t^\mu)$ if $w\in X_{\delta,\mu,q}$. Hence, the
components of $\mu$ are sometimes referred to as the remainder
exponents, with $\mu$ collectively labeled the \keyword{(remainder)
exponent vector}. Generally, we assume here and below that exponent vectors are smooth. Also, for a given exponent vector $\mu$ and a given scalar $\epsilon$, we use the notation $\mu + \epsilon$ to indicate a new exponent vector obtained by adding $\epsilon$ to each component of $\mu$.

We now discuss the conditions on $\SO{u}$, $\ST{u}$, $\NN{u}$, and
$\f{u}$ in \eqref{eq:1stordersystem} which, together with further conditions on the space of
leading-order terms and the space of remainder functions, are
sufficient for establishing the well--posedness of the singular initial value
problem. The main set of conditions needed is included in the
following definition.

\begin{definition}
\label{def:quasilinearlimit}
Fix
positive constants $\delta$ and $s$, a pair of non-negative integers $q_0$ and $q$  (possibly $+\infty$), and an exponent vector $\mu:T^1\rightarrow\R^d$, 
together with a leading-order term $u_0 : (0, \delta] \times T^1 \to \mathbb R^d$ (with so far unspecified regularity). The system \Eqref{eq:1stordersystem} is called a \keyword{quasilinear symmetric hyperbolic Fuchsian system} around $u_0$
if, for each $x \in T^1$,  there exist a matrix $\SOLu(x)$
that is positive definite and symmetric and independent of $t$, 
a matrix $\STLu(x)$ that is symmetric and independent of $t$, 
and a matrix $\NLu(x)$ that is independent of $t$, all defining matrix-valued functions in the Sobolev space $H^{q_0}(T^1)$; and if there exists a smooth vector function $\beta : T^1 \to \mathbb R^d$ with strictly positive components, 
such that for every $\delta'\in (0,\delta]$,
each of the``remainder matrices''
\begin{align*}
\SOHu{w}&:=\SOu{w}-\SOLu,\\
    \STHu{w}&:=\RR{\beta-1}\STu{w}-\STLu,\\
    \NHu{w}&:=\NNu{w}-\NLu,
 \end{align*}
considered as an operator of the form (for example) $w\mapsto\SOHu{w}$, maps  all functions $w\in B_{\delta', \mu, q,s}$ to elements in $B_{\delta', \zeta, q,r}$, in which  $\zeta$ is some exponent matrix with strictly positive entries, and $r>0$ is some constant. It is furthermore required that $\SOHu{w}$ and $\STHu{w}$ are symmetric matrices for all $w\in B_{\delta, \mu, q,s}$.
\end{definition}

Before discussing further conditions which are needed in order to obtain existence and uniqueness
for the singular initial value problem, we note the following remarks:
\begin{enumerate}[label=\textit{(\roman{*})}, ref=(\roman{*})]   

\item If a system \Eqref{eq:1stordersystem} satisfies the conditions in the above definition and is thus a quasilinear symmetric hyperbolic Fuchsian system, then the matrices $\SO{u}$, $\ST{u}$, and $\NN{u}$ in \Eqref{eq:1stordersystem} (acting on $u=u_0 +w$) decompose as
 \begin{align*}
    \SOu{w}&= \SOLu +\SOHu{w},\\
    \STu{w}&= \RR{1-\beta}(\STLu + \STHu{w}),\\
    \NNu{w}&=\NLu+\NHu{w}.
  \end{align*}  
Quasilinear symmetric hyperbolic Fuchsian systems are therefore
  ``essentially'' semilinear (described by the coefficients
  $S_{1,0}$, $S_{2,0}$ and $N_{0}$), up to ``quasilinear
  perturbations'' (given by $S_{1,1}$, $S_{2,1}$ and $N_{1}$), which
  decay as $t \rightarrow 0$ with a rate controlled by $\zeta$. The
  purely semilinear case has been treated earlier within a second--order framework 
by Beyer and LeFloch~\cite{Beyer:2010tb,Beyer:2010fo,Beyer:2010wc}. (See also \cite{Beyer:2011uz,Beyer:2011ce}).

\item Presuming that a fixed leading-order term $u_0$ has been chosen, it is convenient to use the short-hand notation $\SOHus{w}$, $\STHus{w}$, $\NHus{w}$ in place of  $\SOHu{w}$, $\STHu{w}$, $\NHu{w}$, respectively, and similarly to use $\SOLus$, $\STLus$, $\NLus$ in place of $\SOLu$, $\STLu$, $\NLu$. It is important, however, to keep in mind the dependence of these matrices on the choice of the leading-order term $u_0$.

\item If \Eqref{eq:1stordersystem} satisfies the conditions in \Defref{def:quasilinearlimit}, then it is symmetric hyperbolic for all $t\in (0,\delta]$, provided $\delta$ is sufficiently small in order to guarantee that $S_1$ is positive definite. Consequently, standard theorems ensure that the Cauchy problem with initial data  specified at $t=t_0\in (0,\delta]$ is well--posed in the usual sense (away from $t = 0$), so long as the order of differentiability (determined by $q, q_0$) is sufficiently large. The solutions belong to the space $C(I, H^{q}(T^{1}))$ for $q \ge 2$ and for some interval $I \subset (0, \delta]$; however, nothing is known a priori regarding the behavior of these solutions as $t$ approaches the singularity at $t=0$.

\item We note that the matrix-valued operators $\SOHus{w}$, $\STHus{w}$, and $\NHus{w}$ are non-singular in a neighborhood of $t=0$ since they take values in $X_{\delta, \zeta, q}$ with $\zeta>0$. Hence quasilinear symmetric hyperbolic Fuchsian systems are singular precisely at $t=0$, where the PDE coefficients are singular. 

\item Presuming  that $q$ is sufficiently large, the conditions in \Defref{def:quasilinearlimit} require that the coefficient matrices on the left-hand side of \Eqref{eq:1stordersystem}, e.g.,\ $S_1(t,x,u)$,  are defined on the domain $(0,\delta]\times T^1\times U$, where $U$ is an open subset of $\R^d$ about the origin. This domain $U$ must be compatible with the choice of parameters $\delta$, $\mu$ and $s$. For example, if $\mu>0$, $q>1$ and $w\in B_{\delta,\mu,q,s}$, then $\|w\|_{L^\infty}\le C s \delta^{\mu_{\text{min}}}$, where $C$ is the Sobolev constant and $\mu_{\min}$ is the minimal value over all components of $\mu$ over all spatial points. Hence, if necessary, $s$ and or $\delta$ must be chosen sufficiently small in order to fit into $U$.
\end{enumerate}

We
discuss a collection of  useful technical tools in the appendix (primarily  in \Sectionref{sec:productsfunctions}) which allow us to check if the conditions of \Defref{def:quasilinearlimit} are satisfied for a given problem.

The remaining conditions we consider concern the coupling between the components of an $\R^d$-valued function $u=u(t,x)$, presumed to satisfy \Eqref{eq:1stordersystem}, and the effects of these couplings on the asymptotic behavior of the components as $t$ approaches the singularity. 

\begin{definition} 
  \label{def:nonessentiallycouple}
Given the singular initial value problem (\Defref{def:SIVP}) for a specified quasilinear symmetric hyperbolic system (\Defref{def:quasilinearlimit})  with specified leading-order term $u_0$ and specified function space $X_{\delta, \mu,q}$, the system \Eqref{eq:1stordersystem} is called \keyword{block
    diagonal with respect to $\mu$}, provided the following commutation conditions hold
  \[
\RR{\mu} \SO{u}=\SO{u} \RR{\mu},\quad \RR{\mu} \ST{u}=\ST{u} \RR{\mu},\quad
  \RR{\mu} \NN{u}=\NN{u} \RR{\mu},
  \] 
  for all $u=u_0+w$ with $w\in X_{\delta,\mu,q}$, and provided 
the same condition holds for all
 relevant spatial derivatives of $\SO{u}, \ST{u}$, and $\NN{u}$. (Recall that  $\RR{\mu}$ is defined in \Eqref{eq:defR}.)
\end{definition}

This block diagonality condition is used in the derivation of energy estimates (in \Sectionref{section:proofs}, below)
and thus plays a major role in the proof of existence and uniqueness. 
Roughly speaking, this condition guarantees  that the ``principal part
operator"\footnote{This operator appears in \Eqref{eq:1stordersystem} in the form $\LPDE{u}{u}$; however, it is useful in the discussion below to define this operator with distinct functions $u, v$.}
\begin{equation}
  \label{eq:DefLPDE}
  \LPDE{u}{v}:= \SO{u}Dv+\ST{u} t \partial_x v+\NN{u}v,
\end{equation}
takes block diagonal form for $v$, and that each block is
associated with only one component of the exponent vector
$\mu$.  Recall that the components of $\mu$ determine the order of the singularity in $t$ at
$t=0$ for the components of the remainder of the singular initial value
problem. Hence, this condition requires that the principal part may only be
coupled within those components of the solution whose remainders behave
the same at $t=0$. Note that the condition does allow all of the
matrices $S_1$, $S_2$ and $N$ to depend on all components of $u$. It also
allows for arbitrary coupling in the source--term, which we write as
\[\FPDEu{w}:=\f{u_0+w},\]
(or, in short form, as $\FPDEus{w}$, whenever it does not lead to confusion).

We are now ready to state our main existence and uniqueness results for the singular initial value problem associated with first-order quasilinear symmetric hyperbolic Fuchsian systems. Our hypotheses below include conditions on the matrix 
\begin{equation}
  \label{eq:energydissipationmatrix}
  M_0:= \SOLus\, \diag(\mu_1,\ldots,\mu_d) + \NLus, 
\end{equation}
which we refer to as the {\bf energy dissipation matrix} and depends on the space coordinate $x$, only. 

\begin{theorem}[Existence theory for symmetric hyperbolic Fuchsian systems]
  \label{th:Wellposedness1stOrderFiniteDiff}
  Suppose that \Eqref{eq:1stordersystem} is a quasilinear symmetric hyperbolic Fuchsian system around a leading-order term $u_0$ (with  the choice of the parameters\footnote{Although  
Definition \ref{def:quasilinearlimit} involves the choice of $\beta, \zeta$ in addition to  $\delta, \mu, q, q_0$, this latter set is crucial to the nature of the results of this theorem, while the former set is in a certain sense incidental. Hence even though $\zeta$ does appear in condition (iii), we do not list $\beta$ and $\zeta$ in the hypothesis here.}
$\delta$, $s$, $\mu$, $q$, and $q_0$ as specified in \Defref{def:quasilinearlimit}) and is block diagonal with respect to $\mu$. Suppose that $q\ge 3$ and $q_0=q+2$. Then there exists a unique solution $u$ to \Eqref{eq:1stordersystem} whose remainder $w:=u-u_0$ belongs to $X_{\widetilde\delta, \mu, q}$ with $Dw \in X_{\widetilde \delta, \mu, q-1}$ for some $\widetilde \delta \in (0, \delta]$, provided the following structural conditions are satisfied:
  \begin{enumerate}[label=\textit{(\roman{*})}, ref=(\roman{*})] 
  \item \label{en:cond2}%
    The energy dissipation matrix $M_0$ defined in 
    \Eqref{eq:energydissipationmatrix} is positive definite
    for every spatial point $x \in T^1$.  
  \item \label{en:cond4N}%
    The map
    \begin{equation}
      \label{eq:defFLu}
      \FLuOp: w\mapsto \FPDEu{w}-\LPDEu{w}{u_0}
    \end{equation}
    is well-defined, and for every $\delta'\in (0,\delta]$, it maps $w\in B_{\delta',\mu,q,s}$ to
    $X_{\delta',\nu,q}$ for some exponent vector $\nu>\mu$.
  \item \label{cond:LipschitzF} For all $\delta'\in(0,\delta]$, for a constant $C>0$ and for all $w, \widetilde w\in B_{\delta',\mu,q,s}$, one has
  \begin{equation}
\|\FLus{w}-\FLus{\widetilde w}\|_{\delta',\nu,q} \le C \|w-\widetilde w\|_{\delta',\mu,q}
\end{equation}
and
    \begin{equation}
      \label{eq:LipschitzF}
      \begin{split}
 &\! \|\FLus{w}-\FLus{\widetilde w}\|_{\delta',\nu,q-1}
        +\|\SOHus{w}-\SOHus{\widetilde w}\|_{\delta',\zeta,q-1}\\
        &+\|\STHus{w}-\STHus{\widetilde w}\|_{\delta',\zeta,q-1}
        +\|\NHus{w}-\NHus{\widetilde w}\|_{\delta',\zeta,q-1}
 \le C \|w-\widetilde w\|_{\delta',\mu,q-1}
      \end{split}
    \end{equation}    
    for all $w, \widetilde w\in B_{\delta',\mu,q,s}$. 

  \end{enumerate}
  If all of these conditions are satisfied for all 
  $q\ge 3$, then there exists a unique solution $u$ of
  \Eqref{eq:1stordersystem} such that $u-u_0$ and $D(u-u_0)$, both, belong to\footnote{This result is short of ``well-posedness'' in the conventional sense because we do not prove continuous dependence of the solution on the asymptotic data. Such a result is certainly desirable, and we plan to investigate it in future work.} $X_{\widetilde \delta,\mu,\infty}$.
\end{theorem}

\Sectionref{section:proofs}, below, is devoted to the proof of this theorem. Observe that, in the hypothesis of this theorem, 
the regularity required for $S_{1,0}$, $S_{2,0}$, and $N_{0}$ (specified by $q_0$) slightly differs from that required for 
$S_{1,1}$, $S_{2,1}$, and $N_{1}$ (specified by $q$). The same observation can be made regarding  the asymptotic data $u_0$ (implicitly specified by \Conditionref{en:cond4N}) and the solution $u$ (specified by $q$). 
These gaps arise in the course of our proof, in particular in obtaining the energy estimates for the Cauchy problem \Lemref{lem:energyestimateshigherorder}. It is not clear whether this discrepancy in regularity could be eliminated by another method of proof, and in any case it disappears in the ``smooth'' case, if  $q$ and $q_0$ are both infinite.

In formulating this theorem, we require that the source term operator $w\mapsto\FLus{w}$, and with it the source term function $f(t,x,u)$ in \Eqref{eq:1stordersystem}, be defined on the domain $(0,\delta]\times T^1\times U$, where $U$ is an open neighborhood of the origin in $\R^d$. In the same way as for the coefficient matrices $S_1$, $S_2$ and $N$, we find that the parameters $\delta$, $\mu$ and $s$ must be compatible with $U$. We note that \Conditionref{en:cond4N} also restricts the regularity of the leading-order term $u_0$.

We also note that the time of existence of the solutions, specified by $\tilde\delta$, could a priori be very small. Indeed, a smaller choice of the parameter $s$ (which may be necessary in order to fit into the domains of the coefficient functions of \Eqref{eq:1stordersystem}) generally leads to a shorter guaranteed time interval of existence.

In its applications, Theorem~\ref{th:Wellposedness1stOrderFiniteDiff}
 often allows one to find an open set of values for the exponent vector $\mu$ for which the singular initial value problem admits unique solutions. A lower bound for this set\footnote{A real $\Lambda$ is defined to be a lower bound for the allowed values of the vector $\mu$ if each component $\nu^a$ of $\nu$ satisfies the condition $\nu^a(t,x) > \Lambda$ for all $x$ in the domain of $\mu$. A similar definition holds for an upper bound for $\mu$.} can originate in
 \Conditionref{en:cond2}, while an upper bound is usually determined by \Conditionref{en:cond4N}. Both bounds on the set of allowed values for $\mu$ provide useful information on the problem. The upper bound for $\mu$ specifies  the smallest regularity space and, hence, the most precise description of the behavior of $w$ (in the limit $t\searrow 0$), while the lower bound for $\mu$ determines the largest space in which the solution $u$ is guaranteed to be unique. Observe  that this \textit{uniqueness} property must be interpreted with care: under the conditions of our theorem, there is exactly one solution $w$ in the space $X_{\tilde\delta,\mu,q}$, although we do not exclude the possibility that another solution may exist in a larger space, for example, in $X_{\tilde\delta,\widetilde\mu,q}$ with $\widetilde\mu<\mu$. Note that if a given system does not satisfy our hypothesis above, there is sometimes a systematic method which allows one to ``improve" a leading--order term $u_0$; cf.~the discussion of (order-n)-leading-order terms in \Sectionref{sec:expansionsHO} and, in particular, \Theoremref{th:Wellposedness1stOrderHigherOrder}.

We also remark that results analogous to those stated in Theorem \ref{th:Wellposedness1stOrderFiniteDiff}
for $D(u-u_0)$ 
 can be derived for
higher--order time derivatives of the solution.

\subsection{Proof of the existence and uniqueness theorem}
\label{section:proofs}

\subsubsection{Outline of the argument}

Before carrying out the details of the proof of the existence and uniqueness Theorem~\ref{th:Wellposedness1stOrderFiniteDiff},
we outline the basic strategy and the basic steps of the proof. We start by working with a linear version of the PDE system. We consider the Cauchy problem for this linear system, verifying that the conditions we have assumed as part of the hypothesis of Theorem~\ref{th:Wellposedness1stOrderFiniteDiff} 
guarantee local existence and uniqueness of solutions for this Cauchy problem, with appropriate levels of regularity. We then use these results pertaining to the Cauchy problem for the linear system and establish that unique solutions to the singular initial value problem for the linear system exist in a neighborhood of the singularity. This is done using the solutions of sequences of Cauchy problems with the initial time for the $j$'th element of this sequence set at $t_j$, and with $t_j$ approaching 
zero, the time of the singularity. To show that the limit of such a sequence of solutions exists, and satisfies the singular initial value problem, we work with the linear PDE system in a weak form, and we also employ a family of energy functionals. To proceed from solutions of the singular initial value problem for the linear system to solutions for the full quasilinear system of Theorem  \ref{th:Wellposedness1stOrderFiniteDiff},
 we use a standard  fixed point iteration argument for a sequence of linearized equations and their singular initial value problems. Observe that arguments similar to those used here have been applied in \cite{Beyer:2010tb} in order to establish existence and uniqueness results for the singular initial value problem for semilinear (second order) Fuchsian PDEs.

\subsubsection{The singular initial value problem for linear PDEs}
\label{sec:lineartheory}

The linear systems we consider here are essentially those of Theorem~\ref{th:Wellposedness1stOrderFiniteDiff} (see \Eqref{eq:1stordersystem})
with $S_1, S_2$, and $N$ set to be independent of $u$, and with $f$ set to be linear in $u$. More specifically, we introduce
 the following definition. 

\begin{definition}
  \label{def:linearity}
  Suppose that $\delta$ and $r$ are positive reals, $q$ and $q_0$ are non-negative integers,  $\mu:T^1\rightarrow\R^d$ is an exponent vector, and $\zeta:T^1\rightarrow\R^{d \times d}$ is an exponent matrix such that $\zeta>0$. The system \Eqref{eq:1stordersystem} is called a \keyword{linear symmetric hyperbolic Fuchsian system} if the following conditions are satisfied:
  \begin{enumerate}[label=\textit{(\roman{*})}, ref=(\roman{*})] 
  \item The operators $S_1, S_2$, and $N$ are independent of $u$, and  they can be decomposed as 
    \begin{align}
      S_1(t,x)&=S_{1,0}(x)+S_{1,1}(t,x),\\
      S_2(t,x)&=\RR{1-\beta(x)}\bigl(S_{2,0}(x)+S_{2,1}(t,x)\bigr),\\
      N(t,x)&=N_{0}(x)+N_{1}(t,x),
    \end{align}
    where $S_{1,0}$ is symmetric and positive definite at every spatial point, where  $S_{1,1}, S_{2,0}$, and $S_{2,1}$ are symmetric, and in addition, the maps $S_{1,0},S_{2,0}$ and $N_{0}$ belong to $H^{q_0}(T^1)$, while $S_{1,1},S_{2,1}$ and $N_1$  are $d\times d$--matrix--valued functions in $B_{\delta,\zeta,q,r}$. Here, $\beta : T^1 \to \mathbb R^d$ is a smooth vector function with strictly positive components.
  \item \label{enum:condition2lindef}
    The constant $\delta$ is sufficiently small so that $S_1$ is uniformly positive definite\footnote{By
      \keyword{uniformly positive definite} we mean that $S_1(t,x)=S_{1,0}(x)+S_{1,1}(t,x)$ is uniformly positive definite (I) \textit{at all} $t\in(0,\delta]$ in the $L^2$-sense (with respect to $x$), and (II) \textit{for all} $S_{1,1}\in B_{\delta,\zeta,q,r}$. Since we assume in addition above that $S_{1,0}$ is positive definite and the perturbation $S_{1,1}$ is bounded, this implies that $S_1$ is positive definite pointwise at all $(t,x)\in (0,\delta]\times T^1$ for all $S_{1,1}\in B_{\delta,\zeta,q,r}$ if $q_0$ and $q$ are sufficiently large (thanks to the Sobolev inequalities).}.
  \item \label{enum:linearsource} The source term is linear in the sense that 
    \begin{equation*}
      \FPDEu{w}=f(t,x,u_0+w) =  f_0(t,x) + F_1(t,x)w,
    \end{equation*}
    with $f_0 \in X_{\delta, \nu, q}$ and the matrix $F_1$ satisfying $\RR{\mu} F_1 \RR{\mu}^{-1} \in B_{\delta, \zeta, q, r}$.
  Here $\nu$ is an exponent vector with $\nu>\mu$.
\end{enumerate}
\end{definition}
In this definition, we note the condition $\nu>\mu$. It is used primarily in the proof of \Propref{prop:linearfirstexistence},  to enforce the needed rapid decay of the source term $f_0(t,x)$ as $t \rightarrow 0$.

It is clear from this definition that a linear symmetric hyperbolic Fuchsian system is a special case of a quasilinear symmetric
hyperbolic system,  with the leading-order term $u_0=0$ (this is no loss of generality for linear systems). Both in the linear and in the quasilinear case, we consider the functions $S_{1,1}$, $S_{2,1}$, $N_1$, and $F_1$ to be perturbations of $S_{1,0}$, $S_{2,0}$, $N_0$, and $f_0$. An important step in our analysis is to seek uniform estimates for these perturbations. It turns out that such estimates can only be obtained if the perturbations are bounded. This is the reason for introducing the balls with radius $r$ above, $B_{\delta,\zeta,q,r}$, which can be considered as those spaces in which we seek the perturbations. 

In carrying out the proof, it is important that we keep careful track of which quantities the constants arising in various estimates are allowed to depend upon. To make this precise, it is useful to have  the following definition.

\begin{definition}
\label{def:uniform}
Suppose that \Eqref{eq:1stordersystem} is a {linear symmetric hyperbolic Fuchsian system} for a chosen set of the parameters 
$\delta, \mu, \zeta, q, q_0$ and $r$. Suppose that a particular estimate (e.g., the energy estimate \Eqref{EnergyEstimate}), involving a collection $\mathcal{C}$ of constants,  holds for solutions of \Eqref{eq:1stordersystem} under a certain collection of hypotheses $\mathcal H$. The constants $\mathcal{C}$ are defined to be \keyword{uniform} with respect to the system and the estimate so long as  the following conditions hold: 
\begin{enumerate}[label=\textit{(\roman{*})}, ref=(\roman{*})] 
\item For any choice of $S_{1,1}, S_{2,1}, N_1$ and $ F_1$ contained in the perturbation space $B_{\delta,\zeta,q,r}$ (see \Defref{def:linearity}) which is compatible with the hypothesis $\mathcal H$, the estimate holds for the \textit{same} set of constants $\mathcal{C}$. %
\item If the estimate holds for a choice of the constants  $\mathcal{C}$ for one particular choice of $\delta$, then for every smaller (positive) choice of $\delta$, the estimate remains true for the same choice of  $\mathcal{C}$.
\end{enumerate}
\end{definition}

Recalling our definition above (see \Eqref{eq:DefLPDE}) of the principal part operator $\widehat L$, we define the linear principal part operator by
\begin{equation}
\label{linprincpartop}
L[w]:=(S_{1,0} + S_{1,1})Dw + \RR{1-\beta}(S_{2,0}+S_{2,1})t \partial_x w +(N_0+N_1)w.
\end{equation}
In terms of this operator, the linearized PDE system \Eqref{eq:1stordersystem} may be written in the form $L[w] =f_0 + F_1 w$.

In summary, the parameters $\delta$, $\mu$ and $q$ determine the space $X_{\delta,\mu,q}$ for the remainder of the solution of the singular initial value problem with leading-order term $u_0$, while $\delta$, $\zeta$, $q$ and $r$ fix the space $B_{\delta,\zeta,q,r}$ of the perturbations of the coefficients. The parameter $q_0$ determines the order of differentiability of the ``leading-order'' coefficient matrices $S_{1,0}$, $S_{2,0}$ and $N_0$.

Suppose that \Eqref{eq:1stordersystem} is a linear symmetric hyperbolic Fuchsian system (for parameters $\delta$, $r$, $q$, $q_0$, $\zeta$,  $\mu$; cf.~\Defref{def:linearity}). We first consider the \textit{Cauchy problem}; that is, we prescribe initial data $v_{[t_0]}$ specified at some $t_0 \in (0,\delta)$ and we seek solutions on $[t_0,\delta]\times T^1$ which agree with $v_{[t_0]}$ at $t=t_0$. It is useful at this stage for us to make the
temporary assumption that $S_{1,1}$, $S_{2,1}$, $N_1$ and $F_1$ are $C^\infty((0,\delta]\times T^1)$ functions contained in their respective function spaces, as discussed in \Defref{def:linearity}. (If $S_{1,1}$, $S_{2,1}$, $N_1$ and $F_1$ satisfy this smoothness assumption, then \Eqref{eq:1stordersystem} is said to have \keyword{smooth coefficients}\footnote{Since the assumption of smooth coefficients is not presumed to hold everywhere below, in all cases that is does hold, we state that explicitly.}.)
 Given such a linear symmetric hyperbolic system \Eqref{eq:1stordersystem} with smooth coefficients and if in addition $q_0\ge 2$ and also $f_0\in X_{\delta,\nu,q}$ is smooth, then it is a standard result (see, e.g., Chapter~16 in \cite{Taylor:2011wn}) that the Cauchy problem is well-posed in the sense that for initial data $v_{[t_0]} \in H^{q_0}(T^1)$, there is a unique solution $v: [t_0, \delta]\times T^1\rightarrow \R^d$ to this Cauchy problem with $v(t_0)=v_{[t_0]}$ and with $v(t,\cdot) \in H^{q_0}(T^1)$ for all $t \in [t_0,\delta]$. It is crucial for the following discussion that indeed this solution exists for the \textit{full} interval $[t_0,\delta]$, regardless of the choice of $t_0\in (0,\delta)$. This is true as a consequence of the positivity of $S_1$ on $(0,\delta]$;  cf., \Conditionref{enum:condition2lindef} in \Defref{def:linearity}.

In fact, this statement about the Cauchy problem remains true if the matrices $S_{1,1}$, $S_{2,1}$, $N_1$, $f_0$ and $F_1$ are not required to be smooth, but are only required to have $q_0$ spatial derivatives. Such a relaxation is, however, not useful for our arguments;  we use a more general continuation argument below by which we recover the non-smooth case.  We also note that although this assumption that \Eqref{eq:1stordersystem} has smooth coefficients implies that $S_{1,1}$, $S_{2,1}$, $N_1$, $f_0$, $F_1$ are differentiable to all orders, it does not \textit{not} guarantee that all derivatives have controlled asymptotic behavior for $t\searrow 0$. This control holds only for a set of derivatives given by $q$, as labeled by the relevant function space.

To establish control over the solutions to the Cauchy problem for the linear version of the PDE (\ref{eq:1stordersystem}) and the regularity of these solutions, we now introduce a two-parameter family of explicitly time-dependent energies: Presuming that the remainder exponent  vector $\mu$ is fixed, for any pair of positive constants $\kappa$ and $\gamma$ we define the energy $E_{\mu, \kappa, \gamma}$ for a function $w:[t_0,\delta] \times T^1 \rightarrow \R^d$ (with $w(t,\cdot) \in L^2(T^1)$ for each $t \in[t_0, \delta]$) as follows: 
\begin{equation}
  \label{energies}
  E_{\mu,\kappa,\gamma}[w](t):=\frac 12 e^{-\kappa t^\gamma}
  \scalarpr{S_1(t,\cdot) \RR{\mu}(t,\cdot) w(t,\cdot)}{\RR{\mu}(t,\cdot)
    w(t,\cdot)}_{L^2(T^1)},
\end{equation}
where $S_1$ is the matrix appearing in \Eqref{eq:1stordersystem}. We emphasize again that, unlike standard definitions of energy, the energy functionals $E_{\mu,\kappa,\gamma}[w](t)$ defined here depend on time explicitly, and not just through the time dependence of $w(t,x)$. Note that it readily follows from this definition, and from the conditions assumed to hold for $S_1$ in \Defref{def:linearity}, that there exist uniform (in the sense of \Defref{def:uniform}) constants $C_1$ and $C_2$ such that for any $L^2(T^1)$ function $w(t,x)$, one has (for all $t$) 
\begin{equation}
  \label{normequiv}
  C_1 \|\RR{\mu}(t,\cdot) w(t,\cdot)\|_{L^2(T^1)}
  \le\sqrt{E_{\mu,\kappa,\gamma}[w](t)} \le C_2
  \|\RR{\mu}(t,\cdot)w(t,\cdot)\|_{L^2(T^1)}. 
\end{equation} 

These energies, as is usually the case, have been defined in such a way (including the presence of the factor $e^{-\kappa t^\gamma}$) that for solutions of the Cauchy problem for \Eqref{eq:1stordersystem}, which we label $v(t,x)$, the growth of the energies is controlled.
Explicitly, we obtain the following estimate.

\begin{lemma}[Basic energy estimates for the Cauchy initial value problem]
  \label{lem:energyestimates}
Suppose that for some choice of the parameters $\delta, \mu, \zeta, q, q_0$, and $r$, with $q=0$ and $q_0=2$, and for $u_0=0$, \Eqref{eq:1stordersystem}  is a linear symmetric hyperbolic  Fuchsian system with smooth coefficients and with $f_0$ both smooth and contained in $X_{\delta,\nu,q}$ for some $\nu>\mu$. Suppose also that the system is block diagonal with respect to $\mu$, 
that the energy dissipation matrix \Eqref{eq:energydissipationmatrix} is positive definite for all $x\in T^1$ and, in addition, that $DS_{1,1}$ and $\partial_x S_{2,1}$ are contained in $B_{\delta,\xi,0,s}$ for some constant $s>0$ and some exponent matrix $\xi$ with strictly positive entries. Then there exist positive constants $\kappa$, $\gamma$, and $C$ such that for any initial data $v_{[t_0]} \in H^2(T^1)$ specified at some $t_0 \in
  (0,\delta]$, the solution of the Cauchy problem $v$ for this system and this initial data satisfies the energy estimate
  \begin{equation}
    \label{EnergyEstimate}
    \sqrt{E_{\mu,\kappa,\gamma}[v](t)}\le \sqrt{E_{\mu,\kappa,\gamma}[v](t)}|_{t=t_0}
    +C\int_{t_0}^t s^{-1}\|\RR{\mu}(s,\cdot) f_0(s,\cdot)\|_{L^2(T^1)} \, ds 
 \end{equation}
 for all $t\in [t_0,\delta]$. The constants $C$, $\kappa$, and $\gamma$ may be chosen to be uniform\footnote{While the constants $C$, $\kappa$ and $\gamma$ here can be chosen to be uniform in the sense of \Defref{def:uniform}, there generally does {not} exist a choice which holds for all $\delta, S_{1,0}, S_{2,0}, N_0, \beta, r, s, \zeta, \xi, \mu$ and $\nu$.} and do not depend on $f_0$. In particular, if one replaces $v_{[t_0]}$ specified at $t_0$ by any  $v_{[t_1]}$ specified at any time $t_1 \in (0, t_0]$, then the energy estimate holds for the \emph{same} constants $C$, $\kappa$, $\gamma$.
\end{lemma}
Before proving this lemma, we make a few remarks: 
I) Lemma~\ref{lem:energyestimates} does \textit{not} imply that the energy estimate
\Eqref{EnergyEstimate} holds for\ $t<t_0$; in particular, it need not hold for $t\searrow 0$. II) The well-posedness of the Cauchy problem which is used implicitly in the proof of  \Lemref{lem:energyestimates} requires sufficiently high regularity on the coefficients (see for example \cite{Taylor:2011wn}); this gives rise to the condition $q_0=2$ stated in the hypotheses. III) We remind the reader that the condition that the  coefficients be smooth does \textit{not} imply either the $q=0$ condition or the conditions that $DS_{1,1},\partial_x S_{2,1} \in B_{\delta,\xi,0,s}$. While the smoothness condition implies the \textit{existence} of all derivatives, the latter are statements about the \textit{behavior} of the lowest derivatives in the limit $t\searrow 0$. 
It may appear that since Lemma~\ref{lem:energyestimates} focuses on the Cauchy problem at times $t_0>0$ only, control of behavior near $t=0$ is not necessary. However, such control is in fact needed to obtain an energy estimate with constants which are uniform and independent of $t_0$. IV) In  view of the norm equivalence (\ref{normequiv}) stated above, the estimate
\Eqref{EnergyEstimate} can be rewritten as
\begin{equation}
  \label{eq:EnergyEstimate2}
  \begin{split}
    \|\RR{\mu}(t,\cdot) v(t,\cdot)\|_{L^2(T^1)}  \le
    \widetilde{C}\Bigl(&\|\RR{\mu}(t_0,\cdot) v_{t_0}\|_{L^2(T^1)} \\
&+ \int_{t_0}^t s^{-1}\|\RR{\mu}(s,\cdot)
      f_0(s,\cdot)\|_{L^2(T^1)}ds\Bigr).
  \end{split}
\end{equation}
We observe that for this version of the energy estimates, all of the constants
$C$, $\kappa$ and $\gamma$ are absorbed into the constant $\widetilde C$;
every change of the former constants is therefore reflected in a
corresponding change of the latter one. V) For some of the following discussion it is important to note that the particular values of the parameters of the perturbations space $\zeta$ and $r$ (and also $\xi$ and $s$) do not play an essential role in this lemma: if we change from one perturbation space $B_{\delta,\zeta,q,r}$ to another one $B_{\delta,\widetilde\zeta,q,\widetilde r}$ (and $\widetilde\xi$ and $\widetilde s$), the same result is obtained with possibly different, but still uniform, constants $\widetilde C$, $\widetilde\gamma$, $\widetilde\kappa$. This is true for all of the following results.

\begin{proof}
 The basic idea of the proof  is to compute $DE[v](t)$,  then bound the terms on the right hand side  and finally integrate the equation in time. For simplicity we write $E[v]$ in place of $E_{\mu, \kappa, \gamma}[v]$. 
 Computing\footnote{In calculating this time derivative, we use the fact that the solution $v$ is $C^1$ in both time and space.} $DE[v]$, and using the symmetry of the matrix  $S_1$, we obtain 
\begin{align*}
\label{DECalc}
DE[v](t) =& 
	-\kappa \gamma t^\gamma E[v](t) 
 	+\frac{1}{2} e^{-\kappa t^\gamma} \int_{T^1}  \left< (DS_1) \RR{\mu} v,\RR{\mu}v \right> dx\\
& 	+ e^{-\kappa t^\gamma} \int_{T^1}  \left< S_1 (D\RR{\mu}) v, \RR{\mu}v \right>dx
    	+ e^{-\kappa t^\gamma} \int_{T^1}  \left< S_1 \RR{\mu} \,Dv, \RR{\mu}v \right> dx.
\end{align*}
We first analyze the fourth term on the right hand side of this
expression, which we label $I$. Using the fact that $v$ is a solution
of equation \Eqref{eq:1stordersystem}, using the block diagonal
condition (\Defref{def:nonessentiallycouple}), and integrating by
parts, we calculate
 \begin{align*}
I = e^{-\kappa t^\gamma} \int_{T^1} \Big(& 
   \left< \RR{\mu}f_0, \RR{\mu}v \right> 
+  \left< \RR{\mu}F_1v, \RR{\mu}v \right> 
+ 1/2 t \left< \left(\partial_x S_2 \right) \RR{\mu}v, \RR{\mu}v \right>  \\
&+ t \left< S_2 (\partial_x   \RR{\mu}) v, \RR{\mu}v \right> 
-  \left< N \RR{\mu}v , \RR{\mu}v \right> \Big) dx.
\end{align*}
Using the H\"{o}lder inequality, we may then estimate the first term in this expression as follows: 
\begin{align*}
 e^{-\kappa t^\gamma} \int_{T^1}  \left< \RR{\mu} f_0, \RR{\mu}v \right>dx \le    e^{-\kappa t^\gamma} || \RR{\mu}f_0 ||_{L^2} || \RR{\mu}v ||_{L^2}.
\end{align*}

We now argue that for appropriate choices of $\kappa$
and  $\gamma$, all the other terms besides this one can be
neglected in a certain sense. First, we use the properties of the linear symmetric
hyperbolic Fuchsian system to expand the coefficient matrices $S_1$, $S_2$, $N$ into terms which are $O(1)$ at $t\rightarrow 0$,  and terms which decay as a power of $t$.  We
thereby obtain
\begin{align*}
DE[v] \le &\,
-  e^{-\kappa t^\gamma} \int_{T^1}
\left<\left(N_0 - S_{1,0} (D \RR{\mu}) \RR{\mu}^{-1}\right)\RR{\mu} v, \RR{\mu} v\right> dx \\
&+ e^{-\kappa t^\gamma} \int_{T^1}  
\Bigl<\Bigl(
- \frac 12\kappa \gamma t^\gamma S_1
+ S_{1,1} D \RR{\mu} \RR{\mu}^{-1}
+\frac 12 DS_{1,1}
- N_1 
+ \RR{\mu} F_1 \RR{\mu}^{-1}\\
&\qquad 	\qquad\quad				
+ \frac{1}{2} t\partial_x S_{2}				
+ t S_2 (\partial_x \RR{\mu}) \RR{\mu}^{-1}
\Bigr)\RR{\mu}v,\RR{\mu}v\Bigr> dx\\
&+ e^{-\kappa t^\gamma} || \RR{\mu}f_0 ||_{L^2} || \RR{\mu}v ||_{L^2},
 \end{align*}
where we use the expansion for $S_2$ to write
\begin{equation*}
t \partial_x S_2 =  \partial_x \RR{-\beta(x)}  \left( S_{2,0} + S_{2,1} \right) 
 + \RR{-\beta(x)}  \left( \partial_x S_{2,0} + \partial_x S_{2,1} \right).
\end{equation*}

The first integral on the right hand side of this inequality  for $DE[v]$ is negative definite if the energy dissipation matrix $M_0=N_0-S_{1,0} (D \RR{\mu}) \RR{\mu}^{-1}$ (see \Eqref{eq:energydissipationmatrix})
is positive definite, and so can be neglected. All of the 
terms in the second integral  on the right hand side of this inequality decay as some
positive power of $t$. We also know that as a consequence of \Defref{def:linearity}, the matrix
$S_1$ is positive definite uniformly. 
It is at this point that we use the factor of $e^{-\kappa t^\gamma}$
which appears in the definition of the energy functionals. The scheme
is to choose $\kappa$ and $\gamma$
in such a way that the second integral in the estimate above is
negative definite. This can be achieved if we choose $\gamma$ small enough and $\kappa$ large enough so that the negative definite $S_1$-term dominates all of the other terms in the second integral on $(0,\delta]$. 
To see that the constants $\kappa$ and $\gamma$ may be chosen so that
they are independent of the functions $S_{1,1}, S_{2,1}, N_1$ and
$F_1$ and are therefore uniform in the sense of \Defref{def:uniform}, we recall that by assumption (see
Definition \ref{def:linearity}), $S_{1,1}, S_{2,1}$ and $N_1$, and $\RR{\mu}F_1\RR{\mu}^{-1}$ are
contained in the ball $B_{\delta, \zeta,q,r}$. Hence these functions all must have
finite norms bounded by $r$. Since the role played by $S_{1,1},
S_{2,1}, N_1$ and $F_1$ in determining the constants $\kappa$ and $\gamma$
depends strictly on the norms of these functions, we may choose a fixed
set of the constants such that the inequality holds for any $S_{1,1},
S_{2,1}, N_1$ and $F_1$ contained in these balls.
In total, we obtain
 \begin{align*}
DE[v](t) 
 \le &  e^{-\kappa t^\gamma} || \RR{\mu}f_0 ||_{L^2} || \RR{\mu}v ||_{L^2},
 \end{align*}
  which implies that 
  \begin{align*}
\partial_tE[v](t) 
 \le &  t^{-1} e^{-\kappa t^\gamma} || \RR{\mu}f_0 ||_{L^2} || \RR{\mu}v ||_{L^2}.
 \end{align*}
 Then using the norm equivalence \Eqref{normequiv}, we may rewrite this as 
  \begin{align}
  \label{E[v]ineq}
\partial_tE[v](t) 
 \le &  C t^{-1} e^{-\kappa t^\gamma} || \RR{\mu}f_0 ||_{L^2} \sqrt{E[v](t)}.
 \end{align}

To integrate this inequality, it would be useful to divide both sides by $\sqrt{E[v](t)}$. However, since the $L^2$ norm of $v$ may vanish in special cases, we  use the following strategy. We set $E_\epsilon := E + \epsilon$ for some constant $\epsilon > 0$ (see, for instance, \cite[Page 59]{Ringstrom:2009cj}), and we check that (\ref{E[v]ineq}) holds if we replace $E$ by $E_{\epsilon}$. Then dividing, and using  $\frac{1}{\sqrt{E_\epsilon}} \partial_t E_\epsilon = 2 \partial_t \sqrt{E_\epsilon}$, we obtain 
\begin{align*}
 \partial_t \sqrt{E_\epsilon[v](t) }
 \le &  C t^{-1} e^{-\kappa t^\gamma} || \RR{\mu}f_0 ||_{L^2},
 \end{align*}
 after rescaling the constant $C$. We now integrate both sides
over $\int_{t_0}^t ds$, thereby obtaining
 \begin{align*}
\sqrt{E_\epsilon[v](t) }
 & \le \sqrt{E_\epsilon[v](t_0) } + {C} \int_{t_0}^t s^{-1} e^{-\kappa s^\gamma} || \RR{\mu}f_0 ||_{L^2}(s) ds
 \\
 & \le \sqrt{E_\epsilon[v](t_0) } + {C} \Big( \sup_{s \in (t_0,t)} e^{-\kappa s^\gamma} \Big)  \int_{t_0}^t s^{-1}  || \RR{\mu}f_0 ||_{L^2}(s) ds \\
  & \le \sqrt{E_\epsilon[v](t_0) } + C  \int_{t_0}^t s^{-1}  || \RR{\mu}f_0 ||_{L^2}(s) ds,
 \end{align*}
where we note that the constant $C$ changes from the second to the third line of this calculation. Taking the limit $\epsilon \to 0$ finishes the proof that the inequality (\ref{EnergyEstimate}) holds. 
It also  follows directly that  the constant $C$ is uniform.
\end{proof}

In order to derive the solution of the singular initial value problem
from a sequence of solutions of the Cauchy problem, we  need
estimates involving higher order spatial derivatives. We obtain these as
follows.

\begin{lemma}[Higher-order energy estimates for the Cauchy initial value problem]
  \label{lem:energyestimateshigherorder}
Suppose that a linear symmetric hyperbolic Fuchsian system has been
chosen which satisfies all of the conditions of the energy estimate
Lemma \ref{lem:energyestimates}, except that\footnote{We need the conditions $\partial_x S_{2,1} \in B_{\delta,\xi,0,s}$ and $DS_{1,1}\in B_{\delta,\xi, 0,s}$ to both be implicitly included in the hypothesis for this lemma. The first of these follows automatically from the following assumptions on $q$ and $q_0$. The second does not, but it is included in the hypothesis for Lemma \ref{lem:energyestimates}, and therefore is included in the hypothesis for this lemma. } (rather than $q=0$ and $q_0=2$) $q$ is an arbitrary integer greater than one, and $q_0=q+2$.
 Then there exists a pair of positive constants $C$ and $\rho$ such that for every sufficiently small $\epsilon>0$, the solution $v$ of the Cauchy initial value problem with initial data $v_{[t_0]} \in H^{q_0}(T^1)$ specified at $t_0$ satisfies (for all $t\in [t_0, {\delta}]$) 
 \begin{equation}
    \label{eq:EnergyEstimateHigherOrder}
    \begin{split}
      \|&\RR{\mu-\epsilon}(t,\cdot) v (t,\cdot)\|_{H^q (T^1)}
     \le
      C\,  \|\RR{\mu}(t_0,\cdot) v_{t_0}\|_{H^q (T^1)}\\
      &+C \, \int_{t_0}^t
      s^{-1}\big(\|\RR{\mu}(s,\cdot) f_0(s,\cdot)\|_{H^q (T^1)}
      +s^\rho\|\RR{\mu}(s,\cdot) v(s,\cdot)\|_{H^{q-1}(T^1)}    \big) \, ds.
  \end{split}
\end{equation}
The constants $C$ (which in general differs from $\widetilde C$ in \Eqref{eq:EnergyEstimate2}) and $\rho$ are uniform\footnote{The constant $C$ generally only depends on $\delta$, $S_{1,0}$, $S_{2,0}$, $N_0$, $\beta$, $r$, $\zeta$, $\mu$, $\epsilon$ and $q$, and the constant $\rho$ depends only on $\epsilon$ and $q$.} in the sense of \Defref{def:uniform} and do not depend on $f_0$. If we replace $v_{[t_0]}$ specified at $t_0$ by any $v_{[t_1]}$ specified at any $t_1 \in (0,t_0]$, then the same estimate holds, for the same constants $C$ and $\rho$.
\end{lemma}

Observe that it is necessary (as stated in the hypothesis of this lemma) that 
 the solution (and hence the data and coefficients) be contained in
$H^{q+2}$ if we wish to obtain
 an energy estimate for $q$ spatial
derivatives. The reason for this requirement is made clear in the course of  the proof.
The main difference between the hypotheses of
\Lemref{lem:energyestimates} and \ref{lem:energyestimateshigherorder}
is that we require stronger control of the
behavior of spatial derivatives of $S_1$, $S_2$, $N$,
$f_0$, and $F_1$ in the limit $t\searrow 0$ in
\Lemref{lem:energyestimateshigherorder} (i.e.,\ $q\ge 1$ as opposed to
$q=0$ in \Lemref{lem:energyestimates}), while we 
presume smoothness for $S_1$, $S_2$, $N$,
$f_0$, and $F_1$ in both lemmas. 

\begin{proof}
This lemma is proven by taking $q$ spatial derivatives of \Eqref{eq:1stordersystem}, reorganizing the resulting equations into a linear symmetric hyperbolic Fuchsian system for the $q$'th order derivative of $v$, applying Lemma \ref{lem:energyestimates} to that system, and then carrying out a number of estimates needed to derive \Eqref{eq:EnergyEstimateHigherOrder} from  the inequality resulting from this application. We discuss some of the details for the $q=1$ case here; the  $q >1$ cases are similar.

Presuming that $v(t,x)$ is the solution to the Cauchy problem for the linear system (before differentiating) with  initial data $v_{[t_0]}$ (contained in $H^{q_0}(T^1))$ specified at $t_0$, we carry out the differentiation and obtain the following PDE system for $\partial_x v$:
\begin{equation}
\label{eq:equationforderivative}
S_1 D \partial_x v+S_2 t\partial_x (\partial_x v)+N \partial_x v
=\widehat f_0+\widehat F_1 \partial_x v,
\end{equation}
where 
\begin{equation}
  \label{eq:defhatf0}
  \widehat f_0:=\partial_x f_0+(\partial_x F_1-\partial_x N)v
  +\partial_x S_1 S_1^{-1}
  (N v-f_0-F_1 v),
\end{equation}
and
\begin{equation*}
\widehat F_1:=F_1-t\partial_x S_2+t \partial_x S_1 S_1^{-1} S_2.
\end{equation*}
Here, we interpret $v$ as a given function so that $\widehat f_0$ can
be considered as a source term function, and $\partial_x v$ as the unknown.
This PDE system is clearly of the desired form \Eqref{eq:1stordersystem} (with $\partial_x v \in H^{q_0-1}(T^1)$ for each value of $t$). However, it is not a linear
symmetric hyperbolic Fuchsian system with respect to the same exponent vector $\mu$: the term $\partial_x N v$ in $\widehat f_0$ violates \Conditionref{enum:linearsource} of \Defref{def:linearity} since it is in $X_{\delta,\mu,q-1}$ rather than in $X_{\delta,\nu,q-1}$ for some $\nu>\mu$. However, \Eqref{eq:equationforderivative} \textit{is} a linear symmetric hyperbolic Fuchsian system if we choose $\widehat\mu:=\mu-\epsilon/2$ as the remainder exponent vector for any scalar constant $\epsilon>0$. One verifies that
\Eqref{eq:equationforderivative} has  block diagonal form with respect to $\widehat\mu$
and also that the energy dissipation matrix is positive definite if $\epsilon$ is sufficiently small.
Consequently, this system \Eqref{eq:equationforderivative} 
satisfies the hypothesis of \Lemref{lem:energyestimates}. It follows that there
exist uniform (in the sense above) constants $\widehat C$, $\widehat\kappa$ and
$\widehat\gamma$ (which generally differ from the ones for the original
equation) such that 
$\partial_x v$ satisfies the energy estimate (for all $t\in [t_0,\delta]$) 
\begin{equation}
 \label{eq:EnergyEstimatestep1}
  \sqrt{E_{\widehat\mu,\widehat\kappa,\widehat\gamma}[\partial_x v](t)}
  \le \sqrt{E_{\widehat\mu,\widehat\kappa,\widehat\gamma}[\partial_x v]} |_{t_0}
  +\widehat C\int_{t_0}^t
  s^{-1}\|\RR{\widehat\mu}(s,\cdot)
  \widehat f_0(s,\cdot)\|_{L^2(T^1)}ds. 
\end{equation}

To derive the $q=1$ version of the estimate (\ref{eq:EnergyEstimateHigherOrder}) from the energy estimate (\ref{eq:EnergyEstimatestep1}), we first note two useful inequalities. Letting $f: (0,\delta] \times T^1 \rightarrow \R^d$ denote any function for which the following norms are finite, we find\footnote{There are versions of this inequality, as well as the one below, which hold for higher order spatial derivatives of $f$ as well.} that, for all $t\in (0,\delta]$, 
\begin{equation}
  \label{eq:estimateswapderivatives1}
  \|\RR{\mu-\epsilon} f(t)\|_{H^1(T^1)}\le
  C\sum_{\xi=0}^1\|\RR{\mu-\epsilon/2}\partial_x^\xi f(t)\|_{L^2(T^1)}. 
\end{equation} 
 Here, the constant $C>0$ may depend on $\mu$ and
$\epsilon$, but, in particular, is independent of $t$.
Observe that the use of $\mu-\epsilon$ on the left hand side of \Eqref{eq:estimateswapderivatives1} and of $\mu-\frac{\epsilon}{2}$ on the right hand side, is needed to dominate the terms
 of the form $\log t$ which are picked up on the left hand side when
$R[\mu]$ is differentiated in space if $\mu$ is not constant (as a result of the $H^1(T^1)$ norm).
We also readily check that 
\begin{equation}
  \label{eq:estimateswapderivatives2}
  \|\RR{\mu-\epsilon/2}\partial_x f(t)\|_{L^2(T^1)}\le 
  C \|\RR{\mu} f(t)\|_{H^1(T^1)}.
\end{equation}

We now work on inequality \eqref{eq:EnergyEstimatestep1}: Observe first that inequality \eqref{EnergyEstimate} from Lemma \ref{lem:energyestimates} holds for $\widehat\mu=\mu-\frac{\epsilon}{2}$ so long as $\epsilon$ is sufficiently small; hence we may add the left hand side of \Eqref{EnergyEstimate} (with $\mu-\frac{\epsilon}{2}$) to that of \Eqref{eq:EnergyEstimatestep1}, and the right hand side of \Eqref{EnergyEstimate} (again with $\mu-\frac{\epsilon}{2}$) to that of \Eqref{eq:EnergyEstimatestep1}. If we now use (i) the norm equivalence \Eqref{normequiv} on both sides to replace energy terms by terms involving norms of $\RR{\cdot}$, (ii) the definition of the Sobolev norm $\| \cdot \|_{H^1(T^1)}$ to combine terms on each side, and (iii) the inequalities \eqref{eq:estimateswapderivatives1} and \eqref{eq:estimateswapderivatives2}, then we obtain the following inequality:
\begin{align*}
  \|\RR{\mu-\epsilon} v\|_{H^1(T^1)}(t)
  \le C\Biggl(\|\RR{\mu} v\|_{H^1(T^1)}(t_0)
  +&\int_{t_0}^t s^{-1}\Bigl(\|\RR{\mu-\epsilon/2}(s,\cdot) \widehat f_0(s,\cdot)\|_{L^2(T^1)}\\
    &+\|\RR{\mu-\epsilon/2}(s,\cdot) f_0(s,\cdot)\|_{L^2(T^1)}\Bigr) ds\Biggr).
\end{align*}
 It remains to substitute in the definition of $\widehat f_0$ from
 \Eqref{eq:defhatf0}. Noting the properties of the functions on the
 right hand side of \Eqref{eq:defhatf0}, we determine that there
 exists a uniform constant $C$ (in the sense above) such that, 
for all $t\in (0,\delta]$, 
\begin{align*}
  \|\RR{\mu-\epsilon/2}\widehat f_0(t)\|_{L^2(T^1)}
  &\le C\left(\|\RR{\mu-\epsilon/2}\partial_x f_0(t)\|_{L^2(T^1)}
    +\|\RR{\mu-\epsilon/2} v(t)\|_{L^2(T^1)}\right)\\
  &\le C\left(\|\RR{\mu} f_0(t)\|_{H^1(T^1)}
  +s^{\epsilon/2}\|\RR{\mu} v(t)\|_{L^2(T^1)}\right). 
\end{align*}
Here in the second step, the constant $C$ has been inconsequentially changed. Combining these last two inequalities, we obtain the desired result
\Eqref{eq:EnergyEstimateHigherOrder} with $q=1$ by setting $\rho=\epsilon/2$.

This concludes the proof that the inequality (\ref{eq:EnergyEstimateHigherOrder}) holds for the case $q=1$. The proof for $q > 1$ proceeds very similarly. The argument that the constants $C$ and $\widehat \delta$ may be chosen so that the inequality holds for all $S_{1,1}, S_{2,1}, N_1$, and $F_1$ contained in $B_{\delta,\zeta,q,r}$ is essentially the same as that used in proving Lemma \ref{lem:energyestimates}. 
\end{proof}

We remark that while the introduction of $\epsilon$ into the estimate \Eqref{eq:EnergyEstimateHigherOrder} is certainly needed,  one \emph{can} choose this $\epsilon$ to be arbitrarily small.  
One might worry that as one proceeds from $q=1$ to higher values, the necessary value of $\epsilon$ grows and causes trouble. However, since the incremental value needed for each step is arbitrarily small, one sees that the total value of $\epsilon$ needed for arbitrary differentiability values can be kept small (below any chosen positive value).

With these results for the Cauchy problem for linear symmetric
hyperbolic Fuchsian systems established, we now set out to use
solutions of the Cauchy problem to establish the existence of
solutions to the \emph{singular} initial value problem. We do this via
an approximation scheme which works as follows: We first choose a
monotonically decreasing sequence of times $t_n \in (0,\delta]$ which
converges to zero. Then for each $n$, we construct a function $v_n:
(0,\delta] \times T^1 \rightarrow \R^n$ which vanishes on $(0,t_n] $,
and which is equal on $(t_n, \delta]$ to the solution of the Cauchy
problem with zero initial data at $t_n$. One readily checks that for
every choice of $\mu$, one has $v_n \in C^0((0,\delta]\times T^1)\cap
X_{\delta,\mu,0}$. The central result of this section is that if
certain hypotheses hold, then the sequence $(v_n)$ -- whose elements
we label \keyword{approximate solutions} -- converges to a solution of
the singular initial value problem for the linear system with
vanishing leading term.

The first step in showing this convergence is to set up the formalism
to work with weak solutions to the linear system. To do this, we
define a \keyword{test function} for this system to be any smooth function
$\phi: (0,\delta] \times T^1 \rightarrow \R^d$ for which there is a
$T\in (0,\delta]$, such that $\phi(t,x)=0$ for all $t>T$. We then
define the operators $\mathcal L$ and $\mathcal F$ acting on functions
$w \in  X_{\delta,\mu,0}$ via\footnote{The operator $\mathcal L$ is
  the adjoint form of the linear principal part operator $L$ (see
  \Eqref{linprincpartop}).}
\begin{align*}
  &\left<\mathcal L[w],\phi\right>
  :=-\int_0^\delta\Bigl(
  \scalarpr{\RR{\mu}  S_1 w}{D\phi}_{L^2(T^1)}+\scalarpr{\RR{\mu}S_2 w}{ t\partial_x\phi}_{L^2(T^1)}\\
  &+\scalarpr{\RR{\mu}\left(S_1 -N w+
    \RR{\mu}^{-1}D \RR{\mu} S_1+DS_1
 +\RR{\mu}^{-1} t\partial_x \RR{\mu} S_2+t\partial_x S_2\right)}{\phi}_{L^2(T^1)}\Bigr)dt 
\end{align*}
and
\begin{equation*}
  \left<\mathcal F[w],\phi\right>
  :=\int_0^\delta \scalarpr{\RR{\mu}\left(f_0+F_1 w\right)}{\phi}_{L^2(T^1)} dt,
\end{equation*}
where $\phi$ is an arbitrary test function\footnote{Note that  we do not need to set $\phi(0,x)=0$ in order to avoid boundary terms at $t=0$; the decay of $w(t,\cdot)$ and of $f_0(t,\cdot)$ as $t$ approaches zero removes such potential boundary terms.}.
These operators are well-defined for $w\in X_{\delta,\mu,0}$ so long as the system 
\Eqref{eq:1stordersystem} is a linear symmetric hyperbolic Fuchsian system  for parameters $\delta$, $\mu$, $\zeta$, $r$, $q$ and $q_0$ as in
\Defref{def:linearity}.
We now define $w$ to be a \keyword{weak solution} of \Eqref{eq:1stordersystem} with vanishing leading term provided it satisfies, for \emph{all} test functions $\phi$,  
\begin{equation}
\label{weaksoln}
\left<\mathcal P[w],\phi\right>:=\left<\mathcal L[w]-\mathcal F[w],\phi\right>=0.
\end{equation}
Here, we note the discussion of distributional time derivatives in  \Sectionref{sec:propertiesspaces} of the appendix. 

Before proceeding to show that weak solutions exist, we establish the following useful technical result. 

\begin{lemma}
\label{lem:bnddweakoperators}
Suppose that \Eqref{eq:1stordersystem} satisfies the conditions to be
a linear symmetric hyperbolic Fuchsian system for a fixed set of parameters $\delta$, $\mu$, $\zeta$, $\delta$, $r$,
$q$, $q_0$ as per \Defref{def:linearity}, and is block diagonal with respect to $\mu$. Then for every test function $\phi$, the maps
$\left<\mathcal L[\cdot],\phi\right>$ and $\left<\mathcal
  F[\cdot],\phi\right>$ are bounded linear functionals on
$X_{\delta,\mu,0}$.
\end{lemma}

\begin{proof}
To prove this lemma it is sufficient to  show that each term in  
 $\left< \mathcal L[w], \phi \right>$ is bounded by $C || w ||_{\delta, \mu, 0}$, for some positive constant $C$ and for every $w\in X_{\delta,\mu,0}$.
We demonstrate this for the first term, $\int_0^\delta \scalarpr{\RR{\mu} S_1 w}{D \phi}_{L^2}dt$.
 Using H\"older's inequality, the spatial continuity\footnote{This follows from the definition of a linear symmetric hyperbolic Fuchsian system, and from Sobolev embedding.}  of $S_1$  and the block-diagonal property, we find that
\begin{align*}
\left| \int_0^\delta \left< \RR{\mu} S_1 w, D \phi \right>_{L^2} dt \right |
	\le\, & \int_0^\delta || \RR{\mu} S_1 w  ||_{L^2} ||D \phi ||_{L^2} dt\\
        \le\, & \delta \sup_{t\in (0,\delta]} || \RR{\mu} S_1 w  ||_{L^2} ||D \phi ||_{L^2}
        \le  C || w  ||_{\delta,\mu,0}.
\end{align*}
The constant $C$, which is used to estimate both the contributions from $S_1$ and from $\phi$, is uniform in the sense defined above.
Other terms in $\left<\mathcal L[w],\phi\right>$ follow similarly, and the same arguments hold for the $\left<\mathcal F[\cdot],\phi\right>$ operator.
\end{proof}

We now determine that, for a given linear symmetric hyperbolic Fuchsian system with certain conditions holding, the singular initial value problem with zero leading term has a weak solution. In the proof of this result, we show that these solutions can be obtained as a limit of approximate solutions of the Cauchy problem, as described above. 

\begin{proposition}[Existence of weak solutions of the
  linear singular initial value problem with smooth coefficients]
  \label{prop:linearfirstexistence}
  Suppose that \Eqref{eq:1stordersystem} satisfies the same conditions as stated in \Lemref{lem:energyestimates} and hence is a linear symmetric hyperbolic Fuchsian system (with smooth coefficients)  for $\delta$, $\mu$, $\zeta$, $r$, $q$, and $q_0$ as per \Defref{def:linearity} with $q=0$ and $q_0=2$. Then there exist weak solutions $w:(0,{\delta}] \times T^1\rightarrow \R^d$ to the singular initial value problem (with vanishing leading term)  which are elements of $X_{{\delta}, \mu, 0}$.
\end{proposition}

Observe  that \Propref{prop:linearfirstexistence} is the most general
existence result which we obtain for linear equations, in the sense that only
minimal control of the behavior of the coefficients of the equation is
required (i.e., $q=0$, $q_0=2$ as in \Lemref{lem:energyestimates}).  We
discuss higher regularity of the solutions under stronger regularity
assumptions in \Propref{prop:linearexistenceregularitypre}
below. We also note that while \Propref{prop:linearfirstexistence} provides sufficient conditions for the existence of solutions, it tells us nothing regarding uniqueness. To obtain uniqueness, we need to impose stronger assumptions on the coefficients of the PDE system; see
\Propref{prop:uniqueness} below.
  
\begin{proof}
  As described
  above, we choose a sequence $(t_n)$ converging to zero, and the
  corresponding sequence of approximate solutions $(v_n) \in
  C^0( (0,\delta],H^{q_0}(T^1))\cap X_{\delta,\mu,0}$. 
  We seek to show
  that the sequence $(v_n)$ forms a Cauchy sequence in $X_{\delta,\mu,0}$. Defining
  $\xi_{mn}:= v_m-v_n$, we readily see that
\begin{align}
\xi_{mn}(t,x) 
	= \left\{
    		\begin{array}{llll}
    	 	0, &  t \in (0,t_m], &\\
       		v_m, &  t \in (t_m, t_n], & \\
	       v_m - v_n, &  t \in (t_n,\delta]. &  
     		\end{array}
   	\right.
\end{align}
From the energy estimate for the Cauchy problem \Lemref{lem:energyestimates} on each subinterval, we then derive
\begin{align}
 \label{eq:xiestimate}
|| \RR{\mu}(t, \cdot) \xi_{mn}(t,\cdot) ||_{L^2} \left\{
     	\begin{array}{lll}
	      	=& 0, &  t \in (0,t_m],    \\
       		\le & 0+C \int_{t_m}^t s^{-1} ||\RR{\mu}f_0 ||_{L^2} ds,          &  t \in (t_m, t_n], \\
       		\le & ||\RR{\mu}(t_n, \cdot) v_m(t_n, \cdot)||_{L^2},             &  t \in (t_n, \delta], 
     \end{array}
   \right.
\end{align}
where in the last inequality we have used the energy/norm equivalence
\Eqref{normequiv} above, and we have also used the fact that the
(linear) PDE system for $v_m-v_n$ has a vanishing source term $f_0$. Recalling the definition of the norm $|| \cdot||_{ \delta, \mu, q}$, noting the monotonicity of  $\int_{t_m}^t s^{-1} ||\RR{\mu}f_0 ||_L^2 ds $, and noting the equality $\xi_{mn}(t_n, \cdot) = v_m(t_n, \cdot)$ for $t=t_n$, we now have
\begin{equation*}
||\xi_{mn}||_{\delta, \mu, 0} 
	= \sup_{t \in (0, \delta]} || \RR{\mu}(t, \cdot) \xi_{mn}(t, \cdot) ||_{L^2} 
	\le \widetilde C  \int_{t_m}^{t_n} s^{-1} ||R[\mu]f_0 ||_{L^2} ds.
\end{equation*}

To complete the argument that we have a Cauchy sequence, it is useful to introduce  
\begin{equation}
\label{eq:gdefinition}
G(t) := \int_0^t s^{-1} ||\RR{\mu} f_0 ||_{L^2} (s) ds,  
\end{equation}
which is well-defined so long as $f_0 \in X_{ \delta, \nu, 0}$ for  $\nu > \mu$. Choosing $\epsilon>0$ as a lower bound for the gap between $\nu$ and $\mu$ among all components, we see that there must exist a constant $C$ such that $G(t) \le C t^\epsilon$; thence, we have
\begin{equation}
\label{eq:errorestimate}
||\xi_{mn}||_{ \delta, \mu, 0} 
	\le C | G(t_n) - G(t_m) |,
\end{equation}
from which it easily follows that $(v_n)$ is a Cauchy sequence in the Banach space $ X_{\delta,\mu,0}.$

Since it has been established (in \Lemref{lem:bnddweakoperators}) that 
$\mathcal P = \mathcal L - \mathcal F$
 is a continuous operator on  $X_{\delta,\mu,0}$, to show that the limit of the Cauchy sequence $(v_n)$ is a weak solution of the system of interest, it is sufficient to show that the limit of the sequence of reals $(\left< \mathcal P[v_n],\phi \right>)$ is zero for all test functions $\phi$. Choosing  any $v_n$ in our sequence, we know from its definition that $v_n$ vanishes on $(0, t_n]$ and is a solution to the equation 
$\left< \mathcal P[\cdot],\phi \right>=0$ on $[t_n, \delta]$. Recalling the definition of $\mathcal P$, we calculate on this latter interval, for any test function $\phi$,  
\begin{align*}
\left | \left< \mathcal P[v_n],\phi \right> \right |=&  \left | - \int_0^{t_n} \left< \RR{\mu}f_0, \phi \right>_{L^2(T^1)} dt \right |.
\end{align*}
Straightforward calculation then shows that \begin{align*}
\left | - \int_0^{t_n}  \left< \RR{\mu} f_0, \phi \right>_{L^2(T^1)} \right | dt 
	\le & \int_0^{t_n}   | \left< \RR{\mu} f_0, \phi \right>_{L^2(T^1)} | dt \\	
	\le & \int_0^{t_n}  \left( 
		\Big( 
			\int_{T^1} dx |\RR{\mu} f_0|^2 
		\Big)^{1/2} 
		\Big( \int_{T^1} dx |\phi|^2 
		\Big)^{1/2} \right) dt \\	
	= &  \int_0^{t_n}  \left(
		 t^{-1} 	
		 \Big(
		 	 \int_{T^1} dx |\RR{\mu} f_0|^2 
		\Big)^{1/2} 
		t 
		\Big(
		 	\int_{T^1} dx |\phi|^2
		\Big)^{1/2} \right) dt \\
	\le & \sup_{t\in (0, \delta]} ||t \phi||_{L^2}  \int_0^{t_n}  t^{-1} ||\RR{\mu} f_0(t) ||_{L^2} dt 
\le C G(t_n),
\end{align*}
from which it follows (from the properties of $G(t))$, that we have a weak solution.
\end{proof} 

Based on this existence result for weak solutions, we would like to define a map which, for a fixed choice of $S_1$, $S_2$, $N$ and $F_1$, maps any smooth function $f_0\in X _{\delta,\nu,0}$ to a weak solution $w\in X _{\delta,\nu,0}$ of $\left< \mathcal P[v_n],\phi\right>=0.$  Then as a next step, we would  like to extend this map to all $f_0$ of $X_{\delta,\nu,0}$, and thereby show that weak solutions exist for all $f_0\in X _{\delta,\nu,0}$, and not just for those $f_0$ which are smooth.  While the lack of a uniqueness result for weak solutions is an impediment to defining the desired map, we can get around this by provisionally defining an operator of this sort which maps a smooth $f_0$ to the solution of the weak solution which is obtained as the limit of the sequence $(v_n)$ (as discussed in the proof of \Propref{prop:linearfirstexistence}). We do this now, noting that the definition makes sense only after we have established that we get the same limit for any choice of the sequence of times $t_n$).  We then  establish an estimate for this operator, and use this estimate to extend the operator to all of $X _{\delta,\nu,0}$.

\begin{proposition}
\label{prop:Hlemma}
Presuming the hypotheses listed in \Propref{prop:linearfirstexistence}, there exists an operator $\mathbb H : X_{\delta, \nu, 0} \to X_{\delta, \mu, 0}$ which maps a smooth source function $f_0$ to the weak solution $w$ of $\left< \mathcal P[w],\phi \right>=0$ which is obtained as the limit of the sequence of approximate solutions $(v_n)$ corresponding to a choice of a monotonic sequence of times ($t_n)$ converging to zero. This operator is well-defined (independent of the choice of the sequence $(t_n)$) and satisfies the estimate
\begin{equation}
  \label{eq:continuityHPDE}
  \|\mathbb H[f_0]\|_{\delta,\mu,0} \le \delta^\rho C\|f_0\|_{\delta,\nu,0},
\end{equation}
for all smooth $f_0\in X _{\delta,\nu,0}$. The positive constants $C$ and $\rho$ are uniform.  

The operator extends to all  (not necessarily smooth) $f_0\in X_{\delta,\nu,0}$, with the estimate \eqref{eq:continuityHPDE} holding for all such $f_0$ with the same constants. Indeed, this extended operator $\mathbb H$ maps \textit{all} $f_0\in X_{\delta,\nu,0}$ to weak solutions of \Eqref{eq:1stordersystem}.
\end{proposition}

The last paragraph in this proposition generalizes the existence result in \Propref{prop:linearfirstexistence} to all, not necessarily smooth, source terms $f_0\in X_{\delta,\nu,0}$. We note, however, that otherwise the system is still assumed to have smooth coefficients in the sense defined above. 

\begin{proof}
We presume initially (as part of the hypothesis of smooth coefficients) that $f_0$ is smooth; i.e., $f_0 \in
C^\infty((0,\delta]\times T^1)\cap X_{\delta,\nu,0}$.
To show that $\mathbb H$ is a well-defined map from
$C^\infty((0,\delta]\times T^1)\cap X_{\delta,\nu,0}$ to 
$X_{\delta,\mu,0}$,
independent of the choice of time sequence, we choose a pair of 
such sequences $(t^1_n)$ and $(t^2_m)$ with their corresponding sequences $(v^1_n)$ and $(v^2_m)$ of approximate solutions, and from the union of the two time sequences we construct a third time sequence $(t_l)$. As is the case for  $(v^1_n)$ and $(v^2_m)$, the combined sequence of approximate solutions  $(v_l)$ must be  a Cauchy sequence, so\footnote{Here, we set $ \delta$ to be the smallest bound among the two sequences.} $||v_n^1 - v_m^2 ||_{ \delta, \mu, 0}$ must vanish in the limit $n,m \to \infty$. Then labeling $w^1$ as the limit of the first sequence and $w^2$ as the limit of the second, we calculate  
\begin{align*}
||w^1 - w^2||_{ \delta, \mu, 0} \le ||w^1 - v_n^1 ||_{ \delta, \mu, 0} +|| v_m^2 - w^2 ||_{ \delta, \mu, 0} + || v_n^1 - v_m^2 ||_{ \delta, \mu, 0} .
\end{align*}
It easily follows that $w^1$ and $w^2$ are equal in $X_{ \delta, \mu, 0}$.  

To prove the estimate for $\mathbb H$ (restricted to smooth $f_0$), we let $(v_n)$ be a sequence of approximate solutions with limit $w =\mathbb H(f_0) $, and then based on \Eqref{eq:errorestimate}  we determine that $|| w - v_1 ||_{ \delta, \mu, 0} \le CG(t_1) \le C G( \delta)$. It then follows that 
\begin{equation*}
|| w ||_{ \delta, \mu, 0} \le || v_1 ||_{ \delta, \mu, 0} + C G( \delta). 
\end{equation*}
If we now apply  the energy estimates to show that $|| v_1 ||_{ \delta, \mu, 0} \le \widetilde C G( \delta)$, we deduce that 
\begin{equation*}
|| w ||_{ \delta, \mu, 0} \le  C G( \delta),
\end{equation*}
for some adapted constant $C$. To relate $G( \delta)$ to the source term, we check that 
\begin{equation*}
s^{-\rho} ||\RR{\mu} f_0 ||_{L^2} \le || f_0 ||_{ \delta, \nu, 0}
\end{equation*}
for some $\rho >0$ so long as $\mu < \nu$. It then follows from multiplying both sides by $s^{-1}$ and integrating over $\int_0^{ \delta}$ that 
\begin{equation*}
G( \delta) \le \frac{1}{\rho}  \delta^ \rho ||f_0 ||_{ \delta, \nu, 0}. 
\end{equation*}
The estimate \Eqref{eq:continuityHPDE} is then a consequence. 

To extend the domain of $\mathbb H$ from $C^\infty((0,\delta]\times T^1)\cap X_{\delta,\nu,0}$ to $X_{\delta,\nu,0}$, we note that this first space is dense in the second by definition. Hence, for any $f_0 \in X_{\delta,\nu,0}$, we can find a sequence of functions $f_{0,j} \in C^\infty((0,\delta]\times T^1)\cap X_{\delta,\nu,0}$ which converges to $f_0$. It follows as a consequence of the estimate \Eqref{eq:continuityHPDE}  that there is a 
unique continuous extension of  $\mathbb H$ to the full space 
$X_{\delta,\nu,0}$.
The extended operator, which we refer to with
the same symbol $\mathbb H$, is continuous and satisfies the same estimate. 

The continuity of the extended operator $\mathbb H$ and the continuity of $\left<\mathcal {P}[w],\phi\right>=0$ with respect to $w$ easily implies that $\mathbb H$ maps any $f_0\in X_{\delta,\nu,0}$, even those which are not smooth,  to weak solutions.
\end{proof}

To proceed from weak solutions to strong solutions of the singular
initial value problem for these linear systems (while still keeping the
smoothness assumption for the coefficients $S_{1,1}$, $S_{2,1}$, $N_1$
and $F_1$),
we need to determine the regularity of these weak solutions.  We do this in the following proposition, and thereby prove the existence of strong solutions.

\begin{proposition}[Regularity of solutions for smooth coefficients]
  \label{prop:linearexistenceregularitypre}
  Suppose that all of the conditions of Proposition
  \ref{prop:linearfirstexistence} hold, with the exception that 
  $q\ge 1$ and $q_0=q+2$. Then, weak solutions
  $w$ of the singular initial value problem (whose existence has been checked in \Propref{prop:linearfirstexistence}) 
  are differentiable in time\footnote{in the sense of \Sectionref{sec:propertiesspaces} in the appendix} and hence
  are  strong
  solutions of \Eqref{eq:1stordersystem}, with $w\in X_{\delta, \mu, q}$ and $Dw\in X_{\delta,\mu,q-1}$. As well, the solution operator
  $\mathbb H$ defined in \Propref{prop:linearfirstexistence} maps
  $X_{\delta,\nu,q}$ to $X_{\delta,\mu,q}$, and satisfies
  \begin{equation}
   \label{eq:continuityHPDENew}
    \|\mathbb H[f_0]\|_{\delta,\mu,q} \le \delta^\rho C\|f_0\|_{\delta,\nu,q},
  \end{equation}
  for all (not necessarily smooth) $f_0\in X _{\delta,\nu,q}$. The constants $C>0$ and $\rho>0$
  are uniform in the sense of \Defref{def:uniform} (but may depend in particular on $q$). 
\end{proposition}

Observe (without pursuing the details here) that an  estimate similar to  \Eqref{eq:continuityHPDENew} can also be proven for the time derivative of the solution. Additional regularity assumptions on the time derivatives of the coefficients of the equation also allow one to prove corresponding statements regarding higher order time derivatives $D^{k'} w$ for $k'\ge 2$.

\begin{proof} 
  Using $w$ to denote the solution to the singular initial value
  problem whose existence is established in Proposition
  \ref{prop:linearfirstexistence} (as an element of $X_{
    \delta,\mu,0}$), and using $(v_n)$ to denote the sequence of
  approximate solutions which converges to $w$, we first note that it
  follows from their definitions\footnote{By definition,  $v_n(t, \cdot)=0$ for $t\in (0,t_n]$,
  and for $t\in [t_n,\delta]$ it is equal to the unique 
  solution of the Cauchy problem with zero data at $t_n$.}
  that the $v_n$ are contained in $C^0( (0,\delta],H^q(T^1))$ -- even in $X_{\delta,\mu,q}$ -- under the hypothesis of \Propref{prop:linearexistenceregularitypre}. Hence, in the
  same way as we have used the energy estimates in
  \Lemref{lem:energyestimates} to show that $(v_n)$ is a Cauchy
  sequence in $X_{ \delta,\mu,0}$ for the proof of
  \Propref{prop:linearfirstexistence}, we can now use the energy
  estimates of \Lemref{lem:energyestimateshigherorder} to show that
  $(v_n)$ is a Cauchy sequence in $X_{ \delta,\mu-\epsilon,q}$  for an arbitrarily small constant
  $\epsilon>0$. We then show that the limit of this Cauchy sequence
  equals $w$ above.  The solution $w$ is hence in $X_{
    \delta,\mu-\epsilon,q}$ and we get the estimate
\[\|\mathbb H[f_0]\|_{  \delta,\mu-\epsilon,q} \le \delta ^\rho
C\|f_0\|_{ \delta,\nu,q}.
\]

Now, if the equation is of linear symmetric hyperbolic Fuchsian form
for a choice of $\mu$, as we have assumed so far, it is also of linear
symmetric hyperbolic Fuchsian form for a choice of
$\widehat\mu:=\mu+\epsilon$ if $\epsilon>0$ is sufficiently small in
comparison to $\nu-\mu$. Moreover, the assumption that the system is
block diagonal with respect to $\mu$ also implies that this is the
case with respect to $\widehat\mu$; the same is true for the condition involving the energy dissipation matrix. Hence, we can apply the argument in
the previous paragraph based on this choice of $\widehat\mu$.  This leads
to the conclusion that, in fact, the solution $w$ is in $X_{
  \delta,\widehat\mu-\epsilon,q}=X_{ \delta,\mu,q}$ (as opposed to $X_{
  \delta,\mu-\epsilon,q}$ above) and
\[\|\mathbb H[f_0]\|_{ \delta,\mu,q} \le \delta^\rho
C\|f_0\|_{ \delta,\nu,q},
\]
possibly after a slight change of the constants $C$ and $\rho$.

Next we show that the solution $w$ is differentiable  in time. We define
\[\widehat v_n:=S_1^{-1}\left(-S_2 t\partial_x v_n- N v_n+f_0+F_1
  v_n\right).
\] 
We know that $\widehat v_n\in X_{\delta,\mu,q-1}$ and $\widehat v_n(t)=D v_n(t)$ for all $t\in [\delta_I,\delta]$ for any $\delta_I\in (0,\delta)$ and for $n$ sufficiently large;  we cannot choose $\delta_I=0$ here since the time
derivative of $v_n$ is in general not defined at $t=0$.
Moreover, we find from the definition and the convergence of the sequence $v_n$ that
\[||\widehat v_n-\widehat v_m\|_{\delta,\mu,q-1}\le C \|v_n-v_m\|_{\delta,\mu,q}\rightarrow 0\]
for a uniform constant $C>0$. Hence there exists $\widehat v\in X_{\delta,\mu,q-1}$ such that $\widehat v_n\rightarrow \widehat v$. The estimate also holds if we restrict the time interval to $[\delta_I,\delta]$ as above and hence we find that $Dv_n(t)=\widehat v_n(t)\rightarrow\widehat v(t)$ uniformly at every $t\in [\delta_I,\delta]$. It is then a standard result that $w$ is differentiable in $t$ at every $t\in [\delta_I,\delta]$ and further that $\widehat v(t)=Dw(t)$. Since $\delta_I$ can be chosen arbitrarily small, it follows that for all  $t\in (0,\delta]$,
 $w$ is differentiable in $t$ and  $\widehat v=Dw$. Consequently, we find  that $Dw=\widehat v\in X_{\delta,\mu,q-1}$.

To argue that $w$ is a strong solution, we start from the fact that $w$ is a solution of the weak equation whose integral representation can be integrated by parts in both time (using \Eqref{eq:distribtimederivative} in the appendix) and space. We may then choose a suitable sequence of test functions, for example those which are used as mollifiers in \Lemref{lem:mollifiers} in the appendix, so that the resulting system converges pointwise almost everywhere to one of the components of \Eqref{eq:1stordersystem} evaluated at one point $(t,x)$. Doing this for every component and for every point $(t,x)\in (0,\delta]\times T^1$, we determine  that $w$ is actually a solution of the strong equation almost everywhere.
\end{proof}

To this point, we have assumed throughout our analysis that the
matrices $S_{1,1}$, $S_{2,1}$, $N_1$ and $F_1$ are smooth; i.e., we
have not thus far allowed these matrices to be general elements of the  spaces $B_{\delta,\zeta,q,r}$ from \Defref{def:linearity}. If we wish to use our current (linear) results as a tool for proving that there
are (unique) solutions to the singular initial value problem for the (nonlinear)
quasilinear system, we need to generalize these  linear results to
include the possibility that the matrices listed above are not smooth
(since, in the quasilinear case, these matrices are functions of the
solutions, which may not a priori be smooth).

Before carrying out this generalization of the existence (and uniqueness) results for the linear singular initial value problem, we note that we can at this stage assume that the term $F_1$ vanishes. This simplification does not constitute an essential loss of generality since in our work below, we replace the linear source term function $f_0$ by a general quasilinear expression shortly; the resulting expression  then incorporates the effects of the term $F_1$. We recall that the $F_1$-term plays a convenient role in our verification that the higher order energy estimates of Lemma \ref{lem:energyestimateshigherorder} hold. Such a term is generated in the linear equation for $\partial_x v$, which we obtain by taking a spatial derivative of the equation for $v$.

\begin{proposition} [Existence for the linear singular initial value problem for non-smooth coefficients]
  \label{prop:linearexistencenonsmooth}
Suppose
that \Eqref{eq:1stordersystem} is a linear symmetric hyperbolic Fuchsian system for $\delta$, $\mu$, $\zeta$, $q$, $q_0$ and $r$ as in \Defref{def:linearity}  (with $F_1=0$)  for $q_0 =q+ 2$ and $q\ge 2$ (not necessarily with smooth coefficients), and suppose that it  is block diagonal with respect to $\mu$. Suppose that the energy dissipation matrix \Eqref{eq:energydissipationmatrix} is positive definite. Then there exists a
solution $w: (0, \delta] \times T^1 \to \mathbb R^d$ to the singular initial value problem with zero leading order term such that $w \in X_{\delta, \mu, q}$ and $Dw \in X_{\delta, \mu, q-1}$. The solution operator $\mathbb H: f_0\mapsto w$ maps $X_{\delta, \mu, q}$ into $X_{\delta, \nu, q}$, and satisfies
\[ \| \mathbb H[f_0] \|_{\delta, \mu, q} \leq \delta^\rho C \| f_0 \|_{\delta, \nu, q}
 \]
for some positive uniform constants $C$ and $\rho$.
\end{proposition}

Observe that this result also holds in the case that $q_0$ and $q$ are both infinite: If the conditions of this proposition are satisfied for all integers $q_0 =q+2$ and $q \ge 2$, then $w \in X_{\delta, \mu, \infty}$ and $Dw \in X_{\delta, \mu, \infty}$. However, the $q$-parametrized sequence of constants $C$ and $\rho$ occurring in the estimate of the solution operator  may in general be unbounded as $q \rightarrow \infty$.

\begin{proof}
The basic idea is to approximate the non-smooth coefficients $S_{1,1}, S_{2,1}$ and $N_1$ by a sequence of smooth ones and then to apply \Propref{prop:linearexistenceregularitypre} to obtain a sequence of approximate solutions. The main work is then to prove that this sequence converges to the solution of the system with non-smooth coefficients in an appropriate sense.
\vspace{0.5ex}

\noindent\textit{Step~1: A sequence of approximate solutions.} We presume that a linear symmetric hyperbolic Fuchsian system with parameters $\delta$, $r$, $q$, $q_0$, $\mu$ and $\zeta$, and with coefficient matrices $S_{1,1}, S_{2,1}$ and $N_1$ in the perturbation space $B_{\delta,\zeta,q,r}$ has been specified. We assume that this system is block diagonal with respect to $\mu$ and that the energy dissipation matrix is positive definite. According to the definition of the space $B_{\delta,\zeta,q,r}$, there exist sequences $(S_{1,1,[j]})$, $(S_{2,1,[j]})$ and $(N_{1,[j]})$ of smooth elements in $B_{\delta,\zeta,q,r}$ which converge to $S_{1,1}, S_{2,1}$ and $N_1$ (in a way described below). We thus obtain a sequence of linear principal part operators (with smooth coefficients)
\begin{equation}
\label{linprincpartopseq}
L_{[j]}[w]:=(S_{1,0} + S_{1,1,[j]})Dw + \RR{1-\beta}(S_{2,0}+S_{2,1,[j]})t \partial_x w +(N_0+N_{1,[j]})w,
\end{equation}
and hence a sequence of systems of the form $L_{[j]}[w]=f_0$. For each $j$, this is a linear symmetric hyperbolic Fuchsian system for $\delta$, $\mu$, $q$, $q_0$, $\zeta$ and $r$ with smooth coefficients. It is clear that the sequences can be chosen so that, for each $j$, the block diagonal condition with respect to $\mu$ is satisfied, and the energy dissipation matrix is positive definite for each equation. Clearly as well, each $S_{1,1,[j]}$ is differentiable in time and $DS_{1,1,[j]}$ is bounded (in a sense which we make more precise below). Hence,
for each $j$, \Propref{prop:linearexistenceregularitypre} implies the existence of a solution operator $\mathbb H_{[j]}$, and therefore a sequence $w_{[j]}\in X_{\delta,\mu,q}$ defined by $w_{[j]}:=\mathbb H_{[j]}[f_0]$.
\vspace{0.5ex}

\noindent\textit{Step~2: Uniformity of the sequence of coefficients.} 
To study the convergence properties of these approximate solutions
$w_{[j]}$ in more detail we make a special choice of the sequences $(S_{1,1,[j]})$, $(S_{2,1,[j]})$ and $(N_{1,[j]})$ as in \Lemref{lem:mollifiers} in the appendix (where we  replace the two indices $i$ and $j$ by just one index $j$). The advantage of this choice is that we will be able to argue that $\|\RR{\widetilde\zeta}DS_{1,1,[j]}(t,\cdot)\|_{L^2}$ is uniformly bounded in $j$ and $t$ under the hypotheses of \Propref{prop:linearexistencenonsmooth}, which will be important for the following argument. The slight disadvantage, however, as stated in
 \Conditionref{cond:uniformconvergence} of  \Lemref{lem:mollifiers}, is that the convergence is guaranteed only with respect to a norm
$\|\cdot\|_{\delta,\widetilde\zeta,q}$ for any exponent matrix $\widetilde\zeta$ \textit{smaller} than $\zeta$; fortunately we will see below that this is not significant. Let us choose such an exponent matrix $\widetilde\zeta$ with strictly positive entries. Let us moreover suppose for the moment that a uniform bound for $\|\RR{\widetilde\zeta}DS_{1,1,[j]}(t,\cdot)\|_{L^2}$ has been found (which we show shortly). By setting
$\xi=\widetilde\zeta$ and choosing some uniform value of the (function space ball) radius $s$ in  Condition~(ii) of \Propref{prop:linearexistenceregularitypre}, we are allowed to apply \Propref{prop:linearexistenceregularitypre}
in such a way that each of the approximate equations $L_{[j]}[w]=f_0$ is a perturbation 
 of \textit{one} common equation in the perturbation space $B_{\delta,\widetilde\zeta,q,r}$. A particular consequence is then that we obtain an estimate for the operators $\mathbb H_{[j]}$ of the form \Eqref{eq:continuityHPDENew} with $C$ \textit{independent of $j$}.

To establish this uniform bound of $\|\RR{\widetilde\zeta}DS_{1,1,[j]}(t,\cdot)\|_{L^2}$, we use \Eqsref{eq:mollifier} and \eqref{eq:mollifier2} from \Lemref{lem:mollifiers} in the appendix to obtain
\begin{align*}
\RR{\widetilde\zeta} DS_{1,1,[j]}
=&\RR{\widetilde\zeta} D\RR{-\widetilde\zeta} \RR{\widetilde\zeta} S_{1,1,[j]}\\
&+\int_0^\infty\int_{T^1}(\RR{\widetilde\zeta} \widehat S_{1,1})(s,y)\frac 1{\alpha_j}\phi\left(\frac{x-y}{\alpha_j}\right) (-1)\frac 1{\alpha_j^2}t\phi'\left(\frac{s-t}{\alpha_j}\right)\,dy\,ds,
\end{align*}
where $\widehat S_{1,1}$ is the extension introduced in  \Lemref{lem:mollifiers}.
We wish to estimate this expression in the $L^2$-norm. The first term can be estimated in the $L^2$-norm by $C(\widetilde\zeta)\|S_{1,1,[j]}\|_{\delta,\widetilde\zeta,q}$ with a constant determined by $\widetilde\zeta$. The second term is treated as follows\footnote{We write e.g.\ $L^2_x(T^1)$ to fix the integration variable corresponding to the norm.}
\begin{align*}
  &\left\|\int_0^\infty\int_{T^1}(\RR{\widetilde\zeta} \widehat S_{1,1})(s,y)\frac 1{\alpha_j}\phi\left(\frac{x-y}{\alpha_j}\right)\frac 1{\alpha_j^2}t\phi'\left(\frac{s-t}{\alpha_j}\right)\,dy\,ds\right\|_{L^2_x(T^1)}\\
&\le \left\|\sup_{s\in (0,\delta]}\left(\int_{T^1}(\RR{\widetilde\zeta} S_{1,1})(s,y)\frac 1{\alpha_j}\phi\left(\frac{x-y}{\alpha_j}\right)\,dy\right)
\left|\int_0^\infty\frac 1{\alpha_j^2}t\phi'\left(\frac{s-t}{\alpha_j}\right)\,ds\right|\right\|_{L^2_x(T^1)}\\
&\le \left\|\sup_{s\in (0,\delta]}\left(\left\|(\RR{\widetilde\zeta} S_{1,1})(s,y)\right\|_{L _y^\infty(T^1)}
\left\|\frac 1{\alpha_j}\phi\left(\frac{x-y}{\alpha_j}\right)\right\|_{L^1_y(T^1)}\right)\right\|_{L^2_x(T^1)}\\
&\qquad\cdot\left|\int_0^\infty\frac 1{\alpha_j^2}t\phi'\left(\frac{s-t}{\alpha_j}\right)\,ds\right|,
\end{align*}
where we have used the definition of the extension $\widehat S_{1,1}$ in the second line.
The properties of the kernel $\phi$ imply that the term with the $L^1_y$-norm is unity (independently of $x$). Since $q\ge 1$, we can use  Sobolev embedding to estimate the term with the $L^\infty_y$-norm by the $H^q_y$-norm and hence the $\sup_s$-term by the norm $\|\cdot\|_{\delta,\widetilde\zeta,q}$. As a consequence, all quantities are independent of $x$, and therefore we find  that the second term of $\RR{\widetilde\zeta} DS_{1,1,[j]}$ above can in total be estimated in the $L^2$-norm as less than or equal to 
\begin{equation*}
C \left\|S_{1,1}\right\|_{\delta,\widetilde\zeta,q}
\left|\int_0^\infty\frac 1{\alpha_j^2}t\phi'\left(\frac{s-t}{\alpha_j}\right)\,ds\right|,
\end{equation*}
for some constant $C>0$ which is, in particular, independent of $t$ and $j$. Hence, we only need to estimate the last integral:
\begin{equation*}
  \int_0^\infty\frac 1{\alpha_j^2}t\phi'\left(\frac{s-t}{\alpha_j}\right)\,ds
  =t \int_{-t/\alpha_j}^{\infty} \frac 1{\alpha_j}\phi'(\sigma)d\sigma
  =-\frac t{\alpha_j}\phi(-t/\alpha_j).
\end{equation*}
This takes  the form $x\cdot \phi(x)$ for all $x\in\R$. This has the property that it vanishes for all $|x|\ge 1$ and is bounded for all $|x|\le 1$. Summarizing, we have thus confirmed that
\[\|DS_{1,1,[j]}\|_{\delta,\widetilde\zeta,0}\le C \left\|S_{1,1}\right\|_{\delta,\widetilde\zeta,q}
\]
for every $j$, and for a constant $C$, which, in particular, does not depend on $j$.
\vspace{0.5ex}

\noindent\textit{Step~3: Convergence of the sequence $(w_{[j]})$ in $X_{\delta,\mu,q-1}$.}
We seek to show that the sequence $(w_{[j]})$ converges\footnote{The reason for considering $q-1$ and not $q$ is made clear below.} in $ X_{ \delta, \mu, q-1}$. 
We do this by showing  that $(w_{[j]})$ is a Cauchy sequence:  Setting $\xi_{[ij]} := w_{[i]} - w_{[j]}$, we derive the equation
\begin{equation}
\label{CauchySeqEqn}
L_{[i]} \xi_{[ij]}= - (L_{[i]}-L_{[j]})w _{[j]},
\end{equation}
where we interpret the right hand side as a source term for this linear equation for $\xi_{[ij]}$. 
One readily checks that \Eqref{CauchySeqEqn} is a linear symmetric hyperbolic system for the same parameters as above, but with differentiability index $q-1$ (since the source term incorporates one spatial derivative); hence, so long as $q\ge2$, we may apply \Propref{prop:linearexistenceregularitypre} also to this equation. 
We thus obtain, as a consequence of \Eqref{eq:continuityHPDENew},
\[
\|w_{[i]}-w_{[j]}\|_{ \delta,\mu,q-1}\le C \left\|-(L_{[i]}-L_{[j]})
  w_{[j]}\right\|_{ \delta,\mu+\widetilde\zeta_{\min},q-1},
\] 
where the (scalar) constant $\widetilde\zeta_{\min}$ is the minimal value of all
components of $\widetilde\zeta$ at all $x\in T^1$ (note that $\widetilde\zeta_{\min}$ is
positive), and where the generic constant $C$ represents the constant
in \Eqref{eq:continuityHPDENew} (for $\nu$ replaced by
$\mu+\widetilde\zeta_{\min}$). It is  crucial  here that the constant $C$
does not depend on the index $i$; this is a consequence of the uniformity of
the constant in \Eqref{eq:continuityHPDENew}.
If we now expand out the definition of
$L_{[j]}$, using the block diagonality conditions, the Sobolev
embedding (for spatial dimension one and $q\ge 2$) and the Moser
inequality (stated for instance in Proposition~3.7 in Chapter~13 of \cite{Taylor:2011wn}), we find that
\begin{align*}
  \|w_{[i]}-w_{[j]}\|_{ \delta,\mu,q-1}\le 
  C\Bigl(&\|S_{1,1,[i]}-S_{1,1,[j]}\|_{ \delta,\widetilde\zeta,q-1}
  +\|S_{2,1,[i]}-S_{2,1,[j]}\|_{ \delta,\widetilde\zeta,q-1}\\
  &+\|N_{1,[i]}-N_{1,[j]}\|_{ \delta,\widetilde\zeta,q-1}
  \Bigr)\|f_0\|_{ \delta,\nu,q}.
\end{align*}
It immediately follows that $(w_{[j]})$ is a Cauchy sequence in
$X_{ \delta,\mu,q-1}$, and hence the sequence has a limit  $w\in
X_{ \delta,\mu,q-1}$.
\vspace{0.5ex}

\noindent\textit{Step~4: The limit $w$ is a solution of the original equation.}
 Standard arguments of the sort carried out
in the proof of Proposition \ref{prop:linearfirstexistence} show
that $w$ is a weak solution of the system $Lw=f_0$ with non-smooth coefficients. 
Since each
$w_{[j]}$ is differentiable in time and a strong solution of the equation, we can solve
each equation for
$Dw_{[j]}$. Relying on an argument similar to that used in the proof of
\Propref{prop:linearexistenceregularitypre}, we see that it follows that $w$ is differentiable
in time, with $Dw \in X_{\delta,\mu,q-2}$, and therefore  the limit $w$ is actually a strong
solution of the equation.
\vspace{0.5ex}

\noindent\textit{Step~5: The limit $w$ is in $X_{\delta,\mu,q}$.}
We now show that $w$ is in fact in $X_{\delta,\mu,q}$ (and not
only in $X_{\delta,\mu,q-1}$) and consequently  $Dw\in
X_{\delta,\mu,q-1}$ (and not only in $X_{\delta,\mu,q-2}$).  It
follows from its definition that the sequence $(w_{[j]})$ is contained
in $X_{\delta, \mu, q}$, and furthermore, as a consequence of the operator estimate
\Eqref{eq:continuityHPDENew}, we have that
\begin{equation}
  \label{eq:boundednessseq}
  \|w_{[j]}\|_{ \delta,\mu,q}\le C \|f_0\|_{ \delta,\nu,q}.
\end{equation} 
Here, the constant $C$ is independent of the index $j$
as a result  of our discussion of uniformity above. We thus find  that the
sequence is uniformly bounded in $X_{ \delta, \mu, q}$. If we now fix
a time $t_0\in (0,\delta)$, then the sequence $(w_{[j]}(t_0))$ is 
bounded in the Hilbert space $H^q(T^1)$.
Since the sequence $(w_{[j]}(t_0))$ is also
convergent in the Hilbert space $H^{q-1}(T^1)$, we can apply \Corref{cor:rescueregularityT1} from the Appendix. We hence find that the limit $w(t_0)$ is in $H^{q}(T^1)$.
We consider the function $w$, which is a strong solution in $X_{\delta,\mu,q-1}$ with $w(t_0)\in H^q(T^1)$, to be a
strong solution of the Cauchy problem of the linear symmetric
hyperbolic equation with the 
``initial data'' $w(t_0)$ contained in $H^q(T^1)$. Given that the coefficients of the system are contained in $C^0((0,\delta], H^q(T^1))$, the
standard theory of linear symmetric hyperbolic equations (see \cite{Taylor:2011wn})
 implies that $w$, and therefore $\RR{\mu} w$ (recall that $\RR{\mu}$ is in $C^\infty((0,\delta]\times T^1)$) is a continuous map $(0,\delta]\rightarrow
H^q(T^1)$. In fact, this latter map is bounded as a consequence of taking the limit $j\rightarrow\infty$ of
\Eqref{eq:boundednessseq}; i.e., $w\in \widehat X_{\delta,\mu,q}$ (see the appendix). Replacing $\mu$ by $\mu+\widetilde\epsilon$ in all of the previous steps, we see that the same arguments go through as long as $\widetilde\epsilon>0$ is sufficiently small. We therefore find that  $\RR{\mu+\widetilde\epsilon}w$ is a bounded continuous map $(0,\delta]\rightarrow H^q(T^1)$. It then follows from 
\Lemref{lem:relationspaces} in the appendix  that
$w\in X_{ \delta,\mu,q}$.
\vspace{0.5ex}

\noindent\textit{Step~6: Properties of the solution operator $\mathbb H$.}
We have thus 
extended the solution operator $\HPDELinOp$ to the case of
non-smooth coefficients by the above limit procedure. It can thus be checked that the estimate 
\Eqref{eq:continuityHPDENew} still holds with uniform constants.
\end{proof}

Having  obtained a comprehensive existence result for the singular
initial value problem of linear symmetric hyperbolic Fuchsian
equations, we  now show uniqueness of these solutions.

\begin{proposition}[Uniqueness for the linear singular initial value problem]
\label{prop:uniqueness}
Suppose that all of the conditions of
\Propref{prop:linearexistencenonsmooth} hold for a chosen singular
initial value problem (with zero leading-order term). The solution for
this problem is unique in $X_{\delta,\mu,q}$.
\end{proposition}

\begin{proof}
  We consider $w$ and $\widetilde w$ to be a pair of (generally different) solutions to the same singular initial  
  value problem, and we define $\omega:=w- \widetilde w$ to be the difference
  between the two. It follows that $\omega$ is a solution of the same
  equation with vanishing source-term $f_0$, with $\omega(t)$ being an element of
  $H^2(T^1)$ for  every time $t\in (0,\delta]$.  Choosing any $t_0 \in
  (0, \delta]$, we can also consider $\omega|_{(t_0, \delta]}$ to be the
  unique solution of the Cauchy initial value problem (for the same
  linear PDE system) with initial data $\omega(t_0)$. Since the
  solution $\omega$ together with the coefficients have
  $H^2$-regularity and since $S_1$ is guaranteed to be positive
  definite on the whole time interval, we may apply the energy estimate
  \Eqref{eq:EnergyEstimate2}. We obtain
  (replacing $\mu$ by $\mu-\epsilon$, which is allowed for any
  $\epsilon>0$)
  \begin{equation}
    \label{yup}
    ||\RR{\mu-\epsilon} \omega ||_{L^2}(t) \le C ||\RR{\mu-\epsilon} \omega ||_{L^2}(t_0),
  \end{equation}
  for all $t\in (t_0,  \delta]$, with the constant $C$ independent of $t$.
  Observe that it follows from the definition of $\RR{\mu}$ that we can rewrite the right hand side of \Eqref{yup}  as 
  \begin{equation}
  C || \RR{\mu-\epsilon}\omega||_{L^2}(t_0) = t_0^\epsilon C || \RR{\mu}\omega||_{L^2}(t_0), 
  \end{equation}
  so that \Eqref{yup} takes the form
 \begin{equation}
 \label{yupyup}
      ||\RR{\mu-\epsilon} \omega ||_{L^2}(t) \le  t_0^\epsilon C || \RR{\mu}\omega||_{L^2}(t_0).  
  \end{equation} 
  We now take the limits of \Eqref{yupyup}  as $t_0 \to 0$, noting that
  the left hand side of the equation and the constants are unchanged
  by taking this limit.  Since $\RR{\mu} \omega$ is a bounded map from
  $(0,\delta]$ to $L^2(T^1)$, the limit as $t_0 \to 0$ of the right
  hand side of \Eqref{yup} vanishes. It thus follows that for all $t
  \in (0, \delta]$,
  \begin{equation*}
    || \RR{\mu-\epsilon} \omega||_{L^2}(t) =0.
  \end{equation*}
  Then since $\RR{\mu-\epsilon}(t,x)$ is bounded positive at any fixed $t$ on
  $(0, \delta]$, we deduce that $\omega(t,x)=0$ at all $t$; uniqueness
  follows.
\end{proof}

\subsubsection{The nonlinear theory}

The results obtained in Section \ref{sec:lineartheory} pertain
exclusively to linear systems. In this section, we use those results
together with a fixed point iteration procedure to prove Theorem
\ref{th:Wellposedness1stOrderFiniteDiff}, which establishes
existence and uniqueness of solutions to the singular initial value problem for the
nonlinear system \Eqref{eq:1stordersystem} with a (no longer necessarily
vanishing) leading-order term $u_0$.

To start, it is useful to rewrite \Eqref{eq:1stordersystem} in a convenient form. Recalling the definition (see \Eqref{eq:DefLPDE}) of the operator $\LPDE{u}{v} := S_1(u)Dv +S_2(u)t\partial_x v + N(u)v$, we may write \Eqref{eq:1stordersystem} in the form $\LPDE{u}{u} =\f{u}$. Despite the nonlinear nature of $\widehat L$, this operator is linear in the sense that $\LPDE{u}{v_1 + v_2}=\LPDE{u}{v_1}+\LPDE{u}{v_2}$. Hence if we let 
$u_0$ denote a chosen leading-order term (satisfying the hypotheses of Theorem
\ref{th:Wellposedness1stOrderFiniteDiff}), if we set $u=u_0 + w$, and if we recall the definition (see \Eqref{eq:defFLu})
\begin{equation*}
\FLu{w}:= \FPDEu{w}-\LPDEu{w}{u_0},
\end{equation*} 
then  \Eqref{eq:1stordersystem} takes the form
\begin{equation}
\label{LPDEq}
\LPDEu{w}{w}=\FLu{w}.
\end{equation}
The linear analysis discussed in Section \ref{sec:lineartheory} does
not apply to the nonlinear equation \Eqref{LPDEq} directly. Nevertheless, if we linearize this equation by fixing $\widetilde w\in B_{\delta,\mu,q,s}$ (for some $s>0$) and write
\begin{equation}
  \label{eq:linearization}
  \LPDEu{\widetilde w}{w} =\FLu{\widetilde w},
\end{equation}
then (presuming the hypothesis of \Theoremref{th:Wellposedness1stOrderFiniteDiff}) the techniques of Section \ref{sec:lineartheory} are
applicable.  In applying these techniques, we assume
that $\delta$ has been chosen sufficiently small so that
\Conditionref{enum:condition2lindef} of \Defref{def:linearity} is
satisfied; this implies no loss of generality since, as we see below, the argument leading to the proof of \Theoremref{th:Wellposedness1stOrderFiniteDiff} requires further shrinkage of the time
interval. It is important to note that, for  every $\widetilde w\in B_{\delta,\mu,q,s}$, the hypothesis of Theorem 2.4 implies the existence of common quantities $\zeta$ and $r$ (as in Definition 2.2) so that $S_{1,1}(u_0+\widetilde w),S_{2,1}(u_0+\widetilde w)$ and $N_1(u_0+\widetilde w)$ are all contained $ B_{\delta,\zeta,q,r}$. Moreover, if $\widetilde w\in B_{\delta',\mu,q,s}$ for any $\delta'<\delta$, then the same statement regarding $S_{1,1}$, $S_{2,1}$ and $N_1$ holds for the same common $\zeta$ and $r$, but with $\delta$ replaced by $\delta'$.

Replacing the right--hand side of this equation by a fixed function $\phi \in X_{\delta,\nu,q}$, we readily check that the linear system 
\begin{equation}
\label{Eq:linearization}
\LPDEu{\widetilde w}{w}(t,x)=\phi (t,x),
\end{equation}
is of linear symmetric hyperbolic Fuchsian form
(\Defref{def:linearity}) for $q$ and $q_0=q+2$ and for a sufficiently
large constant $r$. Hence, it follows from
Propositions \ref{prop:linearexistencenonsmooth} and \ref{prop:uniqueness} that the system has a
unique solution $w\in X_{\delta,\mu,q}$ (we only require $q\ge 2$ at
this stage of the proof), and we can define the
corresponding solution operator $\HPDEuOp{\widetilde w}$ which maps
the source term $\phi$ to the solution $w= \HPDEu{\widetilde
  w}{\phi}$. The case $\phi=\FLu{\widetilde w}$ corresponds to
\Eqref{eq:linearization}; thus we compose $\HPDEuOp{\widetilde w}$
with $\FLuOp$ to define the operator $\GuOp$ as follows:
\begin{equation*}
\Gu{\widetilde w} := \HPDEu{\widetilde w}{\FLu{\widetilde w}}.
\end{equation*}
Hence, $w=\Gu{\widetilde w}\in X_{\delta,\mu,q}$ is the unique solution of the singular initial value problem
of \Eqref{eq:linearization}.
In terms of $\GuOp$, we see that $w$ is a solution of the singular
initial value problem for the nonlinear equation \eqref{eq:1stordersystem} with leading-order
term $u_0$ if and only if it satisfies $ w= \Gu{w}$; i.e., if and only if $w$ is a
fixed point of $\GuOp$.

The operator $\GuOp$ is the key to the following fixed point iteration
argument. We define the sequence of functions $(w_N)$ by setting
$w_0=0$, and defining $w_{N+1}=\Gu{w_N}$ for $N\in
\N.$ To control this sequence, we need uniform bounds; i.e., we wish
to show that each element of the sequence is contained in 
$B_{\delta,\mu,q,s}$. Suppose that this is true for
$w_0,\ldots,w_N$. 
It follows
from the hypothesis of Theorem
\ref{th:Wellposedness1stOrderFiniteDiff} (given that $w_0=0$) that
\[\|\FLu{w_N}\|_{\delta,\nu,q}\le C \|w_N\|_{\delta,\mu,q}
+\|\FLu{0}\|_{\delta,\nu,q}.
\] 
The constant $C>0$ does not depend on $N$.
Using
the definition of $w_{N+1}$, together with
\Eqref{eq:continuityHPDENew}, we have that
\[\|w_{N+1}\|_{ \delta,\mu,q}\le \delta^\rho
\widetilde C\|\FLu{w_N}\|_{ \delta,\nu,q},
\] 
where $\widetilde C$ and $\rho>0$ are  constants which also do not depend on
$N$. Combining, we obtain
\[\|w_{N+1}\|_{ \delta,\mu,q}\le \delta^\rho
\widetilde C\left(C \|w_N\|_{ \delta,\mu,q}
+\|\FLu{0}\|_{\delta,\nu,q}\right).
\]
We recall that the uniformity of the constants implies that the same
estimate holds with the same constants if we choose to formulate the same singular initial
value problem in terms of a constant $\bar\delta\in (0,\delta)$
instead of $\delta$ itself.
Since we have supposed that $w_N$ is contained in $B_{\delta,\mu,q,s}$ -- that is, $\|w_N\|_{ \delta,\mu,q}\le s$ --
we can find such a sufficiently small $\bar\delta$ so that
\[\bar\delta^\rho
\widetilde C C\le 1/2 \quad\text{and}\quad 
\bar\delta^\rho \widetilde C\|\FLu{0}\|_{ \delta,\nu,q}\le s/2,
\] 
while preserving the
bound $\|w_N\|_{\bar \delta,\mu,q}\le s$.
This  can be done, since $\|\FLu{0}\|_{\bar \delta,\nu,q}\le \|\FLu{0}\|_{\delta,\nu,q}$. For this diminished choice $\bar \delta$, we
thus determine that $\|w_{N+1}\|_{\bar \delta,\mu,q}\le s$. Since the above
estimates do not depend on the index $N$, it follows that the whole sequence is
bounded, and we have  $(w_N)\subset B_{\bar \delta,\mu,q,s}$.

We now consider an arbitrary pair of functions $w, v \in B_{\bar \delta, \mu, q,s}$, and we calculate the following estimate for the
norm of the difference of the operator $\GuOp$ acting on each of
these:
\begin{align*}
  \|\Gu{w}-\Gu{v}\|_{\bar \delta,\mu,q-1}
  \le \, 
& \|\HPDEu{w}{\FLu{w}}-\HPDEu{w}{ \FLu{v}} \|_{\bar \delta,\mu,q-1}\\
  &+\| \HPDEu{w}{ \FLu{v}} 
  -\HPDEu{v}{\FLu{v}}
  \|_{\bar \delta,\mu,q-1}.
\end{align*}
Note that, for reasons discussed below, we work with $\| \cdot
\|_{\bar \delta,\mu, q-1}$ rather than $\| \cdot \|_{\bar \delta,\mu, q}$. It follows from the hypothesis of Theorem
\ref{th:Wellposedness1stOrderFiniteDiff} and from
\Eqref{eq:continuityHPDENew} that the first term on the right hand
side of this estimate satisfies the inequality
\[
\|\HPDEu{w}{\FLu{w}}-\HPDEu{w} {\FLu{v}}\|_{\bar\delta,\mu,q-1} 
\le \mathcal {C} \, \|w-v\|_{\bar\delta, \mu, q-1},
\] 
with the Lipschitz constant $\mathcal C$ smaller than unity so long as
we allow a further decrease in $\bar\delta$; the argument for this is the same as for the semilinear case \cite{Beyer:2010fo,Beyer:2010tb}. This controls this first
term.

To estimate the second term on the right hand side, we set
\[
w_A:=\HPDEu{w}{\FLu{v}},\qquad
\qquad 
w_B:=\HPDEu{v}{\FLu{v}}.
\] 
It follows from the definition of $\HPDEuOp{w}$ that 
\[
\LPDEu{w} {w_A}=\FLu{v}, \qquad
\qquad 
\LPDEu{v} {w_B}=\FLu{v}.
\]
Therefore, setting $\LPDEu{w}{w_A}$ and $\LPDEu{v} {w_B}$
equal, and using the linear property of the operator $\LPDEuOp{w} $
noted above, we derive
\begin{equation}
\label{eq:wAwB}
\LPDEu{w}{w_A-w_B}=-(\LPDEuOp{w}-\LPDEuOp{v})[w_B].
\end{equation}
The right hand side of this equation is similar to a term which appears in \Eqref{CauchySeqEqn}; thus
 we can treat it using similar techniques to those used in the proof of
Proposition \ref{prop:linearexistencenonsmooth}. In doing this, we rely on the hypothesis for \Theoremref{th:Wellposedness1stOrderFiniteDiff}, and  we use the
condition $q\ge 3$  in order to guarantee that the source term of
\Eqref{eq:wAwB} has at least two spatial derivatives.  We thus obtain
\begin{align*}
&  \|\HPDEu{w}{\FPDEu{v}}
  -\HPDEu{v}{\FPDEu{v}}
  \|_{\bar \delta,\mu,q-1}\\
&  \le 
 \bar  \delta^\rho \widehat C\Bigl(
 \|\SOHus{w}-\SOHus{v}\|_{\bar \delta,\zeta,q-1}
 + \|\STHus{w}-\STHus{v}\|_{\bar \delta,\zeta,q-1}\\
 & \hskip1.3cm + \|\NHus{w}-\NHus{v}\|_{\bar \delta,\zeta,q-1}
  \Bigr) 
 \end{align*} 
 for a constant $\widehat C$ which may depend on $s$, but not on the
 particular choice of $\bar \delta$. 
 
 Again using the hypothesis of
 \Theoremref{th:Wellposedness1stOrderFiniteDiff}, we see that for a
 choice of a possibly even smaller $\bar \delta$, which we now label
 $\widetilde \delta$, we can control this second term from the right hand
 side of the estimate for $\|\Gu{w}-\Gu{v}\|_{\widetilde \delta,\mu,q-1}$
 via a term of the form $\mathcal {C} \|w-v\|_{\widetilde \delta, \mu, q-1}$ for $\mathcal C\in (0,1)$. We thus determine that indeed the
 operator $\GuOp$ is a contraction mapping on $B_{\widetilde\delta,\mu,q-1,s}$ (for sufficiently small $\widetilde \delta$). It
 follows from standard arguments that the sequence $(w_N)$ has a unique limit
 $w$, contained in $B_{\widetilde \delta,\mu,q-1,s}$, which is a fixed
 point for $\GuOp$ and hence is a weak solution.
 
The sequence $(w_N)\subset B_{\widetilde \delta,\mu,q,s}$ is bounded
 in $X_{\widetilde \delta,\mu,q}$, but to this stage is known only to converge in
 $X_{\widetilde \delta,\mu,q-1}$ to $w$. This situation is similar to that encountered in
 the proof of \Propref{prop:linearexistencenonsmooth}. A similar argument involving \Corref{cor:rescueregularityT1} and the standard Cauchy problem of hyperbolic equations implies that $w$ is indeed an element of $X_{\widetilde \delta,\mu,q}$.

To show that $w$ is the remainder of a strong solution of the singular initial value problem,
 it remains for us to check that $w$ is differentiable in time. The
 definition of the sequence $(w_N)$ shows that for each integer $N$,
 $Dw_N$ exists and is contained in $X_{\widetilde \delta, \mu,q-1}$. Furthermore, this sequence converges in $X_{\widetilde \delta, \mu,q-1}$ by a similar argument as in the proof of \Propref{prop:linearexistenceregularitypre} using the \Conditionref{cond:LipschitzF}, and the positivity of $\zeta$. Since this convergence is uniform in time, it follows that $w$ is
 differentiable at each $t$ and that $Dw(t)$ is the limit of $(Dw_N(t))$ at
 each $t$. It follows from this limiting procedure that $u$ is indeed a strong solution to the singular initial value problem; similar arguments have been used before also in the proof of \Propref{prop:linearexistencenonsmooth}. 

The uniqueness of this solution $w$ follows from the
 uniqueness of the fixed point for the contraction mapping $\GuOp$.

In order to complete the proof of
 \Theoremref{th:Wellposedness1stOrderFiniteDiff}, we must consider the
 case $q=\infty$. We do this inductively in $q$. It is important to
 notice here that all of the constants in the previous estimates may depend
 on $q$ and hence we may have to adapt the choice of $\widetilde\delta$ in
 each induction step. It is thus  possible that the sequence of these constants
 $(\widetilde\delta_q)$ tends to zero as $q\rightarrow\infty$. To show that this possibility is avoided we use the result that any solution to the Cauchy problem for symmetric hyperbolic systems with a bounded first spatial derivative can be extended to a common time interval. Let us fix any $q \ge 3$.  \Theoremref{th:Wellposedness1stOrderFiniteDiff} (with finite $q$) shows that there exists a solution $w \in X_{\widetilde \delta, \mu, q}$ for some $\widetilde \delta \in (0, \delta]$. Let $t_{*} \in (0, \widetilde \delta]$, and consider the regular Cauchy problem with data $w(t_{*},x)$. Since $q \ge 2$, the Sobolev inequalities guarantee that the first spatial derivative is bounded on $[t_{*}, \widetilde \delta]$, and we may apply Proposition~1.5 in Chapter~16 of \cite{Taylor:2011wn} to show that there exists a $\widetilde \delta_{2} > \widetilde \delta$ such that the solution may be extended as a $H^q$-solution to $(0,\widetilde \delta_{2}]$. The same argument applied to any other value of $q\ge 3$ implies that the solution can be extended as $H^q$-solutions to the \textit{same} time interval $(0,\widetilde \delta_{2}]$. For $q=\infty$, we therefore find a unique solution $w$ on the same time interval in $X_{\widetilde\delta_2,\mu,\infty}$ (and hence $Dw\in X_{\widetilde\delta_2,\mu,\infty}$).

\subsection{Existence and uniqueness results based on ODE-leading-order terms}
\label{sec:expansionsHO}

Definition \ref{def:quasilinearlimit} of a quasilinear symmetric hyperbolic Fuchsian system, as well as the conditions which the singular initial value problem for such a system must satisfy if we wish to 
apply Theorem \ref{th:Wellposedness1stOrderFiniteDiff}
and thereby guarantee the existence of a solution, involve the specified leading-order term $u_0$ just as crucially as they involve the exponent vector $\mu$ and the functions $S_1, S_2$ and $N$ appearing in \Eqref{eq:1stordersystem}. In some applications, it is not easy to determine which choices of $u_0$ (if any) lead to these conditions being satisfied. Here we discuss an approach which starts with the choice of a leading-order term of a very restricted type (which we label ``ODE"-leading-order terms), and provides an alternate set of criteria for the existence of solutions to the singular initial value problem (with an ODE-leading-order term). This approach, presuming the criteria are satisfied, also systematically produces a sequence (possibly finite) of improved leading-order terms, which effectively serve as progressively higher order approximations to the solution of the singular initial value problem. We detail this approach here. 

As we see in Section \ref{sec:optimalexistence} below, the ODE-leading-order term approach is very useful in our analysis of the $T^2$--symmetric spacetimes. In particular, this approach plays a crucial role in our use of Fuchsian methods to obtain an optimal collection of $T^2$--symmetric solutions of Einstein's equations with AVTD behavior.

\subsubsection{The ODE-Fuchsian operator and ODE-leading-order terms}
We start by defining the differential operator $\LODEu{\cdot}$, which
plays a central role in carrying out this approach. Presuming that we
are working with a specified quasilinear symmetric hyperbolic Fuchsian
system \Eqref{eq:1stordersystem} with specified (as yet arbitrary)
leading-order term $u_0$ and with specified parameters $\delta$,
$\mu$, $q_0$ and $q$, we define the
\keyword{ODE-Fuchsian operator} as follows:
\begin{equation}
\label{LODE}
\LODEu{v}:=Dv+\SOLInvu\NLu v.
\end{equation}
Here we note that since (by Definition \ref{def:quasilinearlimit})
$S_{1,0}$ is invertible, it follows that $\LODEu{\cdot}$ is
well-defined. We also note that, since $\LODEu{\cdot}$ does not
involve any spatial derivatives, it is essentially a parametrized set
of ordinary differential operators (one for each point $x \in T^1$)
rather than a partial differential operator (hence the ``ODE"
label). 

Although not necessary yet at this stage, we assume $q\ge 3$ and $q_0=q+2$ (consistent with
\Theoremref{th:Wellposedness1stOrderFiniteDiff}) in all of what
follows. In particular this guarantees that all maps (including
their first spatial derivatives) are continuous with respect to $x$
(as follows from the Sobolev embedding theorem). For example, the ODE
operator above is well-defined at every spatial point $x$ under this
condition.

We wish to write the Fuchsian PDE system \Eqref{eq:1stordersystem} in terms of the operator $\LODEu{\cdot}$. Recalling the operational form $\LPDEu{w}{u_0+w}=\FPDEu{w}$ for \Eqref{eq:1stordersystem}, and noting that we can relate the operators $ \LPDEu{w}{\cdot}$ and $\LODEu{\cdot}$ as follows
\begin{equation}
  \label{eq:LPDE2ODE}
  \begin{split}
    \LPDEu{w}{v}=&\,\SOu{w}\LODEu{v}+\STu{w} t \partial_xv \\
    &+\SOu{w}\left(\SOInvu{w}\NNu{w}-\SOLInvu\NLu\right) v, 
  \end{split}
\end{equation}
we find that we can write out \Eqref{eq:1stordersystem} in the form
\begin{equation}
 \label{eq:equation}
\LODEu{u_0+w}=\FODEu{w},
\end{equation} 
if we define the term on the right hand side of \Eqref{eq:equation} as
\begin{equation}
  \label{eq:defFODE}
  \begin{split}
    \FODEu{w}
    :=&\,\SOInvu{w}\FPDEu{w}-\SOInvu{w}\STu{w} t \partial_x (u_0+w)\\
    &-\left(\SOInvu{w}\NNu{w}-\SOLInvu\NLu\right) (u_0+w).
  \end{split}
\end{equation}
The expression $ \FODEu{\cdot}$ is well-defined so long as $\delta$ is
sufficiently small so that $\SOu{w}$ is invertible for any given $w$
in $B_{\delta,\mu,q,s}$ for some $s>0$; we presume this is the case in
all of what follows. 

As noted above, a key aspect of this approach is the selection of a special class of leading-order terms.

\begin{definition}
\label{ODElead}
A leading-order term $u_0$ is an \keyword {ODE-leading-order term} if it satisfies the condition
\begin{equation}
  \label{eq:canonleadingterm}
  \LODEu{u_0}(t,x)=0.
\end{equation}
\end{definition}
Expanding out the expression for $\LODEu{\cdot}$, we see that an
ODE-leading-order term $u_0$ must satisfy $D u_0(t,x)+\widetilde N(x)
u_0(t,x)=0$, where $\widetilde N:=\SOLInvu\NLu$ is (by
definition) independent of  $t$.  For those very special cases in
which $\widetilde N$ is independent of $u_0$, \Eqref{eq:canonleadingterm}
is a parametrized set of linear ODEs, which can be readily solved for
$u_0$. More generally,  \Eqref{eq:canonleadingterm} is nonlinear and
therefore not so easy to analyze. We are interested here only in those
cases in which we can establish that solutions to
\Eqref{eq:canonleadingterm} exist. For those cases, we proceed to seek
solutions of the singular initial value problem for the system
\Eqref{eq:1stordersystem} with ODE-leading-order term $u_0$; we call
this the \keyword{ODE-singular initial value problem.} Observe that
the solution to an ODE-singular initial value problem, if obtained,
behaves in a way which suggests that as $t \rightarrow 0$, the spatial
derivative terms in \Eqref{eq:1stordersystem} become
negligible. Indeed, this is true for the AVTD solutions of the Einstein field equations 
which we treat in  \Sectionref{application} below.

It is useful for our work below to notice that under the assumption \Eqref{eq:canonleadingterm}, a combination of \Eqsref{eq:defFODE} and \eqref{eq:LPDE2ODE} yields
\begin{equation}
  \label{eq:FODEexprP0}
  \begin{split}
    \FODEu{w}=&\,\SOInvu{w}\FLu{w}
    -\SOInvu{w}\STu{w} t \partial_x w\\
    &-\left(\SOInvu{w}\NNu{w}-\SOLInvu\NLu\right) w,
  \end{split}
\end{equation}
where $\FLu{w}$ is defined in \Eqref{eq:defFLu}.

\subsubsection{(Order n)-leading-order terms}

We now use the ODE-leading-order term $u_0$ (presuming that it exists) and the ODE-Fuchsian
operator $\LODEu{\cdot}$ to generate a (possibly finite) sequence
$(u_n)$ of ``qualitative'' solutions (in a sense described shortly) to
the corresponding ODE-singular initial value problem; these $(u_n)$
play an important role in establishing a set of conditions which are
sufficient to show that this ODE-singular initial value problem does
admit a solution (see \Theoremref{th:Wellposedness1stOrderHigherOrder} below). 

We first consider the $x$-parametrized set of linear inhomogeneous ODEs
\begin{equation}
\label{eq:InhomogODE}
\LODEu{v}(t,x)=f_0(t,x),
\end{equation}
where $u_0$ is a fixed ODE-leading-order term and $f_0$ is a specified
inhomogeneity (whose regularity we discuss below). If we use
$W(t,x)$ to denote a fundamental matrix for the linear homogeneous equation
$\LODEu{v}(t,x)=0$, and if we let $(u_{*,1}(x),\ldots,u_{*,d}(x))$
represent free data for the initial value problem at $t_0\in
(0,\delta)$, then the general solution to \Eqref{eq:InhomogODE} may be
formally written as follows:
\begin{equation}
\label{eq:ODESoln}
v(t,x)=W(t,x)(u_{*,1}(x),\ldots,u_{*,n}(x))^T+W(t,x)\int_{0}^t s^{-1} W^{-1}(s,x) f_0(s,x) ds.
\end{equation}
We may then formally define the operator 
\begin{equation}
\label{eq:ODEop}
\HODEu{f_0}(t,x):=W(t,x)\int_{0}^t s^{-1} W^{-1}(s,x) f_0(s,x) ds,
\end{equation}
which, if it exists, maps a given source function $f_0$ to the
particular solution $w=\HODEu{f_0}$ of \Eqref{eq:InhomogODE} determined
by $(u_{*,1}(x),\ldots,u_{*,d}(x))=0$. We notice that the definition of this operator is invariant if the fundamental matrix $W$ is replaced by an equivalent fundamental matrix $W\mapsto W\cdot M$ for any invertible $d\times d$-matrix $M\in H^{q_0}$.

To proceed, we need to identify conditions which are sufficient for the existence of $\HODEu{\cdot}$. Noting that we may always choose the free data $(u_{*,1}(x),\ldots,u_{*,d}(x))$ in such a way that the first term in \Eqref{eq:ODESoln} equals (the already specified) $u_0$, we wish to also show that these same  conditions are sufficient to guarantee that the second term in \Eqref{eq:ODESoln}, i.e., $\HODEu{f_0}$, is higher order in time as $t$ approaches $0$ and therefore serves as a remainder term (in the sense of \Defref{def:SIVP}) for the singular initial value problem. We state the needed conditions in the following lemma, which is readily checked. 

\begin{lemma}[Existence and properties of $\HODEu{\cdot}$]
\label{HODElemma}
Suppose that a quasilinear symmetric hyperbolic Fuchsian system
\Eqref{eq:1stordersystem} has been chosen
(\Defref{def:quasilinearlimit}) for, in particular, a fixed
sufficiently small parameter $\delta$, for differentiability indices
$q\ge 3$ and $q_0=q+2$, and for an ODE-leading-order term $u_0$
(\Eqref{eq:canonleadingterm}). Suppose that $S_{1,0}^{-1}N_0$ is of
Jordan normal form. Then $\HODEu{\cdot}$ is well-defined on the domain
$X_{\delta,\widetilde\nu,q}$ for every smooth exponent vector $\widetilde\nu$, so
long as each component of $\widetilde\nu$ is strictly
larger
than the real part of the
negative of the corresponding diagonal element (eigenvalue) of
$S_{1,0}^{-1}N_0$. The target space of $\HODEu{\cdot}$ is
$X_{\delta,\widetilde\mu,q}$ for any smooth exponent vector
$\widetilde\mu<\widetilde\nu$, and one has the estimate
\begin{equation*}
\|\HODEu{f_0}\|_{\delta,\widetilde\mu,q}
\le C \delta^\rho \|f_0\|_{\delta,\widetilde\nu,q},
\end{equation*}
where $C>0$ depends only on the eigenvalues of $S_{1,0}^{-1}N_0$, on the dimension $d$ of the first-order system, on the
choices of $q$, $\widetilde\nu$, and on the difference between $\widetilde\nu$
and $\widetilde\mu$; the constant $\rho>0$ only depends on the difference
between $\widetilde\nu$ and $\widetilde\mu$, and on $q$.
\end{lemma}

In particular, we may choose $\widetilde\mu$ arbitrarily close to
$\widetilde\nu$, and we then check that the remainder $w=\HODEu{f_0}$,
(the second term in \Eqref{eq:ODESoln}) is of higher order in $t$ near $t=0$ (as measured by $\widetilde\mu$) as the order of $f_0$ (measured by $\widetilde\nu$) becomes large. As well, as the difference between $\widetilde\mu$ and $\widetilde\nu$ diminishes, one may have to choose the constant $C$ to be larger, and the constant $\rho$ to be smaller.

Some comments about this lemma are in order. For each quasilinear symmetric hyperbolic Fuchsian system and each choice of leading-order term $u_0$, there exists an invertible matrix $T\in H^{q_0}$ such that $T S_{1,0}^{-1}N_0 T^{-1}$ is in Jordan normal form. We assume in the following that such a transformation has been applied to the system, and hence that $S_{1,0}^{-1}N_0$ is in Jordan normal form. The fact that the matrices $S_1$ and $S_2$ in the principal part are in general not symmetric after such a transformation has been carried out is not important for the arguments that follow. Moreover, for simplicity we assume for each exponent vector here that those of its components which correspond rto the same Jordan block of $S_{1,0}^{-1}N_0$ have the same value.

As a consequence of this result, we may formally 
define the following  sequence:

\begin{definition}[(Order n)-leading-order sequence]
  \label{Ordern}
  Suppose $q\ge 3$.  With $w_0=0$, we formally set
  \begin{equation}
    \label{eq:defsequence}
    w_n:=\HODEu{\FODEu{w_{n-1}}},
  \end{equation}
  for all positive
  integers $1\le n\le q-2$.
   The \keyword{(order n)-leading-order
    terms}  are then defined by
  \begin{equation}
    \label{eq:Defun}
    u_n:=u_0+w_n,
  \end{equation}
  for $0\le n\le q-2$.
\end{definition}

To turn this formal specification of the (order-n)-leading-order-sequence into a definition, we need to state sufficient conditions for the composition on the right hand side of \Eqref{eq:defsequence} to be well-defined for each $n$. We do this in the proposition below. This proposition also proves that the sequence is characterized by certain properties which are relevant to the two roles which it plays: i) an increasingly accurate sequence of 
approximations to the solution of the singular initial value problem with
ODE-leading-order-term $u_0$ (presuming that such a solution exists); and ii) a sequence of ``new'', and ``better'' leading-order terms which can be used to define new singular initial value problems (closely tied to the original) for which we can prove the existence of solutions. The following proposition states the manner
in which the first use makes sense, and provides the first step
towards proving that the second use works.

\begin{proposition}[Existence and properties of the (order
  n)-leading-order terms $u_n$]
  \label{prop:propHO}
  Let $q\ge 3$ and $q_0=q+2$.
  Suppose that a quasilinear symmetric hyperbolic Fuchsian system
  \Eqref{eq:1stordersystem} has been chosen satisfying
  \Defref{def:quasilinearlimit} for an ODE-leading-order term $u_0$
  satisfying \Eqref{eq:canonleadingterm}, for  fixed parameters $\delta$ (sufficiently small) and $\mu$, and for         
  all differentiability
  indices $q'$ in the interval $[3,q]$.  Here we require that the exponent matrix $\zeta$, whose existence (in specifying  the function spaces containing $S_{1,1}, S_{2,1}$, and $N_1$)  is a necessary part of the definition of a quasilinear symmetric hyperbolic Fuchsian system (see Definition \ref{def:quasilinearlimit}), can be written as $\zeta_{ij}=\xi_i$ for some vector-valued exponent $\xi$ with strictly positive entries.  Suppose that the matrix
  $\SOLInvu\NLu$ is given in Jordan normal form, and suppose in
  addition that the following conditions are satisfied for all $\delta'\in (0,\delta]$ and all integers $q'\in (3,q]$:
  \begin{enumerate}[label=\textit{(\roman{*})}, ref=(\roman{*})]  
  \item \label{en:modBD} The remainder exponent vector $\mu$ is strictly larger than
    the negative of the corresponding diagonal elements (eigenvalues) of
    $S_{1,0}^{-1}N_0$ and satisfies the \keyword{modified block diagonality conditions}:
    For every $w\in B_{\delta',\mu,q',s}$, we have
    \[
    \RR{\mu} \SOu{w}=\SOu{w} \RR{\mu},\quad
    \RR{\mu} \NNu{w}=\NNu{w} \RR{\mu},
    \] 
    and, there exists $r>0$, so that
    \begin{equation}
      \label{eq:particularconditionS2}
      \RR{\mu}t \STu{w} \RR{-\mu}\in B_{\delta',\zeta,q',r},
    \end{equation}
    for $\zeta$ defined in terms of $\xi$, as above. 
   \item \label{en:cond1HOSource} There exists an exponent vector $\nu$
    with $\nu>\mu$ and a constant $r>0$, so that $\FLuOp$ maps $B_{\delta',\mu,q',s}$ into
    $B_{\delta',\nu,q',r}$.
  \item \label{en:cond4HO}%
    For all $w\in B_{\delta',\mu,q',s/2}$ and $\omega\in
    B_{\delta',\widehat\mu,q',s/2}$ for any exponent vector $\widehat\mu$ which
    satisfies $\widehat\mu\ge\mu$ and with respect to which the system is block diagonal, there exists a constant $r>0$ for which
    \begin{align*}
      \FLu{w+\omega}-\FLu{w}&\in B_{\delta',\widehat\mu+\nu-\mu,q',r},\\
      \SOu{w+\omega}-\SOu{w}&\in B_{\delta',\widehat\mu+\xi-\mu,q',r},\\
     \RR{\mu} t(\STu{w+\omega}-\STu{w})\RR{-\mu}&\in B_{\delta',\widehat\mu+\xi-\mu,q',r},\\
      \NNu{w+\omega}-\NNu{w}&\in B_{\delta',\widehat\mu+\xi-\mu,q',r}.
    \end{align*}    
    Moreover, there exists a constant $C>0$ such that the norm of each of these quantities can be bounded as  follows, 
  \[  \|\FLu{w+\omega}-\FLu{w}\|_{\delta',\widehat\mu+\nu-\mu,q'}
    \le  C \|\omega\|_{\delta',\widehat\mu,q'},\]
    with analogous inequalities holding for $S_1, S_2$, and $N$.

  \end{enumerate}
  Then, the sequence $u_n$ specified in \Defref{Ordern} is well-defined and, for
  some $\widetilde\delta\in (0,\delta]$ and constants $\gamma>0$, one has
  \begin{equation}
    \label{eq:propertyHO}
    u_n-u_0\in B_{\widetilde\delta,\mu,q-(n-1),s/2},\quad\qquad 
    u_{n}-u_{n-1}\in B_{\widetilde\delta,\mu+(n-1)\gamma,q-(n-1),s/2},
  \end{equation}
  for all $1\le n\le q-2$. Moreover, the
  \keyword{residual}\footnote{The residual of a function $u$ is defined so that  $u$ satisfies the
    system \Eqref{eq:1stordersystem} if and only if its residual
    vanishes; i.e., \Resu{u}=0. So the residual measures the accuracy of the sequence $(u_n)$ as approximate solutions to the system  \Eqref{eq:1stordersystem}}
  of $u_n$, defined by
  \begin{equation}
  \label{Res}
    \Resu{u_n}:=\LPDE{u_n}{u_n}-\f{u_n},
  \end{equation}
 is contained in  $X_{\delta,\mu+n\kappa\gamma,q-n}$.
\end{proposition}

We make a few remarks here concerning some of the details of this proposition. First, we observe that as a consequence of the
definition of this sequence, presuming that we start with an ODE-leading-order term $u_0$ of a certain order of differentiability, we find that  the first term of the sequence $u_1$ retains that regularity (up to order $q$), while the rest of the elements of the sequence ($u_2, u_3...$) generally do not. This is true because, since $w=0$ is smooth, it follows from the formula  \Eqref{eq:FODEexprP0} that $\FODEu{w_0}$ has $q$ derivatives; the same is then true for $w_1$. But then since  $\FODEuOp$ maps $X_{\delta,\mu,q}$ to
$X_{\delta,\widehat\nu,q-1}$, we see that $w_2$ has  only $q-1$ derivatives. The same loss of a derivative occurs for each successive element of the sequence. 

Secondly, we remark that the modified block diagonal conditions for $\mu$ are a slight generalization of \Defref{def:nonessentiallycouple}. In particular, it is not necessary that $\RR{\mu}$  commute with $S_2$ here. We have chosen a formulation of \Conditionref{en:modBD} 
 which applies directly our applications --- see \Sectionsref{sec:EPD}~and~\ref{application}. \Conditionref{en:modBD} can, however, be generalized to match more general situations. 
 
Thirdly, we note that \Conditionref{en:cond4HO} can be checked in our applications using the tools in the appendix; see \Sectionref{sec:productsfunctions}.

Finally, 
we note that for the special choice of  $\zeta_{ij}=\xi_i$ used here,  the space of matrix-valued functions $X_{\delta,\zeta,q}$ can be equivalently written as $X_{\delta,\xi,q}$. This latter space---a Banach space of matrix-valued functions $X_{\delta,\xi,q}$ with a vector-valued exponent $\xi$---is defined in essentially the same way as above in Eq. (2.3); the key difference is that the norm used to define this space  is the $H^q$-norm of $\RR{\xi}S$ - the matrix product of the matrix $\RR{\xi}$ formed from $\xi$ (see Eq. (2.2)) times $S$.  The motivation for specializing the matrix-valued exponent $\zeta$ in this way  and thence introducing the new notation $X_{\delta,\xi,q}$ for matrix-valued functions is to be able to express Condition (iii) of Proposition 2.18 in a natural way.

The proof of \Propref{prop:propHO} depends upon tight control of $\FODEu{\cdot}$ as given by \Eqref{eq:FODEexprP0}, and tight control of the inverse of $S_0$. We obtain this needed control using the following two lemmas.

\begin{lemma}
  \label{lem:S0matrix} 
  If the hypothesis for \Propref{prop:propHO} holds (presuming as usual that $\delta$ is sufficiently small so that $\SO{u_0+w}$ is invertible for all $w\in B_{\delta,\mu,q,s}$ and presuming that $\delta'$ and $q'$ satisfy the conditions stated in that hypothesis), it follows that for some $r>0$, $\SOInv{u_0+w}\in B_{\delta',0,q',r}$ for all $w\in B_{\delta',\mu,q',s}$.  Moreover, the operator given by
\[w\mapsto \SOInv{u_0+w} -(\SOLu)^{-1}\]
maps $B_{\delta',\mu,q',s}$ into $B_{\delta',\zeta,q',r}$ for some constant $r>0$, and this operator satisfies the 
difference condition 
\[\SOInvu{w+\omega}-\SOInvu{w}\in B_{\delta',\widehat\mu+\zeta-\mu,q',r}\]
for some exponent vector $\zeta>0$, and for some constant $C>0$ it satisfies the inequality 
  \[  \|\SOInvu{w+\omega}-\SOInvu{w}\|_{\delta',\widehat\mu+\zeta-\mu,q'}
    \le C \|\omega\|_{\delta',\widehat\mu,q'},\]
 for all $w\in B_{\delta',\mu,q',s/2}$ and all $\omega\in
    B_{\delta',\widehat\mu,q',s/2}$; here $\widehat\mu$ is any exponent vector which satisfies the inequality $\widehat\mu\ge\mu$, and for which the system is 
 block diagonal.
\end{lemma} 
The proof of this lemma relies on i) the fact that the inversion of invertible matrices is a smooth map, ii) the fact that both $\RR{\mu}$ and $\RR{\hat\mu}$ commute with $S_1$, and iii) Proposition~3.9 in Chapter~13 of \cite{Taylor:2011wn}; we omit the details here.

\begin{lemma}
  \label{lem:propertyFODE}
If the hypothesis for \Propref{prop:propHO} holds, then there exist positive constants $\gamma$ and $r$ so that for all constants $\delta'\in (0,\delta]$ and all integers $q'\in [3,q]$, and for all exponent vectors $\widehat\mu\ge\mu$ with respect to which the system is block diagonal, we have that $\FODEu{w}\in B_{\delta',\mu+\gamma,q'-1,r}$
and
  \begin{equation*}
    \FODEu{w+\omega}-\FODEu{w}\in B_{\delta',\widehat\mu+\gamma,q'-1,r}.
  \end{equation*} 
Further, there exists a constant $C>0$ such that for  all $w\in B_{\delta',\mu,q',s/2}$ and all  $\omega\in  B_{\delta',\widehat\mu,q',s/2}$, we have
  
\[  \|\FODEu{w+\omega}-\FODEu{w}\|_{\delta',\widehat\mu+\gamma,q'-1}
    \le C \|\omega\|_{\delta',\widehat\mu,q'}.\]

\end{lemma}

\begin{proof}[Proof  of Lemma \ref{lem:propertyFODE}]
  The first statement
is easily obtained by multiplying the expression \Eqref{eq:FODEexprP0} for $\FODEu{\cdot}$ by the quantity $\SOu{w}$ and then using the facts that $\RR{\mu}$ commutes with $\SOu{w}$, and  that $\SOu{w}$ is in $B_{\delta,0,q,r}$ for some $r>0$, and also applying  \Lemref{lem:productmatrix} and \Lemref{lem:S0matrix}.  To prove the rest, we multiply this same expression for $\FODEu{\cdot}$ by $\SOu{w}$, and then calculate

  \begin{equation}
    \label{eq:FODEDiff}
    \begin{split}
    \SOu{w}(&\FODEu{w+\omega}-\FODEu{w})\\
    =&\FLu{w+\omega}-\FLu{w}\\
    &-\left(\SOu{w+\omega}-\SOu{w}\right)\FLu{w+\omega}\\
    &-(\STu{w+\omega}-\STu{w})\, t\partial_x (w+\omega)\\
    &-\STu{w}\, t\partial_x \omega\\
    &-(\NNu{w+\omega}-\NNu{w}) \, (w+\omega)\\
    &-(\NNu{w}-\NLu)\omega\\
    &- \SOu{w}(\SOInvu{w}-\SOLInvu) \NLu \omega\\
    &-\left(\SOu{w+\omega}-\SOu{w}\right) \NLu (w+\omega).
  \end{split}
\end{equation}
Applying arguments of the sort used to verify \Propref{prop:propHO} and \Lemref{lem:S0matrix} together with the estimates included in the hypothesis, we obtain the conclusion.
\end{proof}

 We now proceed to prove \Propref{prop:propHO}:

\begin{proof}[Proof of \Propref{prop:propHO}]
  We
first show that the sequence $(w_n)$ (and the corresponding sequence $(u_n)$) is well-defined at least for finitely many sequence elements. It follows from \Conditionref{en:cond1HOSource} of the hypothesis that $\FLu{0}\in B_{\delta,\nu,q,r}$ for some $r>0$. Noting (see \Eqref{eq:FODEexprP0}) that $\FODEuOp$ evaluated at $0$ reduces to $S_{1,0}^{-1}\FLu{0}$, we infer from \Lemref{HODElemma} and \Lemref{lem:productmatrix} that the term
  $w_1$ is hence well-defined and is contained in $X_{\delta,\mu,q}$. It then follows from \Lemref{lem:propertyFODE} and \Lemref{HODElemma}  (whose hypotheses are satisfied) that the operator $\HODEu{\cdot}$ is well-defined, and consequently that $w_n$  is well-defined for all $2\le n\le q-2$. These functions are all elements of $X_{\delta,\mu,q-(n-1)}$. Using the estimate for the operator $\HODEu{\cdot}$ stated in \Lemref{HODElemma}, we verify that if we shrink the time interval $(0,\delta]$ to $(0,\widetilde\delta]$ as stated in the hypothesis of the Proposition under consideration, then we can show that finitely many of the sequence elements stay in a ball of fixed radius. We have thus verified  the first statement appearing in  \Eqref{eq:propertyHO}. Note that for convenience, in the remainder of this proof we continue to write $\delta$ instead of $\widetilde\delta$; however, we reserve the right to repeatedly shrink the time interval as necessary (a finite number of times).

We next argue by induction that the second statement in
  \Eqref{eq:propertyHO} holds. We presume that the differentiability
  index $q$ is sufficiently large so that there exist  nontrivial $n\le q-2$. To initialize the induction, we note that for $n=1$, this statement says that $u_1-u_0 \in B_{\delta,\mu,q,s/2}$. Noting that, by definition, $u_1-u_0=w_1=\HODEu{\FODEu{0}}$, and recalling from above that this term is contained in $B_{\delta,\mu,q,s/2}$, we verify the initialization. 
  
  To continue the induction argument, we suppose now that for some positive integer 
  $m<n$ there is an exponent vector $\mu^{(m)}\ge\mu$ such that
  $w_m-w_{m-1}\in B_{\delta,\mu^{(m)},q-(m-1),s/2}$, and further suppose
  that the same is true for all positive integers $i$ less than $m$
  (with corresponding exponent vectors $\mu^{(i)}$). Using the
  definition of $w_{i}$, together with \Lemref{HODElemma} and \Lemref{lem:propertyFODE}, we find that there exists a pair of exponent
  vectors $\mu\le\mu^{(m+1)} < \nu^{(m+1)}$ such that
\begin{align*}
  \|w_{m+1}-w_{m}\|_{\delta,\mu^{(m+1)},q-m}
  &=\|\HODEu{\FODEu{w_{m}}-\FODEu{w_{m-1}}}\|_{\delta,\mu^{(m+1)},q-m}\\
  &\le C\delta^\rho \|\FODEu{w_{m}}-\FODEu{w_{m-1}}\|_{\delta,\nu^{(m+1)},q-m}\\
  &\le C\delta^\rho  C\|w_{m}-w_{m-1}\|_{\delta,\nu^{(m+1)}-\gamma,q-(m-1)}.
\end{align*}
In carrying out this calculation (with $\gamma$ being the quantity hypothesized in  \Lemref{lem:propertyFODE}), we note that 
the operator $\HODEu{\cdot}$ is well-defined here according to
\Lemref{HODElemma} and \Lemref{lem:propertyFODE} since $w_m-w_{m-1}\in
B_{\delta,\mu^{(m)},q-(m-1),s/2}$ with $\mu^{(m)}\ge\mu$.
Finally we note  that the constants $C$ and
$\rho$ may depend in particular on $q$ and $m$, but this dependence is not a
problem for carrying out our argument since we are only interested in finitely many sequence elements.  

To complete the induction argument, we verify that since we have assumed (as part of the induction) that
$w_m-w_{m-1}\in B_{\delta,\mu^{(m)},q-(m-1),s/2}$, it follows that so long as $\nu^{(m+1)}-\gamma<\mu^{(m)}$ holds, we have the final right hand side of the above inequality finite. Therefore the initial left hand side must be finite, and this holds for any $\mu^{(m+1)}$, so long as $\mu^{(m+1)} <\mu^{(m)} +\gamma$. We satisfy these conditions by choosing 
$\mu^{(m+1)}=\mu+m\kappa\gamma$ for any $\kappa < 1$. Noting that this is the case for all $m$, with $\kappa$ chosen independently  of $m$, we conclude that
 \Eqref{eq:propertyHO} holds, after having identified $\kappa\gamma$ with $\gamma$ to simplify the notation.

It remains to verify that \Eqref{Res} holds for the residuals of the sequence $(u_n)$.
 Using
\Eqsref{eq:equation}, \eqref{eq:canonleadingterm} and
\eqref{eq:defsequence} we  calculate
\begin{align*}
  \Resu{u_n}&=\LPDEu{w_n}{u_0+w_n}-\FPDEu{w_n}\\
  &=-\SOu{w_n}\left(\FODEu{w_{n}}-\FODEu{w_{n-1}}\right).
\end{align*}
Since $w_{n}-w_{n-1}\in B_{\delta,\mu+(n-1)\gamma,q-n+1,s/2}$, it
follows from \Lemref{lem:propertyFODE} that
\begin{equation*}
\Resu{u_n}\in X_{\delta,\mu+n\gamma,q-n}.
\end{equation*}
\end{proof}

\subsubsection{(Order n)-singular initial value problem}
Proposition \ref{prop:propHO} shows that, so long as we can find an
ODE-leading-order term $u_0$ and so long as certain conditions hold,
the difference $u_{n+1}-u_n$ behaves like a power of $t$ near
$t=0$, with this power increasing monotonically with $n$.  It is hence
meaningful to consider, in addition to the ODE-singular initial value
problem with leading-order term $u_0$, a sequence of (order
n)-singular initial value problems which use (order n)-leading-order
terms $u_n$ ($n\le q-2$). In view of the relationship between $u_0$
and the sequence $(u_n)$, we may write the same solution $u$ of a
given singular initial value problem either in the form $u=u_0 +w$ for
a remainder $w$ in $X_{\delta,\mu,q}$, or, as $u=u_n+\omega$ for a
remainder $\omega$ in $X_{\delta,\widehat\mu,q}$ with $\widehat\mu$ increasing
suitably with $n$. The same can be done for any of the  $u_m$ ($m\le n$) in the (order n)-leading-order term sequence.

We now use the (order n)-leading-order terms to argue that, at least
for the smooth case ($q=\infty$), if the conditions of Proposition
\ref{prop:propHO} are met, the ODE-singular initial value problem, and 
correspondingly the (order n)-singular initial value problem have
solutions.

\begin{theorem}[Existence and uniqueness for the ODE-singular initial value
  problem]
  \label{th:Wellposedness1stOrderHigherOrder}
  Suppose that a quasilinear symmetric hyperbolic Fuchsian system with
  ODE-leading-order term $u_0$ has been chosen which satisfies the
  hypotheses of \Propref{prop:propHO} for all finite values of (differentiability index) $q$.
  Then, for some sufficiently small $\delta_1 \in
  (0, \delta]$ and for a sufficiently large $n$, there exists a unique
  solution $u$ of \Eqref{eq:1stordersystem} with $u-u_n, D(u-u_n)\in
  X_{\delta_1,\mu+n \gamma,\infty}$, where $u_n$ is the (order n)-leading-order
  term defined in Definition~\ref{Ordern} for this system. This
  solution   $u$ is also the only solution of the ODE-singular initial value problem
  with $u-u_0\in
  X_{\delta_1,\mu,\infty}$.
\end{theorem}

This result states conditions which are sufficient for the ODE-singular initial value problem
(with leading-order term $u_0$) to admit (unique) solutions.
In doing so, \Theoremref{th:Wellposedness1stOrderHigherOrder} provides a potentially very useful alternative to \Theoremref{th:Wellposedness1stOrderFiniteDiff} of Section~\ref{sec:firstordertheory}. 
Observe in particular that the hypothesis for \Theoremref{th:Wellposedness1stOrderHigherOrder} does \textit{not} require that the energy dissipation matrix be positive definite with respect to $\mu$.

Here we state and prove 
\Theoremref{th:Wellposedness1stOrderHigherOrder} only for the infinite
differentiability case ($q=\infty$). This smoothness restriction plays a role in the proof, since it  allows one to always choose $n$ large enough so that \Conditionref{en:cond2} of
\Theoremref{th:Wellposedness1stOrderFiniteDiff} for the singular
initial value problem with leading-order term $u_n$ is satisfied. If one tries to prove a result like \Theoremref{th:Wellposedness1stOrderHigherOrder} for finite differentiability order, then there is an upper bound for the possible choice of $n$, and consequently one may not be able to choose it large enough to satisfy \Conditionref{en:cond2}. However, in certain circumstances, a large but finite order of differentiability is in fact sufficient to carry through the proof.

\begin{proof}
The basic idea of the proof is to reformulate the system using
$u_n$ for the leading-order term in place of $u_0$, and then verify
that the hypothesis of Theorem
\ref{th:Wellposedness1stOrderFiniteDiff} (in the case $q=\infty$) is
satisfied if $n$ is chosen sufficiently large. To carry this through, we first argue that the system \Eqref{eq:1stordersystem}, which for the ODE-singular initial value problem can be written as
\begin{equation} 
\label{eqn1}
0=\LPDEu{w}{u_0+w}-\FPDEu{w}, 
\end{equation}
can also be written as 
\begin{equation}
\label{eqn2}
0=\LPDEun{\omega}{\omega}-\FLun{\omega},
\end{equation}
where we recall the definition \Eqref{eq:DefLPDE} of the principal
part operator $\widehat L$ and the definition \Eqref{eq:defFLu}
for the operator $\FLun{\cdot}$. Here we use $w$ for the remainder term
corresponding to $u_0$ and we use $\omega$ for the remainder term
corresponding to $u_n$ (hence $u_0 + w= u_n +\omega$). To show this equivalence, we note the relations 
$ \LPDEun{\omega}{v}= \LPDEu{w}{v}$ and $\FPDEun{\omega}=\FPDEu{w}$, and then using these we calculate
\begin{align*}
  0&=\LPDEu{w}{u_0+w}-\FPDEu{w}
   =\LPDEun{\omega}{u_n+\omega}-\FPDEun{\omega}\\
  &=\LPDEun{\omega}{u_n}+\LPDEun{\omega}{\omega}-\FPDEun{\omega}\\
   &=\LPDEun{\omega}{\omega}-\FLun{\omega}.
\end{align*}
The equivalence of \Eqref{eqn1} and \Eqref{eqn2} immediately follows.

We now choose a sequence of exponent vectors $\widetilde\mu^{(n)}$ 
which satisfy
\begin{equation}
  \label{eq:muinequality}
  \mu<\widetilde\mu^{(n)}<\mu+(n-1)\gamma,
\end{equation}
and which are consistent with the block diagonal condition for
\Eqref{eqn2}; we note that this is possible for all  sufficiently large
integers $n$.
Examining the singular initial value problem corresponding to \Eqref{eqn2}, we verify that for any given sufficiently large integer $n$, this PDE system, 
together with $u_n$ as leading order term and exponent vector $\widetilde\mu^{(n)}$, satisfies the conditions to be a quasilinear symmetric hyperbolic Fuchsian system. We also verify, based on \Eqref{eq:energydissipationmatrix}, that for sufficiently large $n$ (and therefore sufficiently large $\mu+(n-1)\gamma$) the exponent vectors $\widetilde\mu^{(n)}$ can be chosen large enough to guarantee that the energy dissipation matrix $M_0$ is positive definite. Consequently, this system satisfies \Conditionref{en:cond2} of the hypothesis of \Theoremref{th:Wellposedness1stOrderFiniteDiff}.

To check that Conditions~\ref{en:cond4N}~and~\ref{cond:LipschitzF} of
\Theoremref{th:Wellposedness1stOrderFiniteDiff} are also satisfied, we examine the
operator $\FLunOp$.  Using 
\Eqsref{eq:LPDE2ODE} and \eqref{eq:FODEexprP0} together with  \Defref{Ordern}, we
calculate
\begin{align*}
  \FLun{\omega}=&\,\SOu{w}\left(\FODEu{w}-\FODEu{w_{n-1}}\right) +\STu{w} t \partial_x \omega
\\&
    +\SOu{w}\Bigl(  \SOInvu{w}(\NNu{w}-\NLu)\\
      &\qquad\qquad\qquad\quad 
      +\bigl(\SOInvu{w}-\SOLInvu\bigr) \NLu\Bigr) \omega.
\end{align*}

By comparing the first line of this expression with \Eqref{eq:FODEDiff}, we notice that all spatial derivative terms cancel; hence there is no loss of regularity in this expression as is the case for the operator $\FODEuOp$ itself. We therefore get estimates analogous to those in
 \Lemref{lem:propertyFODE}, with $q-1$ replaced by $q$. 

Combining the assumptions for $S_1$, $S_2$ and $N$ which are stated in the hypothesis of \Propref{prop:propHO} (and therefore included in the hypothesis of Theorem \ref{th:Wellposedness1stOrderHigherOrder}) with  the upper bound stated in \Eqref{eq:muinequality}, we readily show that all of the conditions
of \Theoremref{th:Wellposedness1stOrderFiniteDiff} hold for $q=\infty$.
Then the consequent application of \Theoremref{th:Wellposedness1stOrderFiniteDiff} shows  that so long as $n$ is sufficiently large, there exists exactly one solution $u=u_n+\omega$ with $\omega\in X_{\delta_1,\widetilde\mu^{(n)},\infty}$.  

We wish to show next that for such a fixed chosen value of $n$, in fact $\omega\in X_{\delta_1,\mu+n\gamma,\infty}$. To show this, we consider an integer $n_+$ which is large enough so that $\widetilde\mu^{(n_+)}>\mu+(n-1)\gamma$. Applying the same argument as above, but now with $n_+$ instead of $n$ (and hence using $u_{n_+}$ as the leading-order term), we obtain a solution $\widetilde u=u_{n_+}+\widetilde\omega$ which has the property that $\widetilde\omega\in X_{\delta_1,\widetilde\mu^{(n_+)},\infty}$. Uniqueness of the singular initial value problem with respect to $u_n$ implies that
$\widetilde u$ equals $u$. Moreover, we have 
\[\omega=w_{n_+}-w_n+\widetilde
\omega=(w_{n+1}-w_n)+\ldots+(w_{n_+}-w_{n_+ -1})+\widetilde\omega.\]
Given that $w_{n+1}-w_n\in X_{\delta_1,\mu+n\gamma,\infty}$, we obtain the desired result
\[\omega\in X_{\delta_1,\mu+n\gamma,\infty}.\]

To conclude the proof of this theorem,
we must show that any solution $\widehat u$ of the form $\widehat
u=u_0+\widehat w$ with $\widehat w\in X_{\delta_1,\mu,\infty}$ must equal the solution
$u$. To show this, it is useful to write $\widehat u=u_n+\widehat w-w_n$, where $u_n$ is defined by \Eqref{eq:Defun} and $w_n$ is defined by \Eqref{eq:defsequence}. Then if 
 we can verify that $\widehat w-w_n\in
X_{\delta_1,\widetilde\mu^{(n)},\infty}$, it follows from uniqueness
that $\omega=\widehat w-w_n$, and hence that
$\widehat u=u$. We make this verification by using induction to show that, in fact, 
 $\widehat w-w_m\in X_{\delta_1,\mu+m  \gamma,\infty}$ holds for every non-negative integer $m$. In the case
$m=0$, we have $\widehat w-w_0=\widehat w\in X_{\delta_1,\mu,\infty}$ which implies the
claim for $m=0$. Suppose the claim has been shown for $m=m_0\ge 1$. We
know that $\widehat w$ is a solution of the equation
\[\LODEu{\widehat w}=\FODEu{\widehat w},\]
while $w_{m_0+1}$ is a solution of
\[\LODEu{w_{m_0+1}}=\FODEu{w_{m_0}}.\]
Taking the difference, we obtain
\[\LODEu{\widehat w-w_{m_0+1}}=\FODEu{\widehat w}-\FODEu{w_{m_0}}.\]
We can write this formally as
\[\widehat w-w_{m_0+1}=\HODEu{\FODEu{\widehat w}-\FODEu{w_{m_0}}}.\]
Now, the fact that $w-w_{m_0}\in
X_{\delta_1,\mu+m_0\gamma,\infty}$ implies that $\FODEu{w}-\FODEu{w_{m_0}}\in
X_{\delta_1,\mu+(m_0+1)\gamma,\infty}$
(\Lemref{lem:propertyFODE}). Consequently (see \Lemref{HODElemma}), the operator $\HODEu{\cdot}$ is
well-defined. This completes the proof.
\end{proof}


\subsubsection{An example: the Euler--Poisson--Darboux equation}
\label{sec:EPD}

We consider now the example of the {Euler-Poisson-Darboux equation} (see also \cite{Ames:2012tm} for another example)
\begin{equation}
  \label{eq:EPD}
  D^2 u(t,x)-t^2 u_{xx}(t,x)=f_0(t,x).
\end{equation}
Here, $u(t,x)$ is the unknown (assumed to be a scalar function), and $f_0(t,x)$ is a specified scalar function. The Euler-Poisson-Darboux equation is second order, and in previous work by two of the authors \cite{Beyer:2010fo} on semilinear second-order Fuchsian systems, it has been shown that this equation admits unique solutions to the singular initial value problem with leading-order term 
\begin{equation}
 \label{eq:2ndorderLOT}
 u_0(t,x)=u_*(x)\log t+u_{**}(x), 
\end{equation} 
(for arbitrary functions  $u_*$ and $u_{**}$) so long as $f_0=O(t^{\widehat\nu})$ with $\widehat\nu>0$. We seek to show that we obtain these same results using the first-order methods which we have developed here. In particular, this example demonstrates the usefulness of the techniques  discussed in \Sectionref{sec:expansionsHO}, thereby serving as a linear warmup example with  which we can  explore some of the issues which arise below in our discussion of the application of these methods to the fully nonlinear $T^2$--symmetric Einstein's vacuum equations in \Sectionref{application}.

To apply the first-order  theory developed in this paper, we first convert this 
equation into a first-order system by setting
\begin{equation}
\label{ICs}
u_1:=u,\quad u_2:=Du,\quad u_3:=t\partial_x u,\quad U:=(u_1,u_2,u_3)^T.
\end{equation}
\Eqref{eq:EPD} then takes the form of a first-order \textit{evolution system}
\begin{equation}
\label{EPDeq}
  S_1 D U+S_2 t\partial_x U+N U=f,
\end{equation}
with
\begin{equation*}
  S_1=\diag(1,1,1),\quad S_2=
  \begin{pmatrix}
    0 & 0 & 0\\
    0 & 0 & -1\\
    0 & -1 & 0
  \end{pmatrix}, \quad
  N=
  \begin{pmatrix}
    0 & -1 & 0\\
    0 & 0 & 0\\
    0 & 0 & -1
  \end{pmatrix},\quad
  f=
  \begin{pmatrix}
    0\\f_0\\0
  \end{pmatrix},
\end{equation*}
plus a \textit{constraint equation}
\begin{equation}
\label{EPDCon}
\Delta_u:=u_3/t-\partial_x u_1=0.
\end{equation}
Observe that in working with the Euler-Poisson-Darboux system in this first-order form, one first treats the components $u_1, u_2$, and $u_3$ as independent functions whose evolution is determined by \Eqref{EPDeq}. This means that we solve the singular initial value problem of this system with respect to a leading-order term motivated by \Eqref{eq:2ndorderLOT}. Then, in a second step, we identify $u_1$ with the original variable $u$ and consider the two remaining relations \Eqref{ICs} as constraints: the one involving the time derivative is automatically implied by the first of \Eqsref{EPDeq} (the evolution equation for $u_1$), while the one involving the spatial derivative gives rise to the condition $\Delta_u\equiv 0$ in \Eqref{EPDCon}. Let us start with the first step.

One readily verifies that this evolution system is of (quasilinear) symmetric hyperbolic Fuchsian form for any choice of leading-order term, and hence our theory can, in principle, be applied. Our approach is to find a leading-order term for the first-order variables which is consistent with \Eqref{eq:2ndorderLOT} and which, in addition, is an ODE-leading-order term. We easily determine that  the general solution to \Eqref{eq:canonleadingterm} for \Eqref{EPDeq} takes the  form
\begin{equation}
\label{EPDODETerm}
U_0=(C_1+C_2\log t, C_2, C_3 t)^T,
\end{equation}
for the spatially-dependent
parameters $C_1(x)$, $C_2(x)$ and $C_3(x)$. However, we see that this leading-order term can only be consistent with \Eqsref{eq:2ndorderLOT} and \eqref{ICs} in the special case $u_*=0$ and $C_2=0$. Hence, this approach for finding a leading-order term fails.

We circumvent this problem as follows. For a specified function $u_*$ (which is at least second order differentiable; we specify its necessary regularity more precisely below), we define 
\begin{equation}
  \label{eq:rescu}
  \widehat u:=u-u_*(x)\log t,
\end{equation}
and work with the evolution equation for $\widehat u$ rather than that for $u$. Substituting \Eqref{eq:rescu} into \Eqref{eq:EPD}, we obtain
\[D^2 \widehat u-t^2 \widehat u_{xx}=t^2\log t u_*''+ f_0(t,x),\]
where $u_*''$ indicates the second derivative of $u_*$. 
Now, setting
\[\widehat u_1:=\widehat u,\quad \quad
\widehat u_2:=D\widehat u,\quad \quad
\widehat u_3:=t\partial_x \widehat u,\quad\quad  
\widehat U:=(\widehat u_1,\widehat u_2,\widehat u_3)^T,\]
we obtain the evolution equation
\begin{equation}
\label{EPDhat}
  S_1 D \widehat U+S_2 t\partial_x \widehat U+N \widehat U=\widehat f,
\end{equation}
for the same matrices $S_1$, $S_2$, $N$ as above, but with
\[\widehat f=\bigl(0,f_0+t^2\log t\, u_*'',0\bigr)^T.\]
In terms of $\widehat u$, the constraint \Eqref{EPDCon} takes the form
\begin{equation}
  \label{eq:constraintviolation}
  \Delta_{\widehat u}:=\widehat u_3/t-\partial_x \widehat u_1=0.
\end{equation}
Choosing the ODE-leading-order term for the $\widehat U$ formulation to be of the same form as \Eqref{EPDODETerm}, we have 
\begin{equation*}
\widehat U_0=(C_1+C_2\log t, C_2, C_3 t)^T,
\end{equation*}
but now (in view of \Eqref{eq:2ndorderLOT}) we are led to 
choose the parameter functions in the form  $C_1=u_{**}$, $C_2=0$ and $C_3=u_{**}'$; hence 
\begin{equation}
  \label{eq:1stLOTEPD}
  \widehat U_0(t,x)=\bigl(u_{**}(x),0,t u_{**}'(x)\bigr)^T.
\end{equation}
The function $u_*$ appearing in \Eqref{eq:rescu} together with  the function $u_{**}$ introduced here together comprise the full range of free data  suggested by \Eqref{eq:2ndorderLOT}. Both play the role of asymptotic data functions.

Having found a suitable representation of the equations and the leading-order term, we write 
 the unknown $\widehat U$ of the evolution system as $\widehat U=\widehat U_0+W$, and look for  sufficient conditions for the existence of solutions to the singular initial value problem in this form, with $W$ as a remainder term. To enforce the remainder falloff properties, we 
 choose an exponent vector $\mu=(\mu_1,\mu_2,\mu_3)$ and, in view of \Eqref{eq:1stLOTEPD}, we require that $\mu_1,\mu_2>0$ and $\mu_3>1$. 

We first seek to prove existence of solutions using 
 \Theoremref{th:Wellposedness1stOrderFiniteDiff}.  To satisfy the  block diagonality condition of \Theoremref{th:Wellposedness1stOrderFiniteDiff} we must set $\mu_1=\mu_2=\mu_3$. We therefore simplify the notation by writing the exponent vector as $(\mu,\mu,\mu)$ for some smooth scalar function $\mu$ which, from above considerations, must be greater than one. Observe here that, while this equality of all components of the exponent vector is necessary to satisfy the hypothesis of  \Theoremref{th:Wellposedness1stOrderFiniteDiff}, it does appear to be an artificial restriction. Under reasonable regularity assumptions, we might rather expect that  if the first and second
components are $O(t^\mu)$, then the third component of $W$ should be
$O(t^{\mu+1}\log t)$; the $\log t$ factor may arise from derivatives of $t^\mu$ since  $\mu$ is generally not constant. 
In any case, we readily verify that the energy dissipation matrix 
\[M_0=\begin{pmatrix}
\mu & -1 & 0\\
0 & \mu & 0\\
0 & 0 & \mu-1
\end{pmatrix},
\] 
is positive definite so long as $\mu>1$. 

Calculating 
\begin{equation*}
  \FL{\widehat U_0}{W}
  =\FPDE{\widehat U_0}{W}-\LPDE{\widehat U_0+W}{U_0}
  =\bigl(0,f_0+t^2(\log t\, u_*''+u_{**}''),0\bigr)^T,
\end{equation*}
 we now suppose that $W\in X_{\delta,(\mu,\mu,\mu),q}$ and $f_0\in X_{\delta,\widehat\nu,q}$ for $\widehat\nu>1$. Then $\FL{\widehat  U_0}{W}\in X_{\delta,(\nu,\nu,\nu),q}$  if $u_*,u_{**}\in H^{q+2}(T^1)$, where $\nu=\widehat\nu$, if $\widehat\nu<2$, or, we have $\nu<2$, if $\widehat\nu\ge 2$. Choosing $q\ge 3$, we verify that \Theoremref{th:Wellposedness1stOrderFiniteDiff}  implies the existence of solutions of the evolution system $\widehat U=\widehat U_0+W$ with $W\in X_{\delta_1,(\mu,\mu,\mu),q}$ for $\delta_1$
sufficiently small\footnote{In fact, since the equations are linear, we can extend the solution to all positive times $t>0$.} and for an exponent $\mu\in (1, \min\{ 2, \widehat\nu \})$. For any specified set of the asymptotic data
$u_*$ and $u_{**}$, we find that the solution is unique for remainders in the  space $X_{\delta_1,(\mu,\mu,\mu),q}$.

Given any such solution of the first-order evolution system, our  next step is to identify $\widehat u_1$ with $u-\log(t) u_*$ and then, \emph{if} the remaining constraint $\Delta_{\widehat u}\equiv 0$ is satisfied, to conclude that $u$ is actually a solution of the original second-order equation \Eqref{eq:EPD} with leading-order term $u_0=u_*\log t+u_{**}$ and with remainder $w=w_1$ (the first component of the vector $W$) in $X_{\delta_1,\mu,q}$.
To determine  if the constraint is satisfied, we use the evolution equation \Eqref{EPDhat} to calculate the time derivative of the constraint violation term $\Delta_{\widehat u}$, obtaining 
\begin{equation}
\label{EPDtime}
D\Delta_{\widehat u}=0.
\end{equation}

We then note that (i) if we construct  $\Delta_{\widehat u}$ using $\widehat U_0$ from \Eqref{eq:1stLOTEPD},  we get $\Delta_{\widehat U_0}=0$; and (ii) if we combine the evolution equation \Eqref{EPDtime} with the leading order term $\Delta_{\widehat U_0}$ as well as $q\ge 3$ and other appropriate choices of $\mu$, etc., then we find that $\Delta_{\widehat u}$ satisfies a singular initial value problem which satisfies the hypothesis of \Theoremref{th:Wellposedness1stOrderFiniteDiff}. Noting that $\Delta_{\widehat u}=0$ is a solution to this singular initial value problem, and recalling that \Theoremref{th:Wellposedness1stOrderFiniteDiff} implies that solutions are unique, we see that indeed, the constraint  $\Delta_{\widehat u}=0 $ must be satisfied.

While this approach to analyzing  the singular initial value problem for the Euler-Poisson-Darboux system does produce a solution,
it is unsatisfactory
for two reasons. First, it does not allow us to treat the case in which $f_0\in X_{\delta,\widehat\nu,q}$ for $\widehat\nu<1$. Second, if $\widehat\nu>1$, this approach does not exclude the possible existence of other solutions $u$ with remainders $w$ in $X_{\delta,\mu,q}$ for $\mu<1$. Both of these issues are resolved if we use an alternative approach based on \Theoremref{th:Wellposedness1stOrderHigherOrder} and the use of (order n)-leading order terms. In doing this, we pay a price in that we must require a that the spatial derivative parameter $q$ is infinite. 

If we are to work with \Theoremref{th:Wellposedness1stOrderHigherOrder}, a key requirement is that we start with an ODE-leading-order term; we have already fulfilled this requirement by our choice of $\widehat U_0.$ We now have the advantage that we do not need to impose the block diagonal condition, but only the modified block diagonal conditions, see \Conditionref{en:modBD} in \Propref{prop:propHO}, and also not the positivity of the energy dissipation matrix in choosing the remainder exponent vector $\mu$; we may work with $\mu=(\mu_1,\mu_2,\mu_3)$ for \emph{any} $\mu_1,\mu_2>0$ and $1<\mu_3<\mu_2+1$, thereby permitting the full range of values of $\mu$ for which the singular initial value problem is meaningful. Notice that the upper bound for $\mu_3$ is implied by \Eqref{eq:particularconditionS2} and is related to the observation above that a spatial derivative of a spatially dependent power of $t$ may introduce additional $\log t$-terms.
 Proceeding, we suppose that we have chosen some $f_0\in X_{\delta,\widehat\nu,\infty}$ with $\widehat\nu>0$. Any choice of $\mu$ satisfying the above conditions is consistent with \Conditionref{en:cond1HOSource} of \Propref{prop:propHO} (as part of \Theoremref{th:Wellposedness1stOrderHigherOrder}) if $\mu_2<\min\{2,\widehat\nu\}$. Choosing $u_*,u_{**}\in C^\infty(T^1)$, we then verify straightforwardly  that \Conditionref{en:cond4HO} of \Propref{prop:propHO} is satisfied. We conclude that there exists a solution $\widehat U$ of the evolution system with $\widehat U-\widehat U_n\in X_{\delta_1,(\mu_1,\mu_2,\mu_3)+n\gamma,\infty}$ for some constants $\delta_1>0$ and a sufficiently large integer $n$. This solution is unique, with the remainder $\widehat U-\widehat U_0$ contained in
$ X_{\delta_1,(\mu_1,\mu_2,\mu_3),\infty}$.

Having verified the existence of solutions to
the first-order evolution system, we wish to show again that the corresponding solution is actually a solution of the original second-order equation by considering the constraint \Eqref{eq:constraintviolation}. This can be done essentially as discussed above.

To illustrate the use of the leading-order term approach, we choose the source term in the form
 $f_0(t,x)=f_*(x) t^{1/2}$ for a  smooth function $f_*$, and calculate
\[\widehat U_1=\widehat U_0+
\begin{pmatrix}
  4f_*(x) t^{1/2}+\frac 14 t^2\left(\log t u_*''+u_{**}''-u_*''\right)\\
  2f_*(x) t^{1/2}+\frac 12 t^2\left(\log t u_*''+2u_{**}''-u_*''\right)\\
  0
\end{pmatrix}
\]
and
\[\widehat U_2=\widehat U_1+
\begin{pmatrix}
  0\\0\\
  4 t^{3/2}\log t f_*'+\frac 14 t^3\left(\log t u_*^{(3)}+u_{**}^{(3)}-u_*^{(3)}\right)
\end{pmatrix}.
\]
One may continue to calculate the sequence, and one verifies (in accord with the last statement in Proposition \ref{prop:propHO}) that the residuals corresponding to this sequence are contained in $X$ spaces of monotonically increasing exponent. 

\section
{\texorpdfstring{$T^2$}{T2}--symmetric vacuum Einstein spacetimes}
\label{application}

\subsection{Objective of this section} 

 As noted in the Introduction, one of the main motivations for this work is to explore  the singular regions of certain classes of solutions of the Einstein gravitational field equations. In particular, as a step towards studying the strong cosmic censorship conjecture in families of solutions characterized by relatively large isometry groups, we use the Fuchsian formulations developed here to show that there are large sets of solutions in these families which exhibit AVTD behavior in a neighborhood of their singularity.

We work here with spacetimes which are characterized by a spatially-acting $T^2$ isometry group, but do not have the further restriction of a non-vanishing ``twist", which defines the familiar Gowdy spacetimes. Following convention, we refer to them as the ``$T^2$--symmetric spacetimes"; if they also satisfy the Einstein equations, we call them ``$T^2$--symmetric solutions". While much is known regarding the Gowdy spacetimes, including a proof that strong cosmic censorship holds for the Gowdy spacetimes with $T^3$ spatial topology \cite{Ringstrom:2009ji} and for polarized Gowdy spacetimes with any allowed spatial topology \cite{Chrusciel:1999dk}, much less is known about the $T^2$--symmetric solutions. For both the Gowdy and $T^2$--symmetric families, the presence of the $T^2$ isometry effectively reduces the analysis  to that of a PDE system on a $1+1$ dimensional manifold.  One notable difference, however, is that while the Gowdy PDE system is  semilinear,  that of the $T^2$--symmetric solutions is  quasilinear.

The first work showing that there are (non-polarized) Gowdy spacetimes with AVTD behavior  is that of Kichenassamy and Rendall \cite{Kichenassamy:1999kg} which uses Fuchsian methods to show that this is true for analytic Gowdy solutions on $T^3$. The later work of Rendall \cite{Rendall:2000ki} shows this for Gowdy spacetimes which are smooth, again using Fuchsian methods (adapted to smooth solutions rather than analytic solutions). Fuchsian methods have been used \cite{Isenberg:1999ba}  to verify that there are analytic polarized $T^2$--symmetric solutions with AVTD behavior. Here, we use the results presented above to show the same for $T^2$--symmetric solutions (polarized and half-polarized) which are not analytic.

\subsection{ \texorpdfstring{$T^2$}{T2}-symmetric
  spacetimes}
  \label{spacetimes}

The family of vacuum $T^2$--symmetric spacetimes  is characterized by a $T^2$ isometry group which acts effectively on each spacetime in the family, with the generating Killing vector fields being everywhere spacelike. We assume that each such spacetime is the maximal globally hyperbolic development of an initial data set on a compact Cauchy surface, with the data invariant under an effective $T^2$ action. One more condition distinguishes the spacetimes we consider here from the Gowdy subfamily. Let $Y$ and $Z$ be the generators of the $T^2$ isometry. The Gowdy subfamily is characterized by the assumption that the distribution defined by the tangent planes orthogonal to the generators $Y$ and $Z$ is integrable. This condition is usually expressed as the vanishing of the two twists $K_Y$ and $K_Z$.\footnote{Let $\xi := g(Y,\cdot), \zeta : = g(Z, \cdot)$ be the generating forms of the distribution $D$. Frobenius' theorem states that $D$ is integrable if and only if $ K_Y:= \star d\xi \wedge \xi \wedge \zeta$ and $ K_Z:= \star d\zeta \wedge \xi \wedge \zeta $ both vanish.} We work here with $T^2$-spacetimes with at least one non-vanishing twist. Chru{\'s}ciel has shown \cite{Chrusciel:1990ti} that the vacuum Einstein equations force the twists to be constants, and that the condition of non-vanishing twist implies the Cauchy surfaces must have $T^3$ topology.

Such spacetimes can be foliated by areal coordinates,  in which the time coordinate labeling each symmetry group orbit is equal to the area of that orbit. This coordinate system conveniently locates the singularity at $t=0$ except in the special case of flat Kasner, as is shown by Isenberg and Weaver in \cite{Isenberg:2003dr}. Local existence for these coordinates is shown by Chru{\'s}ciel, \cite{Chrusciel:1990ti}, and  global existence is proved by Berger et.\ al.\ in \cite{Berger:1997dp}, and further clarified in  \cite{Isenberg:2003dr}.

Let $y,z$ be coordinates on $T^2$, and let $x$ be the remaining spatial coordinate, which takes values in $S^1$. The metric can be written \cite{Chrusciel:1990ti} in the form\footnote{In some representations of these metrics, two ``shift constants" $M_x$ and $M_y$ also appear. These can be removed using gauge choices; hence we leave them out.}
 \begin{equation*}
    g = e^{2(\eta -U)} \Big( -\alpha dt^2 + dx^2 \Big) 
    + e^{2U} \bigg( dy + A dz + \Big(G_1 + A G_2 \Big) dx \bigg)^2 
    + e^{-2U} t^2 \bigg( dz + G_2 dx \bigg)^2,
  \end{equation*}
  where all the metric functions $\{\eta, U, \alpha, A, G_1, G_2 \}$ depend only on $t$ and $x$. 

If both twist constants vanish, then  the function $\alpha $ can be chosen to be a constant, in which case the above metric reduces to the Gowdy metric \cite{Gowdy:1974hv}. 

The \emph{polarized} class of $T^2$--symmetric spacetimes results from setting $A$ equal to a constant in the initial data (or, equivalently, assuming that the dot product of the generators $Y,Z$ is initially the same at all spatial points\footnote{ Observe that if the dot product is constant, one can always find a new pair of generators (by taking linear combinations) which are orthogonal.}), and verifying that this condition is preserved under evolution. While the polarized spacetimes are characterized by a geometric condition, another subclass we consider, called the \emph{half-polarized} $T^2$--symmetric spacetimes, is defined by a restriction on the asymptotic behavior of the fields (see \Sectionref{sec:defnAVTD}).

Before writing down the Einstein vacuum equations, we make a few further coordinate
choices to simplify the presentation. Without loss of generality we
choose the generators such that $K_Y = 0, K_Z \equiv K \neq 0$. This
can be achieved by choosing an appropriate linear combination of any generators
for the $T^2$ action. It is sufficient to consider $K>0$ since the
transformation $K \to -K$ preserves all conditions imposed thus
far. Next we choose coordinates $y,z$ on $T^2$ so that $Y
= \partial_y$ and $Z= \partial_z$. This can be done without changing
the form of the metric above. Implementing these simplifications, and using the short-hand notation
 $U_t := \partial_t U$ for derivatives, we write the Einstein equations as the
following system of PDEs, which includes a set of second order equations
  \begin{eqnarray}
    \label{eq:WaveequationU}
    U_{tt} +\frac{U_t}{t} - \alpha U_{xx} &=&  
    \frac{\alpha_x U_x}{2} + \frac{\alpha_t U_t}{2 \alpha } 
    + \frac{e^{4 U}}{2 t^2} \big(A_t^2 - \alpha A_x^2 \big), \\
    \label{eq:WaveequationA}
    A_{tt} -\frac{A_t}{t} - \alpha A_{xx} &=&  \frac{\alpha_x A_x}{2}
    + \frac{\alpha_t A_t}{2 \alpha} - 4 A_t U_t + 4 \alpha A_x U_x,\\
    \label{eq:Waveequationeta}
    \eta_{tt} - \alpha \eta_{xx} &=&  \frac{\alpha_x \eta_x}{2} 
    + \frac{\alpha_t \eta_t}{2 \alpha} - \frac{\alpha_x^2}{4 \alpha} 
    + \frac{\alpha_{xx}}{2} - U_t^2 + \alpha U_x^2, \\
    & & + \frac{e^{4 U}}{4 t^2} \big(A_t^2 - \alpha A_x^2 \big) 
    - \frac{3 e^{2 \eta} \alpha}{4 t^4} K^2,\notag
  \end{eqnarray}
a set of  first-order equations
\begin{eqnarray}
\label{eq:Firstorderequationeta}
\eta_t &=& t U_t^2 + t\alpha U_x^2 
+ \frac{e^{4 U}}{4t} ( A_t^2 +\alpha A_x^2)
+ \frac{e^{2 \eta}}{4t^3} \alpha K^2,\\
\label{eq:Constrequationeta}
\eta_x  &=& 2t U_t U_x
+  \frac{e^{4 U}}{2t} A_t A_x - \frac{\alpha_x}{2 \alpha}, \\
\label{eq:Firstorderequationalpha}
\alpha_t &=& -\frac{e^{2 \eta}}{t^3} \alpha^2 K^2,
\end{eqnarray}
plus a set of  auxiliary equations
\begin{equation}
  \label{eq:EquationForG}
  G_{1t}= e^{2 \eta} \sqrt{\alpha}\, A K t^{-3},
  \quad G_{2t}= -e^{2 \eta} \sqrt{\alpha}\, K t^{-3}. 
\end{equation}
Here,  the auxiliary equations originate from  the definition of the twist constants $K_Y$ and
$K_Z$ and from the ``gauge" simplification $K_Y=0$ noted above.

Observe that the $T^2$--symmetric Einstein system reduces to the Gowdy
system in the standard areal coordinates if we set $K=0$, $\alpha\equiv
1$, $G_1\equiv 0$, and $G_2\equiv 0$. The Einstein equations in the Gowdy
class are semilinear and a Fuchsian analysis with analytic asymptotic
data has been carried out by Kichenassamy and Rendall 
\cite{Kichenassamy:1999kg}, and with smooth asymptotic data by Rendall
\cite{Rendall:2000ki} and by Beyer and LeFloch \cite{Beyer:2010tb}.

\subsection{Existence of AVTD solutions to the Einstein vacuum equations}
\label{sec:avtdexistence}

\subsubsection{AVTD behavior and heuristics}
\label{sec:defnAVTD}

What is the behavior of a singular solution to Einstein's equations
near the singularity? In principle the behavior could be very
complicated for a solution to a system of nonlinear PDE such as the
Einstein equations. In \cite{Lifshitz:1963hz,Belinskii:1970fu,Belinskii:1982bn} Belinskii, Khalatnikov,
and Lifshitz (BKL) propose that generically the spacetime dynamics
near the singularity is vacuum dominated, local, and oscillatory. According to this picture, an
observer traveling  toward the singularity (either backward or forward in time, depending upon the location of the singularity) would
experience an infinite sequence of Kasner epochs, and each observer at
different spatial points would experience a different, generally unrelated, sequence.

Numerical simulations of $T^2$--symmetric spacetimes \cite{Andersson:2005vg,Berger:2001dl, Lim:2009fz} support this picture, except perhaps at points where spikes occur. Whether the complicated behavior found near spikes, and the apparent prevalence of spikes, invalidates the BKL picture for general $T^2$--symmetric solutions is far from clear. However, for the restricted family of \emph{polarized} $T^2$--symmetric solutions, numerical simulations indicate that a special form of BKL behavior occurs near singularities--asymptotically velocity term dominated, or AVTD, behavior--which is not dominated by what happens near spikes. In a spacetime with  AVTD behavior,  each
observer experiences only a finite sequence of Kasner epochs in the
approach to the singularity \cite{Isenberg:1990gn,Eardley:1972ig,Hanquin:1983dl}, and the limiting spacetime is different for each observer. 

While there are no analytical studies of inhomogeneous cosmological solutions which either confirm or deny the presence of general BKL behavior, as noted above there has been a significant amount of such work supporting the generic presence of AVTD behavior in restricted families of solutions. Studies based on singular initial value problem formulations of Fuchsian PDEs are particularly well-adapted to doing this, since 
they involve specifying a choice of asymptotic behavior (a Kasner evolution independently at each point), and showing that there are solutions of the equations which approach this asymptotic behavior. If we can show that the Einstein equations for the polarized  $T^2$--symmetric spacetimes, together with certain choices of the leading order term, satisfy the conditions of the hypothesis of either Theorem \ref{th:Wellposedness1stOrderFiniteDiff} or Theorem \ref{th:Wellposedness1stOrderHigherOrder}, then we have confirmation that there are such spacetimes which have AVTD behavior.

Observe that finding solutions in a given family of spacetimes with AVTD behavior does not imply that there are not solutions in that same family  with a very different form of asymptotic behavior. However, since numerical simulations support AVTD behavior being generic among polarized $T^2$--symmetric solutions, there have been no searches for alternative forms of asymptotic behavior among them.

The name ``asymptotically velocity term dominated" refers to the fact that the leading order terms are chosen as asymptotic solutions of the ``velocity term dominated" (VTD) system, which is formed
 from the Einstein equations by
dropping terms with spatial derivatives. This step encodes the local
aspect of the BKL proposal. It can be shown \cite{Isenberg:1999ba,Clausen:2007vq} that the
following expansions for the metric functions below asymptotically
solve this VTD system in the limit $t \to 0$. We write these expansions in
terms of asymptotic data $\{ k, U_{**}, A_*, A_{**}, \eta_{*}, \alpha_* ,G_{1*},
G_{2*}\}$
with the regularity of the data specified below.
\begin{align}
  \label{eq:AVTDLeadingorder1}
  U(t,x)&=\frac 12(1-k(x))\log t+U_{**}(x)+\ldots,\\
  A(t,x)&=A_*(x)+A_{**}(x) t^{2k(x)}+\ldots,\\
  \eta(t,x)&=\frac 14(1-k(x))^2\log t+\eta_*(x)+\ldots,\\
  \alpha(t,x)&=\alpha_*(x)+\ldots,\\
  G_{1}(t,x)&=G_{1*}(x) +\ldots,\\
  \label{eq:AVTDLeadingorder6}
  G_{2}(t,x)&=G_{2*}(x) +\ldots.
\end{align}
Of particular importance here is the function $k$. It determines the Kasner exponents $p_1, p_2, p_3$ of the local Kasner solutions which are approached at any spatial point
\[p_1=(k^2-1)/(k^2+3),\quad
p_2=2(1-k)/(k^2+3),\quad
p_3=2(1+k)/(k^2+3).\]

We recall here that a $T^2$--symmetric solution is defined to be {\bf polarized} if
the two Killing vectors corresponding to the $T^2$ isometry can be
chosen to be {\sl orthogonal everywhere.} This is the case if and only if the metric coefficient 
$A\equiv const$. A solution 
with AVTD behavior has this property if and only if the asymptotic data corresponding to $A$ satisfy the conditions $A_{**}\equiv 0$ and $A_{*}\equiv const$. Since $A_{*}\equiv const$ can be gauged to $A_{*}\equiv 0$, we see that in the polarized case, there is effectively no free asymptotic data to choose which relates to $A$. There is an interesting relationship between the polarization condition and the sign of $k$: Examining equations \eqref{eq:AVTDLeadingorder1}--\eqref {eq:AVTDLeadingorder6}, we find that if a solution is not  polarized and has AVTD behavior, then there is power law blow-up at the singularity if and only if $k$ is negative. Yet if that spacetime is polarized, then regardless of the sign of $k$, there is no power law blow-up at the singularity. 

The polarization condition is relevant to our application of our Fuchsian results to $T^2$--symmetric solutions since, as we see below, our results cannot be applied unless  the condition  $\partial_x A_*=0$ holds for the asymptotic data. For polarized $T^2$--symmetric solutions, this restriction on $A_*$ is automatic. It is important to note, however, that requiring $\partial_x A_*= 0$ does not restrict us to polarized solutions. We may consider asymptotic data which has this restriction on $A_*$, but has \emph{no} restriction on $A_{**}$. $T^2$--symmetric solutions which are AVTD and which have asymptotic data of this sort are known to exist, and have been called ``half-polarized"\footnote{The use of the term``half-polarized" to describe AVTD solutions with ``half" of the asymptotic data of one of the gravitational degrees of freedom turned off first appears in a discussion of $U(1)$-symmetric solutions with AVTD behavior, in \cite{Isenberg:2002ku}.} \cite{Clausen:2007vq}. Extending the results of both \cite{Isenberg:1999ba} (analytic and polarized) and \cite{Clausen:2007vq} (higher regularity), we show here that there are large families of both half-polarized and polarized $T^2$--symmetric solutions which are smooth or of even lower regularity, and which have AVTD behavior near their cosmological singularities.

A general (neither polarized nor half-polarized) $T^2$--symmetric solution, were it to be AVTD, would have asymptotic data with both $A_*$ and $A_{**}$ non-vanishing and non-constant. Based on numerical and heuristic considerations, however, it is expected that  spacetimes with non-constant $A_*$ do not generally show AVTD behavior. Rather, these are expected to  show Mixmaster-like BKL behavior at the $t=0$ singularity, or behavior which is even more complicated (with strong spike influence). We do not address this issue here. 

We now discuss two applications of our Fuchsian results which verify AVTD behavior in $T^2$--symmetric solutions. For the first one, \Theoremref{th:existenceAVTDCk}, we make only minimal assumptions regarding  the regularity of the asymptotic data. The price to pay for this is that the result does not cover
the full expected range for the function $k=k(x)$ in
\Eqsref{eq:AVTDLeadingorder1} -- \eqref{eq:AVTDLeadingorder6}. For the second result \Sectionref{sec:optimalexistence},
\Theoremref{th:existenceAVTDCinf}, we add regularity restrictions, but we do get the expected full range of allowed values for $k$.


\subsubsection{Existence of low regularity solutions with AVTD behavior}
\label{sec:firstexistence}

The low regularity result, which we formulate, discuss, and prove in this subsection, is an application of 
 \Theoremref{th:Wellposedness1stOrderFiniteDiff} to the polarized and half-polarized solutions of the $T^2$--symmetric equations. 

\begin{theorem}[First result: AVTD (half)-polarized
  $T^2$--symmetric vacuum solutions -- finite differentiability]
  \label{th:existenceAVTDCk}
Suppose one chooses  a twist constant $K\in\R$,  a pair of asymptotic data constants $A_*$ and
  $\eta_0$, and a set of asymptotic data functions $k, U_{**}, \alpha_*\in H^{q+2}(T^1)$ (with $\alpha_*(x)>0$),
  $A_{**}\in H^{q+1}(T^1)$ and $G_{1*}, G_{2*} \in H^{q}(T^1)$
  for any $q\ge 3$, which satisfy the integrability condition\footnote{This integrability condition results from integrating the constraint equation \Eqref{eq:Constrequationeta} along a constant $t$ circle, and taking the limit as $t \rightarrow 0$. Observe that here and below we use the notation $U_{**}'(x), (\log \alpha_*)'$, etc., to denote derivatives of functions which depend on $x$ only.} 
\begin{equation*}
 \int_0^{2\pi}\Bigl((1-k(x))U_{**}'(x)
  -\frac 12(\log \alpha_*)'(x)\Bigr)dx=0,
  \end{equation*}
  together with, at each point $x\in T^1$, either
  \begin{enumerate}[label=\textit{(\roman{*})}, ref=(\roman{*})]
  \item$k(x)>1+\sqrt 6$ for arbitrary $A_{**}$ (the
    \keyword{half-polarized case}),
  \item$k(x)>1+\sqrt 6$ or $k(x)<1-\sqrt 6$ for $A_{**}\equiv 0$ 
    (the \keyword{polarized case}).
  \end{enumerate}
  Then there exists a $\delta>0$, and a $T^2$--symmetric solution $U$,
  $A$, $\eta$, $\alpha$, $G_1$, $G_2$ of Einstein's vacuum field
  equations with twist $K$ of the form
  \begin{equation*}
    (U,A,\eta,\alpha,G_1,G_2)=(U_0,A_0,\eta_0,\alpha_0,G_{1,0},G_{2,0})+W.
  \end{equation*}
  Here, the leading-order term
  $(U_0,A_0,\eta_0,\alpha_0,G_{1,0},G_{2,0})$ is given by
  \Eqsref{eq:AVTDLeadingorder1}--\eqref{eq:AVTDLeadingorder6}, with
  \begin{equation}
    \label{eq:defetaS}
    \eta_*(x):=\eta_0+\int_0^x\left((1-k(X))U_{**}'(X)
      -\frac 12(\log \alpha_*)'(X)\right)dX.
  \end{equation}
  The remainder
  $W$ is contained in $X_{\delta,\mu,q}$ (and $DW\in X_{\delta,\mu,q-1}$) for any exponent vector
  $\mu=(\mu_1,\mu_2,\mu_3,\mu_4,\mu_5,\mu_6)$ with
  \begin{equation}
    \label{eq:lowregexponents}
    \begin{split}
    1 &<\mu_1(x)
    <\min\{2,(k(x)-3)(k(x)+1)/2\},\\
     (2k(x)+\sqrt{1+4k(x)^2})/{2}&<\mu_2(x)<1+2k(x),\\
    0&<\mu_3(x)<\mu_1(x),\\ 
    0&<\mu_4(x),\mu_5(x),\mu_6(x)<(k(x)-3)(k(x)+1)/2.
  \end{split}
  \end{equation}
  This solution is unique among all solutions with the same
  leading-order term $U_0$ and with remainder $W\in X_{\delta,\mu,q}$.
\end{theorem}

Observe that by taking time derivatives of the Einstein field equations, we can also
obtain corresponding statements about the behavior of a certain number
of time derivatives $D^mW$ of the remainder function $W$. We do not elaborate on this any further here.

This result, based on \Theoremref{th:Wellposedness1stOrderFiniteDiff}, does not imply uniqueness of the solutions within the \textit{whole} class of solutions of interest:
For a given choice of asymptotic data, \Theoremref{th:existenceAVTDCk} determines that there is exactly one solution with remainder $W$ in spaces $X_{\delta,\mu,q}$ with $\mu$ given by \Eqref{eq:lowregexponents}. The full class of remainders compatible with the leading-order behavior \Eqsref{eq:AVTDLeadingorder1}--\eqref{eq:AVTDLeadingorder6} however corresponds to exponents
\[\mu_1,\mu_2-2k,\mu_3,\mu_4,\mu_5,\mu_6>0.\]
Hence, for given asymptotic data there may exist further solutions in such a larger space.
 Strict uniqueness can be explored further using techniques involving (order n)-leading order terms. We return to this issue in  \Sectionref{sec:optimalexistence} below; the price which we have to pay for strict uniqueness is that we need to require higher differentiability for the asymptotic data.

In proving \Theoremref{th:existenceAVTDCk}, it is useful to arrange the $T^2$--symmetric Einstein vacuum equations, \Eqsref{eq:WaveequationU}-\eqref{eq:EquationForG}, as well as the field variables, in a certain hierarchical form: \Eqsref{eq:WaveequationU}, \eqref{eq:WaveequationA},
\eqref{eq:Firstorderequationeta} and
\eqref{eq:Firstorderequationalpha} together form a coupled evolution system (which we label the ``main evolution equations'') for the variables $U, A, \eta$, and $\alpha$. \Eqref{eq:Constrequationeta} serves as a constraint equation for this system, while \Eqref{eq:Waveequationeta} is effectively redundant, and can be ignored. The remaining equations \Eqsref{eq:EquationForG} are evolution equations for $G_1$ and $G_2$, and can be handled after the analysis of the main evolution equations. 

We proceed now to focus on the main evolution equations, with the primary existence result for them -- the main step toward a proof of \Theoremref{th:existenceAVTDCk} -- being \Propref{prop:existencemainevolutionSqrt6}.

\paragraph{Main evolution equations.}

To rewrite the main evolution equations as a first order symmetric hyperbolic Fuchsian system, it is  useful to define certain new variables. Some of the choices of these variables are motivated by considerations in \cite{Isenberg:1999ba}, others by the discussion above in Section \ref{sec:EPD}. First, we set
\begin{equation}
  \label{eq:defxi}
  \xi:=\partial_x\alpha,
\end{equation}
whose evolution equation is obtained by taking the spatial derivative of
\Eqref{eq:Firstorderequationalpha} and by substituting any occurrence
of $\eta_x$ by the constraint \Eqref{eq:Constrequationeta}. One obtains
\begin{equation*}
  \xi_t=-\frac{e^{2 \eta}}{t^4} \alpha K^2\left(
    t\xi+\alpha\left(
    e^{4U}A_xA_t+4t^2 U_xU_t\right)
  \right).
\end{equation*}
In all other evolution equations we use \Eqref{eq:Firstorderequationalpha} to eliminate $\alpha_t$ and replace $\alpha_x$ by $\xi$.

Next, we find that for both $U$ and $\eta$, it is useful to replace the given variable by that which involves the subtraction of the indicated log term in the asymptotic VTD expansions \Eqref{eq:AVTDLeadingorder1}--\eqref{eq:AVTDLeadingorder6}: We set $\widehat{\eta}:=\eta-\frac 14(1-k)^2\log t$ and set $\widehat U:=U-\frac{1}{2}(1-k(x))\log t$; compare this to our approach in \Sectionref{sec:EPD}. Adding  a few other minor modifications, we are led to define the following set of first-order variables:
\begin{align}
  \label{eq:leadingorder1st1}
  u_1 &= \widehat U, &u_2& = D\widehat U, &u_3& = t \partial_x \widehat U, \\
  u_4 &= A, &u_5& = DA,  &u_6& = t \partial_x A, \\
  \label{eq:leadingorder1st3}
  u_7 &=\widehat \eta, &u_8& = \alpha,  &u_9& = \xi.
\end{align}
Observe that, at this stage, $k(x)$ is an arbitrary function (introduced in
\Eqsref{eq:AVTDLeadingorder1}--\eqref{eq:AVTDLeadingorder6}), with no restrictions. In terms of the 
new set of the variables, the main evolution
system \Eqsref{eq:WaveequationU}, \eqref{eq:WaveequationA},
\eqref{eq:Firstorderequationeta} and
\eqref{eq:Firstorderequationalpha}  can  be written in symmetric hyperbolic form as follows:
\begin{align}
  \label{eq:mainevolutionsystem1}
  Du_1 - u_2  =&\,  0, \\
  Du_2 - u_8 t \partial_x u_3 
  =&\,  \frac{1}{2} t u_9 (u_3-\frac 12 t\log t k') 
  + \frac{1}{2}e^{4u_1}t^{-2k}\left( u_5^2 - u_8 u_6^2 \right)\\
  &- \frac{1}{4 } e^{2 u_7}t^{1/2(1-k)^2-2} u_8  K^2 (1-k+2u_2)\notag\\
  &-\frac 12 t^2\log t k'' u_8,\notag \\
  u_8 Du_3 - u_8 t \partial_x u_2 -u_8 u_3=&\,0, \\
  Du_4 -u_5 =&\,  0, \\
  Du_5  -2k u_5- u_8 t \partial_x u_6 =&\, - 4 u_5 u_2 
  + \frac{1}{2} t u_9 u_6+ 2 u_8 u_6 (2u_3-t\log t k') \\
  &- \frac{1}{2} e^{2u_7} t^{1/2(1-k)^2-2} u_8 u_5 K^2,\notag \\
  u_8Du_6  - u_8 t \partial_x u_5  - u_8 u_6 =&\, 0, \\
  Du_7 =&\, (1-k)u_2+u_2^2  +  \frac 14 u_8 (2u_3-t\log t k')^2\\ 
  &+ \frac{1}{4} t^{-2k} e^{4 u_1} \left( u_5^2 + u_8 u_6^2 \right)+ \frac{1}{4} e^{2 u_7}  t^{1/2(1-k)^2-2} u_8 K^2,\notag \\
  Du_8 =&\, - e^{2u_7} t^{1/2(1-k)^2-2} u_8^2 K^2, \\
  \label{eq:mainevolutionsystem9}
  Du_9 =&\, - e^{2 u_7} t^{1/2(1-k)^2-2} u_8 K^2\\
    &\!\!\!\!\!\!\!\!\!\!\!\!\!\!\!\!\!\!\!\!\!\!\cdot\Bigl(\frac{(1-k+2u_2) (2u_3-t\log t k') u_8}{t} 
    + t^{-1-2k} u_5 u_6 u_8e^{4 u_{1}}+ u_9
  \Bigr),\notag
\end{align}
or equivalently as
\begin{equation}
  \label{eq:mainevol1st1}
  S_1 Du+S_2 t\partial_x u+N u=f[u],
\end{equation}
where
\be 
	\label{eq:T2S0matrix}
  \SO{u} =\diag (1,1,u_8,1,1,u_8,1,1,1),
\ee
\be
  	\label{eq:T2S1matrix}
  \ST{u} =
  \begin{pmatrix}
    0 &0 &0 &0 &0 &0 &0 &0 &0\\
    0 &0 & -u_8 &0 &0 &0 &0 &0 &0\\
    0 & -u_8 &0 &0 &0 &0 &0 &0 &0\\
    0 &0 &0 &0 &0 &0 &0 &0 &0\\
    0 &0 &0 &0 &0 &-u_8 &0 &0 &0\\
    0 &0 &0 &0 &-u_8 &0 &0 &0 &0\\
    0 &0 &0 &0 &0 &0 &0 &0 &0\\
    0 &0 &0 &0 &0 &0 &0 &0 &0\\
    0 &0 &0 &0 &0 &0 &0 &0 &0
  \end{pmatrix},
\ee
\be
  \label{eq:mainevol1st2}
  \NN{u} =\begin{pmatrix}
    0 &-1 &0 &0 &0 &0 &0 &0 &0\\
    0 &0 &0 &0 &0 &0 &0 &0 &0\\
    0 &0 & -u_8 &0 &0 &0 &0 &0 &0\\
    0 &0 &0 &0 &-1 &0 &0 &0 &0\\
    0 &0 &0 &0 &-2k &0 &0 &0 &0\\
    0 &0 &0 &0 &0 &-u_8 &0 &0 &0\\
    0 &0 &0 &0 &0 &0 &0 &0 &0\\    
    0 &0 &0 &0 &0 &0 &0 &0 &0\\
    0 &0 &0 &0 &0 &0 &0 &0 &0
  \end{pmatrix}.
\ee
Note that we have multiplied the third and sixth equations by $u_8$. The source-term vector $f$ is given by the right-hand sides of the evolution system
\Eqsref{eq:mainevolutionsystem1}--\eqref{eq:mainevolutionsystem9}. The reason for keeping this particular form of the matrix $N(u)$ (and not absorbing some of its entries into the source-term) becomes clear shortly.


\paragraph{AVTD solutions of the main evolution system.}
We now show as an application of
\Theoremref{th:Wellposedness1stOrderFiniteDiff}, and as a step towards proving Theorem \ref{th:existenceAVTDCk}, that there exist unique solutions to the singular
initial value problem for the main evolution system \eqref{eq:mainevol1st1}-\eqref{eq:mainevol1st2}, with AVTD leading-order term
\begin{equation}
  \label{eq:leadingorderterm}
  \begin{split}
    u_0=&(u_{1,0},u_{2,0},u_{3,0},u_{4,0},u_{5,0},u_{6,0},u_{7,0},u_{8,0},u_{9,0})\\
    =&\left(U_{**}, 0, t U_{**}',
      A_*+A_{**} t^{2k},
      2k A_{**} t^{2k}, 0,
      \eta_{*},
      \alpha_*,
      \xi_*\right).
  \end{split}
\end{equation}
Although not needed for our present argument, we note (by inspecting \Eqref{eq:canonleadingterm}) that this choice of $u_0$ is an ODE-leading-order term; cf.\ \Sectionref{sec:EPD}.

To check that we have a quasilinear symmetric hyperbolic system, we need to specify an exponent vector along with the PDE system and a leading order term. Looking ahead to the conditions of block diagonality, we choose 
\begin{equation}
  \label{eq:T2symmexpo}
  \mu=(\mu_1,\mu_1,\mu_1,\mu_2,\mu_2,\mu_2,\mu_3,\mu_4,\mu_4),
\end{equation}
and expect to construct remainders in spaces $X_{\delta,\mu,q}$ with $\mu$ given by
\[\mu_1,\mu_3,\mu_4>0,\quad
\qquad \mu_2>2k.\]
We then find, after replacing $u_8$ by $\alpha_*+w_8$, that so long as we choose $\alpha_* > 0$, and so long as we require that all of the asymptotic data functions be contained in some $H^q(T^1)$ (which we fix below), we indeed have a quasilinear symmetric hyperbolic system, which in addition does satisfy the block diagonality condition. 

Before continuing the argument that the hypothesis of \Theoremref{th:Wellposedness1stOrderFiniteDiff} is satisfied, we state our result. 

\begin{proposition}
  \label{prop:existencemainevolutionSqrt6}
For any twist constant $K\in\R$, for any Sobolev differentiability index $q\ge 3$, and for any choice  of the
asymptotic data functions such that  $A_*$ is an arbitrary constant, $\alpha_*(x)>0$, $k, U_{**}, \alpha_*\in H^{q+2}(T^1)$, $A_{**}\in H^{q+1}(T^1)$ and $\eta_* \in H^{q}(T^1)$,  and $k$ satisfies (at each  $x\in T^1$) either
  \begin{enumerate}[label=\textit{(\roman{*})}, ref=(\roman{*})]
  \item $k(x)>1+\sqrt 6$ (for arbitrary $A_{**}$ the half-polarized case), 
  \item  $k(x)>1+\sqrt 6$ or $k(x)<1-\sqrt 6$ (for $A_{**}\equiv 0$ the polarized case), 
  \end{enumerate}
 there exists a $\delta_1\in (0,\delta]$,
  and a unique solution of the first order main evolution system
  \Eqsref{eq:mainevol1st1}--\eqref{eq:mainevol1st2} with leading-order
  term $u_0$ and remainder $w\in X_{\delta_1,\mu,q}$ (and $Dw\in
  X_{\delta_1,\mu,q-1}$) so long as the exponent vector $\mu$ given by \Eqref{eq:T2symmexpo}
satisfies the following inequalities for all $x\in T^1$: 
  \begin{gather*}
    1 <\mu_1(x)
    <\min\big\{ 2,(k(x)-3)(k(x)+1)/2\big\},
\\
    {1 \over 2}  \left(2k(x)+\sqrt{1+4k(x)^2}\right) <\mu_2(x)<1+2k(x),\\
    0<\mu_3(x)<\mu_1(x),\\ 
    0<\mu_4(x)<\frac12 (k(x)-3)(k(x)+1).
  \end{gather*}
\end{proposition}

Observe here that the inequality just stated for $\mu_2$ is not required to hold in the case of a polarized
solution, since in that case $A$ is not a dynamical variable, and this condition is vacuous. Although
here and below we list results for the polarized and half-polarized
cases together for compactness, the reader focusing on the polarized case may ignore all references to 
$\mu_2$ and to $w_4$, $w_5$ and $w_6$.

As noted above, this proposition is an application of \Theoremref{th:Wellposedness1stOrderFiniteDiff} to   \Eqsref{eq:mainevol1st1}--\eqref{eq:mainevol1st2}.  In the next lemma we verify that under the assumptions of   \Propref{prop:existencemainevolutionSqrt6} the \Conditionref{en:cond2} of \Theoremref{th:Wellposedness1stOrderFiniteDiff} is satisfied. The first condition follows directly from the definition of the energy dissipation
matrix $M_0$.

\begin{lemma}
  \label{lem:cond1}
  The energy dissipation matrix $M_0$ defined in 
    \Eqref{eq:energydissipationmatrix} corresponding to
  \Eqsref{eq:mainevol1st1}--\eqref{eq:mainevol1st2}, to the leading-order
  term $u_0$ given by \Eqref{eq:leadingorderterm} and to the exponent vector $\mu$
  of the form \Eqref{eq:T2symmexpo}
  is positive definite at every $x$, provided that 
 $$
\aligned
& \alpha_*(x)>0, \quad \mu_1(x)> 1, 
\\
& \mu_2(x) > \max\{ 1, k(x)+ {1 \over 2} \sqrt{1+4k(x)^2} \}, \quad \mu_3(x), \mu_4(x) > 0, 
\endaligned
$$
  hold for all $x\in T^1$.
\end{lemma}

The next lemma establishes \Conditionsref{en:cond4N} and \ref{cond:LipschitzF} of \Theoremref{th:Wellposedness1stOrderFiniteDiff}.

\begin{lemma}
  \label{lem:T2symmFLuOp}
  The operator $\FLuOp$ corresponding to
  \Eqsref{eq:mainevol1st1}--\eqref{eq:mainevol1st2}, to the leading-order
  term $u_0$ given by \Eqref{eq:leadingorderterm}, and to the exponent vector $\mu$
  of the form \Eqref{eq:T2symmexpo} satisfies
  \Conditionref{en:cond4N} and \ref{cond:LipschitzF} of
  \Theoremref{th:Wellposedness1stOrderFiniteDiff} for some exponent
  vector $\nu>\mu$, for some sufficiently small $\delta>0$, and for a choice of the 
  differentiability index $q\ge 3$, so long as 
  $\alpha_*$ and $\eta_*$ are functions in $H^q(T^1)$, $A_{**}$ is contained in
  $H^{q+1}(T_1)$, $k$ and $U_{**}$ are elements of  $H^{q+2}(T^1)$, and if at each point $x\in T^1$, the following inequalities hold for $\mu$ and $k$: 
$$
\aligned
    \max\{0,1-(k(x)-3)(k(x)+1)/2\}  &<\mu_1(x)
    <\min\{ 2,(k(x)-3)(k(x)+1)/2\},
\\
     2k(x)<\mu_2(x) &<\min\{ 1+2k(x),\mu_1(x)+2k(x)\},
\\
    0<\mu_3(x)&<\mu_1(x),
\\ 
    0<\mu_4(x)<\min\{ (k(x)-3)(k(x)+1)/2, & \mu_1(x)-1+(k(x)-3)(k(x)+1)/2\},
 \endaligned
$$
  and 
  \begin{gather*}
  3 < k(x) \quad \text{in the half-polarized case,}\\
  3 < k(x) \quad \text{or} \quad k(x) < -1 \quad \text{in the polarized case}.
  \end{gather*}
In both the polarized and the half-polarized cases, it follows from the two inequalities stated above for $\mu_1$ that $k(x)$ must either satisfy $k(x)>1+\sqrt 5$ or $k(x)<1-\sqrt{5}$.
\end{lemma}

\begin{proof}
If the operator  $\FLuOp$, defined in \Eqref{eq:defFLu}, is written out explicitly, it consists of products of asymptotic data functions, and  components of the unknown function $w$ (or products involving exponential functions of these). All of the multiplicands in these products are, by hypothesis, contained in designated function spaces (of the form $X_{\delta,\mu,q}$). Thus, to check \Conditionref{en:cond4N}  of  \Theoremref{th:Wellposedness1stOrderFiniteDiff}, we primarily need to know the multiplication algebra of spaces such as $X_{\delta,\mu,q}$. The result we need is provided by 
\Lemref{lem:product1} in the appendix. To check \Conditionref{cond:LipschitzF}, we need results concerning Lipschitz properties of products and exponential functions of elements of the spaces $X_{\delta,\mu,q}$. 
\Lemref{lem:product} and \Lemref{lem:exponential} provide these needed results.
\end{proof}

\begin{proof}[Proof of \Propref{prop:existencemainevolutionSqrt6}]
If we wish to  use \Theoremref{th:Wellposedness1stOrderFiniteDiff}  to show that the system discussed in Proposition \ref{prop:existencemainevolutionSqrt6} admits solutions with the stated properties, it is sufficient that i) the asymptotic data functions, which appear in the leading-order matrices $S_{1,0}$, $S_{2,0}$ and $N_0$,
(i.e. the functions $\alpha_*$ and $k$), be contained in $H^{q+2}(T^1)$ (with $q\ge 3$); and ii) we choose the function $k(x)$ so that the hypotheses of both of the above lemmas are satisfied.  We readily check that exponent functions $\mu_1$, $\mu_2$, $\mu_3$ and $\mu_4$, which satisfy the combined inequalities, can be found if and only if $k(x)>1+\sqrt 6$ in the half-polarized case, and either $k(x)>1+\sqrt 6$ or $k(x)<1-\sqrt{6}$ in the polarized case.
\end{proof}

\paragraph{The full set of Einstein's vacuum field equations.}
Thus far, we have constructed solutions $u$ of the main evolution equations for the $T^2$--symmetric system 
with the leading-order behavior \Eqref{eq:leadingorderterm}, according
to \Propref{prop:existencemainevolutionSqrt6}. Given such a solution $u$, we may ask under what conditions is this a solution of the 
 \textit{full} set of Einstein's vacuum
field equations, \Eqsref{eq:WaveequationU}--\eqref{eq:EquationForG},
with $U = u_1+\frac 12 (1 - k) \log t$, $A =u_4$, $\eta = u_7+\frac 14 (1 - k)^2 \log t$, and
$\alpha = u_8$.

\begin{proposition}
\label{prop:fullEFE}
For any solution of \Propref{prop:existencemainevolutionSqrt6} with
asymptotic data satisfying all the conditions in
\Theoremref{th:existenceAVTDCk}, the full set of Einstein's vacuum
field equations \Eqsref{eq:WaveequationU} --
\Eqref{eq:Firstorderequationalpha} are satisfied, and
\Eqsref{eq:EquationForG} can be solved for $G_1$ and $G_2$ as stated
in \Theoremref{th:existenceAVTDCk}.
\end{proposition}

\begin{proof}
Since the equations for $G_1$ and $G_2$, \Eqref{eq:EquationForG}, are semi-decoupled from the rest, we ignore them (as well as $G_1$ and $G_2$) to start, and focus on the subsystem
\Eqsref{eq:WaveequationU}--\eqref{eq:Firstorderequationalpha}. To monitor  the extent to which this subsystem is satisfied by fields which satisfy the main evolution equations, it is useful to define the following set. Based on \Eqref{eq:Waveequationeta}, we define 
\begin{equation}
\label{Heqn}
H:=-\eta_{tt} + \alpha \eta_{xx} +  \frac{\xi \eta_x}{2} 
    + \frac{\alpha_t \eta_t}{2 \alpha} - \frac{\xi^2}{4 \alpha} 
    + \frac{\alpha_{xx}}{2} - U_t^2 + \alpha U_x^2 \\
    + \frac{e^{4 U}}{4 t^2} \big(A_t^2 - \alpha A_x^2 \big) 
    - \frac{3 e^{2 \eta} \alpha}{4 t^4} K^2
\end{equation}
and, based on the constraint \Eqref{eq:Constrequationeta}, 
\begin{equation}
  \label{eq:Constraint1}
  C_1:=-\eta_x+2t U_t U_x
  +  \frac{e^{4 U}}{2t} A_t A_x - \frac{\alpha_x}{2 \alpha}. 
\end{equation}
Based on the constraints which stem from the definition of the new variables which allow us to rewrite the original system in first order, we define 
\begin{align}
  \label{eq:Constraint2}
  C_2:= &\alpha_x-\xi, 
  \end{align}
  and
  \begin{align}
  C_3 := & u_2/t -  \partial_t u_1, &  C_4 := & u_3/t -  \partial_x u_1, \\
  C_5 := & u_5/t -  \partial_t u_4, &   C_6 := & u_6/t -  \partial_x u_4.
\end{align}

The constraint-violation terms $C_3$, $C_4$, $C_5$ and $C_6$ are the easiest to handle. Arguing as we do in the discussion of $\Delta_u$ in \Sectionref{sec:EPD}, we find that the evolution equation for each of these, induced by the main evolution system, takes the form $DC_3=0, DC_4=0$, etc. Then, since the form of the leading order term for the main system implies that each of these terms must asymptotically vanish, it follows that each must vanish for all time. Observe that this determination that $C_3, C_4, C_5$, and $C_6$ all vanish allows us to freely substitute in the consequences of their vanishing in the analysis of $H$, $C_1$, and $C_2$. Such substitution is very useful.

We now focus on $H$, $C_1$ and $C_2$, noting that a solution $u$ of \Eqsref{eq:mainevolutionsystem1} --
\eqref{eq:mainevolutionsystem9} is a solution of
\Eqsref{eq:WaveequationU}--\eqref{eq:Firstorderequationalpha} (with
the above replacements) if and only if $H,C_1$, and $C_2$ vanish identically. Presuming that $u$ is a solution of 
\Eqsref{eq:mainevolutionsystem1} -- \eqref{eq:mainevolutionsystem9},  we calculate
\begin{equation}
\label{H-C1-C2}
 H=-u_8 C_{1,x}-\frac 12 u_{8,x} C_1
 +\frac 14\Bigl(e^{4u_1}t^{-1-2k}u_5u_6+(k-1-2u_2) (\log t k'-2 t^{-1}u_3)\Bigr)C_2,
 \end{equation}
which tells us that so long as we can show that $C_1$ and $C_2$ vanish, it follows that $H$ vanishes as well. We may therefore focus on $C_1$ and $C_2$, for which the evolution equations take the following form:
\begin{align*}
D{C_1}=&-\frac 12K^2e^{2u_7 }t^{(k-3)(k+1)/2}\, {C_1}\\
 &-
\left(u_3^2-\frac 14 e^{4u_1} t^{-2k} u_6^2+t\log t\, u_3\, k'-\frac 14 t^2\log^2 t\,
  (k')^2\right) {C_2}, \\
D{C_2}=&2 e^{2u_7} K^2 t^{(k-3)(k+1)/2}\, u_8^2 {C_1}.
\end{align*}
Under
our hypothesis (which implies in
particular that the coefficients here  are continuous functions of $x$), 
this set of evolution equations for $C_1$ and $C_2$ can be treated 
as an essentially independent system of linear homogeneous ODEs at each spatial point. Noting that (for $q\ge 3$) the coefficients on the right side of both equations are well-behaved as $t\rightarrow 0$ and converge to zero at every spatial point. Hence this system for the unknowns $C_1$ and $C_2$ is of the form \Eqref{eq:1stordersystem} with $S_1$ the identity matrix, and $S_2$ and $N_2$ the zero matrices. The singular initial value problem of the form $C_1=C_{1*}+w_1$, $C_2=C_{2*}+w_2$ with $w_1,w_2\in X_{\delta,\mu,q}$ has a unique solution for every  prescribed $C_{1*}$ and $C_{2*}$ and for every sufficiently small $\mu>0$. The definition of the quantities $C_1$ and $C_2$ in terms of the variables $U$, $A$, $\eta$, $\xi$ and $\alpha$ according to \Eqsref{eq:Constraint1} and \eqref{eq:Constraint2} implies uniqueness for all constraint violations which are compatible with solutions of \Propref{prop:existencemainevolutionSqrt6}. In particular, the unique solution corresponding to $C_{1*}=C_{2*}=0$ is $C_1\equiv C_2\equiv 0$.

It remains to determine how the condition $C_{1*}\equiv C_{2*}\equiv
0$ relates to the choice of the 
asymptotic data functions $k$, $U_{**}$, $A_{*}$, $A_{**}$, $\eta_*$,
$\alpha_*$ and $\xi_*$. For a solution $u$ as above, the functions $C_1$ and $C_2$
defined by \Eqsref{eq:Constraint1} and \eqref{eq:Constraint2}
converge uniformly in space at $t=0$,  and we obtain
\[C_{1*}=-\eta_*'+(1-k)U_{**}'-\frac{\alpha_*'}{2\alpha_*},
\quad\qquad
C_{2*}=\alpha_*'-\xi_*.\]
It follows that  $C_{2*}=0$ if and only if
\begin{equation}
  \label{eq:constraintAsymptData1}
  \xi_{*}=\alpha_{*}',
\end{equation}
and $C_{1*}=0$ if and only if, for an arbitrary constant $\eta_0$, 
\begin{equation}
  \label{eq:constraintAsymptData2}
  \eta_*(x)=\eta_0+\int_0^x\left((1-k(X))U_{**}'(X)
  -\frac 12(\log \alpha_*)'(X)\right)dX. 
\end{equation} 
 In particular, the spatial topology implies that we must
choose the asymptotic data $k$, $U_{**}$ and $\alpha_*$ such that
\[\int_0^{2\pi}\left((1-k(x'))U_{**}'(x')
  -\frac 12(\log \alpha_*)'(x')\right)dx'=0.\]
We thus conclude  that a solution $u$ of
\Propref{prop:existencemainevolutionSqrt6} is a solution of
\Eqsref{eq:WaveequationU}--\eqref{eq:Firstorderequationalpha} (with
the above replacements) if and
only if the asymptotic data functions satisfy
\Eqsref{eq:constraintAsymptData1} and \eqref{eq:constraintAsymptData2}.

It only remains to solve \Eqsref{eq:EquationForG} for $G_1$ and $G_2$.  The right-hand sides of these two equations have the asymptotic behavior (for $t \rightarrow 0$) of  a power of $t$ larger
than $-1$ uniformly in space. Hence the right-hand sides are
integrable in $t$ at $t=0$ at every spatial point, and the general solution is
\begin{align*}
 G_{1}(t,x)&=G_{1*}(x)+\int_0^t e^{2 \eta(t',x)}
\sqrt{\alpha(t',x)}\, A(t',x) K t'^{-3} dt',\\
 G_{2}(t,x)&=G_{2*}(x)-\int_0^t e^{2 \eta(t',x)}
\sqrt{\alpha(t',x)}\, K t'^{-3} dt'.
\end{align*}
The functions $G_1-G_{1*}$ and $G_2-G_{2*}$ are  contained in
$X_{\delta_1,\mu_5,q}$ and $X_{\delta_1,\mu_6,q}$, respectively, for
any choice of exponents $0<\mu_5(x),\mu_6(x)<1/2(k(x)-3)(k(x)+1)$, if $u$ is a solution of 
\Propref{prop:existencemainevolutionSqrt6}.  We therefore
can take $G_{1*},G_{2*}\in H^q(T^1)$.
\end{proof}

\subsubsection{Optimal existence and uniqueness result}
\label{sec:optimalexistence}

The result we prove in Section  \ref{th:existenceAVTDCk} allows for relatively rough asymptotic data, but consequently sacrifices some of the expected range (based on numerical and heuristic studies \cite{Berger:2001dl,Andersson:2005vg,Lim:2009fz}) of allowed values for the asymptotic velocity $k=k(x)$. In this section, we consider only smooth asymptotic data, and can then prove a result which increases the range for $k$. 

\begin{theorem}[Optimal result: AVTD (half)-polarized
  $T^2$--symmetric solutions -- infinite differentiability]
  \label{th:existenceAVTDCinf}
 Suppose one chooses  a twist constant $K\in\R$,  a pair of asymptotic data constants $A_*$ and
  $\eta_0$, and a set of asymptotic data functions $k, U_{**}, \alpha_*$ (with $\alpha_*(x)>0$),
  $A_{**}, G_{1*}, G_{2*} \in C^\infty(T^1)$
 which satisfy the integrability condition
  \begin{equation*}
 \int_0^{2\pi}\Bigl((1-k(x))U_{**}'(x)
  -\frac 12(\log \alpha_*)'(x)\Bigr)dx=0,
  \end{equation*}
  together with, at each point $x\in T^1$, either
\begin{enumerate}[label=\textit{(\roman{*})}, ref=(\roman{*})]
  \item$k(x)>3$ for arbitrary $A_{**}$ (the
    \keyword{half-polarized case}), or
  \item$k(x)>3$ or $k(x)<-1$ for $A_{**}\equiv 0$ 
    (the \keyword{polarized case}).
  \end{enumerate}
  Then there exists a $\delta>0$, and a $T^2$--symmetric solution $U$,
  $A$, $\eta$, $\alpha$, $G_1$, $G_2$ of Einstein's vacuum field
  equations with twist $K$ of the form
  \begin{equation*}
    (U,A,\eta,\alpha,G_1,G_2)=(U_0,A_0,\eta_0,\alpha_0,G_{1,0},G_{2,0})+W.
  \end{equation*}
  Here, the leading-order term
  $(U_0,A_0,\eta_0,\alpha_0,G_{1,0},G_{2,0})$ is given by
  \Eqsref{eq:AVTDLeadingorder1}--\eqref{eq:AVTDLeadingorder6} and
  \Eqref{eq:defetaS}.
  The remainder
  $W$ is contained in $X_{\delta,\mu,\infty}$ (and $DW\in X_{\delta,\mu,\infty}$) for some exponent vector
  $\mu=(\mu_1,\mu_2,\mu_3,\mu_4,\mu_5,\mu_6)$ with
  \begin{equation}
    \label{eq:wishformu}
    \mu_1,\mu_2-2k,\mu_3,\mu_4,\mu_5,\mu_6>0.
  \end{equation}
  This solution is unique among all solutions with the same
  leading-order term $u_0$ and with remainder $W\in X_{\delta,\mu,\infty}$.
\end{theorem}

Comparing this result with \Theoremref{th:existenceAVTDCk}, we see the differences in the hypothesized regularity of the asymptotic data, and in the allowed range of the asymptotic velocity $k(x)$ in the two results. As we find in proving this result, in fact one does not need $C^\infty$ data; data of ``sufficiently high differentiability" is enough. One may ask if the reduced range for $k(x)$ for rough data is a real effect, or an artifact of our method of proof (which remains an open question). Observe one other important difference between the two theorems:  \Theoremref{th:existenceAVTDCinf} provides a stronger result regarding the uniqueness of solutions to the singular initial value problem for these systems: While in \Theoremref{th:existenceAVTDCk}, there could in principle exist more than one solution for a given set of asymptotic
data (since $\mu_1$ has to be larger than one), we find that according to \Theoremref{th:existenceAVTDCinf}, there is exactly one solution for the remainder functions in the full space of interest given by \Eqref{eq:wishformu}. We note, however, that there
may be two solutions for the same asymptotic data which differ by a
factor which goes to zero faster than every power of $t$ at $t=0$.

The proof of  \Theoremref{th:existenceAVTDCinf} is in many ways similar to that of \Theoremref{th:existenceAVTDCk}, but with the big difference that the latter involves the application of \Theoremref{th:Wellposedness1stOrderFiniteDiff}, while \Theoremref{th:existenceAVTDCinf} is obtained by applying 
\Theoremref{th:Wellposedness1stOrderHigherOrder}. In both cases, these results for Fuchsian singular initial value problems (\Theoremref{th:Wellposedness1stOrderFiniteDiff} or \Theoremref{th:Wellposedness1stOrderHigherOrder})
are applied to the main evolution system for $T^2$--symmetric solutions. The portion of the proof of both \Theoremref{th:existenceAVTDCinf} and  \Theoremref{th:existenceAVTDCk} which shows that the existence of a solution for the main evolution system implies the existence of a proof to the full system (since one can choose the asymptotic data in such a way that the constraints are necessarily satisfied) is essentially the same for the two cases. Hence, to prove \Theoremref{th:existenceAVTDCinf}, all we need is the following proposition (with the rest of the proof taken care of by the arguments appearing in the proof of Proposition \ref{prop:fullEFE}).

\begin{proposition}
  \label{prop:existencemainevolutionSqrt4}
For any twist constant $K\in\R$, and for any choice of the asymptotic data functions such that  $A_*$ is an arbitrary constant, $\alpha_*(x)>0$, $k, U_{**}, \alpha_*,\eta_0 \in C^\infty(T^1)$,  and either
  \begin{enumerate}[label=\textit{(\roman{*})}, ref=(\roman{*})]
  \item $k(x)>3$ (for arbitrary $A_{**}$ the half-polarized case), or
  \item  $k(x)>3$ or $k(x)<-1$ (for $A_{**}\equiv 0$ the polarized case), 
  \end{enumerate}
  for all $x\in T^1$, there exists a $\delta_1\in (0,\delta]$,
  and a unique solution of the main evolution system
  \Eqsref{eq:mainevol1st1}--\eqref{eq:mainevol1st2} with leading-order
  term $u_0$ and remainder $w\in X_{\delta_1,\mu,\infty}$ (and $Dw\in
  X_{\delta_1,\mu,\infty}$), so long as the exponent vector $\mu$ given by
  \begin{equation*}
  \mu=(\mu_1,\mu_1,\tilde\mu_1,\mu_2,\mu_2,\mu_2,\mu_3,\mu_4,\mu_4)
\end{equation*}
satisfies the following inequalities for all $x\in T^1$
$$
\aligned
    0 <\mu_1(x)
    &<\min\{2,(k(x)-3)(k(x)+1)/2\},\\
     2k(x)<\mu_2(x)&<\min\{1+2k(x),\mu_1(x)+2k(x)\},\\
    0<\mu_3(x)&<\mu_1(x),\\ 
    0<\mu_4(x)&<(k(x)-3)(k(x)+1)/2,
 \endaligned
$$
and $\tilde\mu_1$ strictly smaller than, but sufficiently close to $1+\mu_1$.
\end{proposition}

One of the key differences between the hypothesis here and that of \Propref{prop:existencemainevolutionSqrt6} is the chosen form for the exponent vector $\mu$. The form we use here -- in particular, the choice of $\tilde\mu_1$ for the third component -- allows for the wider range of $k(x)$. We can choose this form here, but not in \Propref{prop:existencemainevolutionSqrt6}, since here we do only need to satisfy the modified block diagonal conditions in \Theoremref{th:Wellposedness1stOrderHigherOrder} (as part of \Propref{prop:propHO}). 
Note that we could have carried out a similar modification
for the sixth component of $\mu$, but it turns out to be unnecessary; the
regularity of the spatial derivative of the $A$ variable follows
automatically.

\begin{proof}
The proof of this result is based on the use of 
 (order n)-leading-order terms, which are 
  developed and discussed in \Sectionref{sec:expansionsHO}.  and, in particular, in 
  \Theoremref{th:Wellposedness1stOrderHigherOrder}. Recalling  that the leading order term $u_0$
 (which takes the form $u_0$ with components given by 
  \Eqref{eq:leadingorderterm}) is in fact an ODE-leading-order term, we check the conditions in the hypothesis of  \Theoremref{th:Wellposedness1stOrderHigherOrder}, noting that a portion of these conditions appear in the hypothesis of \Propref{prop:propHO}.  We check that  the matrix
  $\SOLInvu\NLu$ is in Jordan normal form, and that our choice of
  $\mu$ is  strictly larger than the negatives of the corresponding diagonal elements of
    $S_{1,0}^{-1}N_0$.
Then,  \Conditionref{en:cond1HOSource} of
    \Propref{prop:propHO} (as part of
    \Theoremref{th:Wellposedness1stOrderHigherOrder}), can be checked
     in essentially the same way as is done in the proof of \Lemref{lem:T2symmFLuOp}, but
     now with the
     inequalities listed in the hypothesis of
     \Propref{prop:existencemainevolutionSqrt4}.
   \Conditionref{en:cond4HO} of
    \Propref{prop:propHO} (as part of
    \Theoremref{th:Wellposedness1stOrderHigherOrder})
    follows by repeated applications of \Lemref{lem:productII}.
\end{proof}


\section{Concluding remarks}
\label{conclusion}

Our results here show that there is a large collection of smooth, polarized and half-polarized $T^2$--symmetric solutions of the Einstein vacuum equations which exhibit AVTD behavior in a neighborhood of their singularities. What can we show further? 

Numerical and heuristic studies of $T^2$--symmetric solutions  \cite{Berger:2001dl, Andersson:2005vg, Lim:2009fz} strongly indicate that AVTD behavior is not found in such spacetimes unless they satisfy a polarization condition. These studies do  support the conjecture that AVTD behavior occurs generically in polarized $T^2$--symmetric solutions. While Fuchsian methods of the sort developed here are not expected to be effective in determining such genericity, further numerical explorations of the polarized $T^2$--symmetric solutions could be very useful. Among the issues which might be explored numerically is whether the distinction in the results we have obtained for solutions of finite differentiability and those which are $C^{\infty}$ is significant in any sense. 

Observe that one distinguishing feature of our approach here is that an approximation scheme is at the core of the method.
This scheme  can be implemented for numerical computations straightforwardly and contains useful built-in convergence and error estimates. In earlier work \cite{Beyer:2010tb,Beyer:2010wc,Beyer:2011ce}, building on \cite{Amorim:2009ka}, we implemented this scheme in the context of semilinear symmetric hyperbolic Fuchsian equations of second-order  and have obtained very accurate simulations for Gowdy solutions. We expect to get similarly good results for the polarized $T^2$--symmetric solutions, which we have studied here. Of particular interest would be to explore the issue of whether the ``optimal domain'' for the asymptotic velocity $k(x)>3$ (or $k(x)<-1$) can only be obtained for smooth solutions (as suggested by our discussion in \Sectionref{application}) and to see what might happen if we try  to construct such a solution with lower differentiability.

It \emph{is} expected, based on numerical simulations \cite{Berger:1998kg}, that polarized (and half-polarized) $U(1)$-symmetric solutions exhibit AVTD behavior. Moreover, Fuchsian methods \cite{Isenberg:2002ku} confirm this, at least for analytic solutions. The methods we have developed here, generalized to PDEs on higher dimensional manifolds (this is done for $T^n$ in \cite{Ames:2012tm}), should be applicable to the polarized and half-polarized $U(1)$-symmetric solutions, showing that smooth solutions of this type also exhibit AVTD behavior.

\section*{Acknowledgements}
F.B.\ was partially supported by a special assistance grant of the University
of Otago during 2011. The authors E.A.\ and J.I.\  are partially supported by NSF grant PHY-0968612. E.A., \ J.I., \ and P.LF. thank the University of Otago for sponsoring their visits to Dunedin while some of this research was carried out. P.LF was also supported by the Agence Nationale de la Recherche through the grant ANR SIMI-1-003-01 (Mathematical General Relativity.~Analysis and geometry of spacetimes with low regularity).

\appendix
\appendixpage
\addappheadtotoc

\section{Properties of the spaces \texorpdfstring{$X_{\delta,\mu,q}$}{X}}
\label{sec:propertiesspaces}
In this section we list further basic properties of the spaces
$X_{\delta,\mu,q}$ which are defined
in \Sectionref{sec:firstordertheory} as the completion of the normed
vector spaces $(C^\infty((0,\delta]\times T^1),\|\cdot\|_{\delta,\mu,q})$, cf.\ \Eqref{eq:norm}. Recall that $\delta>0$ is a constant, $\mu$ is an exponent vector and $q$ is a non-negative integer. We now also
define the spaces $\widehat X_{\delta,\mu,q}$ as the set of maps $f:(0,\delta]\rightarrow H^q(T^1)$ with the property that $\RR{\mu} f$ is bounded and continuous; cf., \Eqref{eq:defR}.  If we endow $\widehat X_{\delta,\mu,q}$ with the norm $\|\cdot\|_{\delta,\mu,q}$, then $\widehat X_{\delta,\mu,q}$ are
Banach spaces. Note that if $f\in \widehat X_{\delta,\mu+\epsilon,q}$ for some $\epsilon>0$, then $\RR{\mu} f:(0,\delta]\rightarrow H^q(T^1)$ is uniformly continuous.

By definition, all functions in $X_{\delta,\mu,q}$ can be approximated by smooth functions. Functions in $\widehat X_{\delta,\mu,q}$, however, can be approximated by a particularly useful  sequence of smooth functions as follows.

\begin{lemma}
\label{lem:mollifiers}
Let $f\in\widehat X_{\delta,\mu,q}$; i.e., $\RR{\mu} f:(0,\delta]\rightarrow H^q(T^1)$ is bounded and continuous. Let $\widehat f$ be defined as follows
\[\widehat f(t)=
\begin{cases}
 f(t),                   & t\in (0,\delta],
\\
 \RR{\mu}^{-1}(t)\RR{\mu}(\delta)f(\delta),            &   t\in [\delta,\infty).
\end{cases} 
\]
Let $\phi:\R\rightarrow\R$ be smooth with $\phi(x)>0$ for all $|x|<1$ and $\phi(x)=0$ for all $|x|\ge 1$, with $\int_\R\phi(x)dx=1$. Let $(\alpha_i)$ be a sequence of positive numbers with limit $0$. For any integers $i,j$, we set
\begin{equation}
  \label{eq:mollifier}
(\RR{\mu} f)_{i,j}(t,x):=\int_0^{\infty}\int_{T^1}(\RR{\mu} \widehat f)(s,y)\frac 1{\alpha_i}\phi\left(\frac{x-y}{\alpha_i}\right) \frac 1{\alpha_j}\phi\left(\frac{s-t}{\alpha_j}\right)\,dy\,ds.
\end{equation}
Then $(\RR{\mu} f)_{i,j}$ has the following properties:
\begin{enumerate}[label=\textit{(\roman{*})}, ref=(\roman{*})]   
\item $(\RR{\mu} f)_{i,j}\in C^\infty((0,\delta]\times T^1)$ for all integers $i,j$.
\item The function
  \begin{equation}
    \label{eq:mollifier2}
    f_{i,j}:=\RR{\mu}^{-1}(\RR{\mu} f)_{i,j}
  \end{equation}
  has the property that
  \[f_{i,j}\in \widehat X_{\delta,\mu,q}\cap X_{\delta,\mu,q}\quad \text{for all integers $i,j$}.\] 
  In particular,  for any given integers $i,j$, one has 
  \[\|(\RR{\mu} f)_{i,j}(t,\cdot)\|_{H^q(T^1)}\le C \|f\|_{\delta,\mu,q},\quad\text{for all $t\in (0,\delta]$},\]
  for a constant $C>0$ independent of $t$ (but possibly dependent on $i$, $j$).
\item \label{cond:pointwiseconvegrence} $(\RR{\mu} f)_{i,j}(t,x)\longrightarrow \RR{\mu} f(t,x)$ for $i,j\rightarrow\infty$ at a.e.\ $(t,x)\in (0,\delta]\times T^1$.
\item \label{cond:uniformconvergence} If $f$ is such that $\RR{\mu} f:(0,\delta]\rightarrow H^q(T^1)$ is a uniformly continuous map (e.g., if $f\in\widehat X_{\delta,\mu+\epsilon,q}$ for some $\epsilon>0$), then
\[\| f_{i,j}-f\|_{\delta,\mu,q}\rightarrow 0\quad \text{for $i,j\rightarrow\infty$}.\]
\end{enumerate}
\end{lemma}

\begin{proof}
Observe that $\RR{\mu}\widehat f$ is a bounded continuous map $(0,\infty)\rightarrow H^q(T^1)$ since $\RR{\mu}(t)\widehat f(t)=\RR{\mu}(\delta)f(\delta)$ for all $t\ge\delta$. We obtain $\widehat f\in \widehat X_{\infty,\mu,q}$ and $\|\widehat f\|_{\infty,\mu,q}=\|f\|_{\delta,\mu,q}$.
The first two properties of the lemma can be proven by standard arguments. The third one follows from Lebesgue's Differentiation Theorem. We only discuss the fourth property. If we fix any $t\in (0,\delta]$, then\footnote{We write $H^q_x(T^1)$ to stress that this is the Sobolev space with respect to integration over the $x$-variable.} we calculate
\begin{align*}
  &\|\RR{\mu}\left(|\RR{\mu}^{-1}(\RR{\mu} f)_{i,j}(t,x)-f(t,x)\right)\|_{H^q_x(T^1)}\\
  &= \|(\RR{\mu} f)_{i,j}(t,x)-(\RR{\mu} f)(t,x)\|_{H^q_x(T^1)}\\
  &= \left\|\int_0^\infty\int_{T^1}\left((\RR{\mu} \widehat f)(s,y)-(\RR{\mu} f)(t,x)\right)\frac 1{\alpha_i}\phi\left(\frac{x-y}{\alpha_i}\right) \frac 1{\alpha_j}\phi\left(\frac{s-t}{\alpha_j}\right)\,dy\,ds \right\|_{H^q_x(T^1)},
\end{align*}
as a consequence of the condition that $\int_\R\phi(x)dx=1$. Now we write
\[(\RR{\mu} \widehat f)(s,y)-(\RR{\mu} f)(t,x)=(\RR{\mu} \widehat f)(s,y)-(\RR{\mu} \widehat f)(s,x)+(\RR{\mu}\widehat f)(s,x)-(\RR{\mu} f)(t,x),\]
and therefore obtain the estimate
\begin{align}
  \label{eq:twoterms}
  &\|\RR{\mu}\left(|\RR{\mu}^{-1}(\RR{\mu} f)_{i,j}(t,x)-f(t,x)\right)\|_{H^q_x(T^1)}\\
  &\le  \left\|\int_0^\infty\int_{T^1}\left((\RR{\mu} \widehat f)(s,y)-(\RR{\mu} \widehat f)(s,x)\right)\frac 1{\alpha_i}\phi\left(\frac{x-y}{\alpha_i}\right) \frac 1{\alpha_j}\phi\left(\frac{s-t}{\alpha_j}\right)\,dy\,ds \right\|_{H^q_x(T^1)}\notag\\
  &+ \left\|\int_0^\infty\int_{T^1}\left((\RR{\mu} \widehat f)(s,x)-(\RR{\mu} f)(t,x)\right)\frac 1{\alpha_i}\phi\left(\frac{x-y}{\alpha_i}\right) \frac 1{\alpha_j}\phi\left(\frac{s-t}{\alpha_j}\right)\,dy\,ds \right\|_{H^q_x(T^1)}.\notag
\end{align}
Writing the first term on the right hand side of \Eqref{eq:twoterms} as $I$, we estimate 
\begin{equation*}
 I \le \int_0^\infty  \left\|\int_{T^1}\left((\RR{\mu} \widehat f)(s,y)-(\RR{\mu} \widehat f)(s,x)\right)\frac 1{\alpha_i}\phi\left(\frac{x-y}{\alpha_i}\right) \,dy \right\|_{H^q_x(T^1)} \frac 1{\alpha_j}\phi\left(\frac{s-t}{\alpha_j}\right)\,ds. 
\end{equation*}
Now, it is a standard result for mollifiers that for every $s\in (0,\infty)$ 
\[\left\|\int_{T^1}\left((\RR{\mu} \widehat f)(s,y)-(\RR{\mu} \widehat f)(s,x)\right)\frac 1{\alpha_i}\phi\left(\frac{x-y}{\alpha_i}\right) \,dy \right\|_{H^q_x(T^1)}\le g_i(s),\]
where $\lim_{i\rightarrow\infty} g_i(s)=0$ at every $s$, and for every integer $i$, the function $g$ is continuous. Since $\RR{\mu} f$ is uniformly continuous, this function $g_i$ extends to the interval $[0,\infty)$ with the same properties. Since $\RR{\mu}(t)\widehat f(t)=\RR{\mu}(\delta)\widehat f(\delta)$ for all $t\ge\delta$, it follows that there  is a sequence $(\widehat g_i)$ with limit $0$, such that $g_i(s)\le \widehat g_i$ for all $s\in[0,\infty)$. Consequently, $I$ can be estimated by a sequence  $(a_i)$, which (i) is independent of $j$, (ii) is independent of $t$, and (iii) goes to zero in the limit $i\rightarrow\infty$. 

We  now discuss the second term of  the right hand side of \Eqref{eq:twoterms}, which we label as $II$. The integral over $y$ is trivial, so consequently
\begin{equation*}
 II \le
 \int_0^\infty  \left\|(\RR{\mu} \widehat f)(s,x)-(\RR{\mu} f)(t,x)\right\|_{H^q_x(T^1)} \frac 1{\alpha_j}\phi\left(\frac{s-t}{\alpha_j}\right) \,ds.
\end{equation*}
The term involving the $H^q$-norm is a uniformly continuous function in $s$ and $t$. Hence, from  Lebesgue's Differentiation Theorem and the definition of $\widehat f$, it follows that  the $s$-integral converges to $0$ for $j\rightarrow\infty$, independently of $t\in (0,\delta]$ and $i$. This completes the proof of the fourth property.
\end{proof}

We can now use Lemma \ref{lem:mollifiers} to relate the spaces $X_{\delta,\mu,q}$ and $\widehat X_{\delta,\mu,q}$.

\begin{lemma}
  \label{lem:relationspaces}
Fix a constant $\delta>0$, an exponent vector $\mu$, and a
  non-negative integer $q$; then for all $\epsilon>0$, 
  \[\widehat X_{\delta,\mu+\epsilon,q} \subset X_{\delta,\mu,q}\subset \widehat X_{\delta,\mu,q}.\]
\end{lemma}

\begin{proof}
The inclusion $X_{\delta,\mu,q}\subset \widehat X_{\delta,\mu,q}$ follows easily from the fact that each element in $X_{\delta,\mu,q}$ is the limit of a Cauchy  sequence in $(C^\infty((0,\delta]\times T^1),\|\cdot\|_{\delta,\mu,q})$, whose elements are in particular bounded continuous maps $(0,\delta]\rightarrow H^q(T^1)$, and the convergence is uniform in time. 

To check the inclusion $\widehat X_{\delta,\mu+\epsilon,q} \subset X_{\delta,\mu,q}$, let a function
$f$ be given in $\widehat X_{\delta,\mu+\epsilon,q}$. Hence $f$ satisfies the condition of the previous lemma, in particular that of \Conditionref{cond:uniformconvergence}. It follows  that $f\in X_{\delta,\mu,q}$.
\end{proof}

We also wish to comment on time derivatives of functions in $X_{\delta,\mu,q}$ and $\widehat X_{\delta,\mu,q}$. Let $f\in \widehat X_{\delta,\mu,q}$. We say that $f$ is differentiable in time $t$ if the (bounded continuous) map $\RR{\mu} f:(0,\delta]\rightarrow H^q(T^1)$ is differentiable in the sense of a map between Banach spaces (Frechet derivatives). Its time derivative (multiplied by $t$) $D(\RR{\mu} f)$ can then be considered to be a map $(0,\delta]\rightarrow H^q(T^1)$, and we set $Df:=\RR{\mu}^{-1}(D(\RR{\mu} f)-D\RR{\mu} f)$. If this map is continuous, then we call $f$ continuously differentiable in $t$. If this is the case for $f$ and if in addition $\RR{\mu} Df$ is bounded, then we have $Df\in\widehat X_{\delta,\mu,q}$.

Now, let $f\in \widehat X_{\delta,\mu,q}$ be continuously differentiable. Then $Df$ is the \keyword{distributional time derivative} of $f$ in the following sense. Let $\phi$ be any test function with the properties as in \Sectionref{sec:lineartheory}. Choose $\epsilon>0$. Then we clearly have that
\[\int_\epsilon^\delta \partial_t(t \scalarpr{\RR{\mu} f}{\phi}_{L^2(T^1)}) dt=-\epsilon \left.\scalarpr{\RR{\mu} f}{\phi}_{L^2(T^1)}\right|_{t=\epsilon}.\]
Hence, the boundary term vanishes in the limit $\epsilon\rightarrow 0$. The following integrals are meaningful for $\epsilon=0$, and hence we obtain
\begin{equation}
  \label{eq:distribtimederivative}
  \begin{split}
  \int_0^\delta &\scalarpr{\RR{\mu} Df}{\phi}_{L^2(T^1)} dt\\
  &=-\int_0^\delta\left(\scalarpr{\RR{\mu} f}{D\phi}_{L^2(T^1)}
    +\scalarpr{\RR{\mu} f+D\RR{\mu} f}{\phi}_{L^2(T^1)}\right) dt.
\end{split}
\end{equation}
The reader should compare this with the expressions for weak solutions in \Sectionref{sec:lineartheory}.

\section{On products of functions}
\label{sec:productsfunctions}

We readily check the following  results which are useful in dealing 
 with products of functions and their relationship to the function spaces $X_{\delta,\mu,q}$.

\begin{lemma}
\label{lem:product1}
  Let $f\in X_{\delta,\mu_1,q}$ and $g\in X_{\delta,\mu_2,q}$
  be two functions $(0,\delta]\times T^1\rightarrow\R$,
  for some constant $\delta>0$, some smooth exponents $\mu_1$
  and $\mu_2$, and an integer $q\ge 1$. Then $f\cdot g$ is in $X_{\delta,\mu_1+\mu_2,q}$ and,   for some constant $C>0$, 
  \[\|f\cdot g\|_{\delta,\mu_1+\mu_2,q}\le C
  \|f\|_{\delta,\mu_1,q}\cdot
  \|g\|_{\delta,\mu_2,q}.
\]
\end{lemma}
Observe that the condition $q\ge 1$ (for one spatial dimension) is essential here.

\begin{proof}
An essential part of the proof of this lemma is the general estimate
\[\|f\cdot g\|_{H^{q}}\le
C(\|f\|_{H^{q}}\|g\|_{L^\infty}+\|g\|_{H^{q}}\|f\|_{L^\infty}),
\] 
for arbitrary functions $f$ and $g$ in $H^q\cap L^\infty$; see Proposition~3.7 in Chapter~13
of \cite{Taylor:2011wn}.  The Sobolev inequalities for $q\ge 1$ in one
spatial dimension then imply
\[\|f\cdot g\|_{H^{q}}\le
C(\|f\|_{H^{q}}\|g\|_{H^{q}}).
\] 
Working with this inequality, we see that the lemma 
follows easily if we choose a sequence $(f_i)$ which converges to $f$ in the function space $X_{\delta,\mu_1,q}$, and a sequence $(g_i)$ which converges to $g$ in  $X_{\delta,\mu_2,q}$, and then write
\[f_i\cdot g_i-f\cdot g=f_i\cdot(g_i-g)+g\cdot(f_i-f).\]
\end{proof}

Another important result is the following.

\begin{lemma}
  \label{lem:productmatrix}
  Let $w$ be a $d$-vector-valued function in $X_{\delta,\mu,q}$ for some exponent $d$-vector $\mu$, a constant $\delta>0$, and an integer $q\ge 1$.
  Let $S$ be a $d\times d$-matrix-valued function so that $\RR{\mu}\cdot S\cdot \RR{-\mu}$ is an element of $X_{\delta,\xi,q}$ for an exponent $d\times d$-matrix $\xi$ of the form $\xi_{ij}=\zeta_i$ where $\zeta$ is an exponent $d$-vector.
Then, the $d$-vector-valued function $S\,w$ is in $X_{\delta,\zeta+\mu,q}$ and
\[\|S\, w\|_{\delta,\zeta+\mu,q}\le C \|\RR{\mu}\cdot S\cdot \RR{-\mu}\|_{\delta,\xi,q} \|w\|_{\delta,\mu,q},\]
for some constant $C>0$.
\end{lemma}

This lemma is proved essentially in the same way as \Lemref{lem:product1}.

\begin{lemma}
  \label{lem:product}
  Suppose that $\delta>0$, $s>0$ and $r>0$ are constants, $n$, $d$ and $q$ integers with $d\ge 1$ and $q\ge 1$, $\mu$ an exponent $d$-vector, and $\nu_1$ and $\nu_2$ exponent scalars.   Let functions $g_1,g_2:U\rightarrow\R$ be given where $U$ is an open subset of $\R^d$. Suppose that $g_1$ maps all functions $w:(0,\delta]\times T^1\rightarrow\R^d$ in $B_{\delta,\mu,q,s}$ to elements $g_1(w)$ in $B_{\delta,\nu_1,q,r}$. Moreover suppose that there is a constant $C_1>0$ with
\[\|g_1[w_1]-g_1[w_2]\|_{\delta,\nu_1,q}
    \le C_1 \|w_1-w_2\|_{\delta,\mu,q},\]
for all $w_1,w_2:(0,\delta]\times T^1\rightarrow\R^d$ in $B_{\delta,\mu,q,s}$. Let us also assume that $g_2$ maps all functions $w:(0,\delta]\times T^1\rightarrow\R^d$ in $B_{\delta,\mu,q,s}$ to elements $g_2(w)$ in $B_{\delta,\nu_2,q,r}$ and that there is a constant $C_2>0$ with
\[\|g_2[w_1]-g_2[w_2]\|_{\delta,\nu_2,q}
    \le C_2 \|w_1-w_2\|_{\delta,\mu,q},\]
for all $w_1,w_2:(0,\delta]\times T^1\rightarrow\R^d$ in $B_{\delta,\mu,q,s}$.
Now, consider
  $h:=g_1\cdot g_2$, $w\mapsto h(w)$. Then, there exists a $\rho>0$ (which is smaller the smaller $r$ is) so that $h$ maps all functions $w:(0,\delta]\times T^1\rightarrow\R^d$ in $B_{\delta,\mu,q,s}$ to elements $h(w)$ in $B_{\delta,\nu_1+\nu_2,q,\rho}$. Moreover, there exists a
  constant $C>0$ with
  \[\|h[w_1]-h[w_2]\|_{\delta,\nu_1+\nu_2,q}
  \le C \|w_1-w_2\|_{\delta,\mu,q},\]
  for all $w_1,w_2:(0,\delta]\times T^1\rightarrow\R^d$ in $B_{\delta,\mu,q,s}$.
\end{lemma}

\begin{proof}
If $w\in B_{\delta,\mu,q,s}$, then $g_1(w)\in B_{\delta,\nu_1,q,r}$ and $g_2(w) \in B_{\delta,\nu_2,q,r}$. \Lemref{lem:product1} implies that $h(w)=g_1(w) g_2(w)\in X_{\delta,\nu_1+\nu_2,q}$ and 
\[\|h(w)\|_{\delta,\nu_1+\nu_2,q}\le C \|g_1(w)\|_{\delta,\nu_1,q}\|g_2(w)\|_{\delta,\nu_2,q}\le C r^2,\]
where $C>0$ is the constant in \Lemref{lem:product1}. This allows us to set $\rho=C r^2$ and hence establishes that $h(w)\in B_{\delta,\nu_1+\nu_2,q,\rho}$. Regarding the Lipschitz estimate, we find
\begin{equation}
\label{eq:expandproduct}
\begin{split}
\|t^{-(\nu_1+\nu_2)} (h[w_1](t)-h[w_2](t))\|_{H^q}&\le C\| t^{-\nu_1}(g_1[w_1](t)-g_1[w_2](t))\|_{H^q} \|t^{-\nu_2} g_2[w_1](t)\|_{H^q} \\
&+C\| t^{-\nu_1}g_1[w_2](t)\|_{H^q}\| t^{-\nu_2}(g_2[w_1](t)-g_2[w_2](t))\|_{H^q}.
\end{split}
\end{equation}
Then we can use the individual Lipschitz estimates for $g_1$ and $g_2$ in order to establish this result.

\end{proof}

While  \Lemref{lem:product} is adequate for the proof of Theorem \ref{th:Wellposedness1stOrderFiniteDiff}, to prove Theorem \ref{th:Wellposedness1stOrderHigherOrder} we require a stronger result, which we present here.

\begin{lemma}
  \label{lem:productII}
  Suppose that $q\ge 1$. Let $g_1$ and $g_2$ be functions satisfying
  all the conditions of \Lemref{lem:product} with exponents $\nu_1$,$\nu_2$ for all $x\in T^1$. Suppose that, in addition, one has the following: For all $w\in
  B_{\delta,\mu,q,s/2}$ with $\omega\in B_{\delta,\widehat\mu,q,s/2}$ for some
  exponent vector $\widehat\mu$ which satisfies $\widehat\mu\ge\mu$, there
  exist scalar exponents $\gamma_1,\gamma_2$, independent of
  $\widehat\mu$, such that
  \begin{equation*}
    g_1(w+\omega)-g_1(w)\in X_{\delta,\widehat\mu+\gamma_1,q},\quad
    g_2(w+\omega)-g_2(w)\in X_{\delta,\widehat\mu+\gamma_2,q},
  \end{equation*}  
  and
  \begin{align*}    
    \|g_1[w+\omega]-g_1[w]\|_{\delta,\widehat\mu+\gamma_1,q}
    &\le \widehat C_1\|\omega\|_{\delta,\widehat\mu,q},\\
    \|g_2[w+\omega]-g_2[w]\|_{\delta,\widehat\mu+\gamma_2,q}
    &\le \widehat C_2 \|\omega\|_{\delta,\widehat\mu,q},
  \end{align*}
  for constants $\widehat C_1,\widehat C_2>0$.

  Then the function $h=g_1\cdot g_2$ has the following property. We  can choose a scalar exponent  $\gamma$ smaller or equal than $\min\{\nu_1+\gamma_2,\nu_2+\gamma_1\}$ (independently of
  $\widehat\mu$), such that for all $w\in B_{\delta,\mu,q,s/2}$ and
  $\omega\in B_{\delta,\widehat\mu,q,s/2}$, one has  
  \begin{equation*}
    h(w+\omega)-h(w)\in X_{\delta,\widehat\mu+\gamma,q},
  \end{equation*}  
  and
  \begin{equation*}    
    \|h[w+\omega]-h[w]\|_{\delta,\widehat\mu+\gamma,q}
    \le \widehat C \|\omega\|_{\delta,\widehat\mu,q} 
  \end{equation*}
for a constant $\widehat C>0$.
\end{lemma}

This follows from a more detailed analysis of \Eqref{eq:expandproduct}.

To handle the exponential function,
we rely on  the following result.

\begin{lemma}
\label{lem:exponential}
Pick constants $\delta>0$, $s>0$, an integer $q\ge 1$, and an exponent $\mu>0$. Let $g^{(i)}:=\exp\circ\Pi_i$, where
  $\Pi_i:\R^d\rightarrow\R$ is the projection to the $i$th
  component of $d$-vectors.
Then, for every function $w:(0,\delta]\times T^1\rightarrow\R$ in ${B}_{\delta,
  \mu,q,s}$, there exists an $r>0$, so that the composed function $g^{(i)}\circ w:(0,\delta]\times T^1\rightarrow\R$ is in $ B_{\delta, 0,q,r}$. Moreover, for all $w_1,w_2\in {B}_{\delta, \mu,q,s}$, there exists a constant $C>0$, so that
\begin{equation*}
  \|g^{(i)} (w_1)-g^{(i)} (w_2)\|_{\delta,0,q}\le C\|w_1-w_2\|_{\delta,\mu,q}.
\end{equation*}
In addition, for every scalar exponent $\hat\mu\ge\mu$ and every $w\in B_{\delta,\mu,q,s/2}$ and $\omega\in B_{\delta, \hat\mu,q,s/2}$, it follows that $g^{(i)} (w+\omega)-g^{(i)} (w)$ is in $X_{\delta,\hat\mu,q}$ and the estimate
\begin{equation*}
\|g^{(i)} (w+\omega)-g^{(i)}(w)\|_{\delta,\hat\mu,q}\le C \|\omega\|_{\delta,\hat\mu,q},
\end{equation*}
holds.
\end{lemma}

\begin{proof}
This follows from Proposition~3.9 in Chapter~13 of \cite {Taylor:2011wn}
applied to $g^{(i)}(w)-1$, together with the Taylor theorem for the exponential
function.
\end{proof}

\section{Duality and convergence results}

\subsection*{Sobolev spaces and duality}

Following  \cite[Chapter~VI]{ChoquetBruhat:1982tv} or \cite{Ringstrom:2009cj}, one defines the Sobolev space $H^s(\R^n)$ for any $s\in\R$ as the set of temperate distributions $u$ such that $\widehat u (1+|\xi|^2)^{s/2}\in
L^2(\R^n)$, where $\widehat u:=\mathcal F u$ is the Fourier transform (in the sense of temperate distributions) of $u$.
The norm defined by
\[\|u\|_s:=\|\widehat u(\xi) (1+|\xi|^2)^{s/2}\|_{L_\xi^2(\R^n)}
\] 
turns this space into a Banach space. If $s=q$ for any non-negative
integer $q$, then $H^s(\R^n)$ is equivalent to  the standard ($p=2$) Sobolev space
$H^q(\R^n)$. For general $s\in\R$, the space $H^s(\R^n)$ is in fact a
Hilbert space for the scalar product
\[\scalarpr{u}{v}_{s}:=\int_{\R^n}{\widehat u(\xi) (1+|\xi|^2)^{s/2}}{\widehat v(\xi) (1+|\xi|^2)^{s/2}} d\xi.\]

Let $u\in H^{-s}(\R^n)$ and $v\in H^s(\R^n)$ for any $s\in\R$. Then the \keyword{dual pairing} between $H^s(\R^n)$ and $H^{-s}(\R^n)$,
\begin{equation}
\label{eq:pairing}
\pairing{u}{v}:=\int_{\R^n}\widehat u(\xi) \widehat v(\xi) d\xi,
\end{equation}
is well-defined, as a consequence of the inequality
\begin{equation}
  \label{eq:pairingestimate}
  |\pairing{u}{v}|\le \left|\int_{\R^n}\widehat u(\xi)  (1+|\xi|^2)^{-s/2} \widehat v(\xi)  (1+|\xi|^2)^{s/2} d\xi\right|
  \le \|u\|_{-s} \|v\|_s.
\end{equation}
By means of this pairing, we can identify $H^{-s}(\R^n)$ with $H^s(\R^n)^*$ (the dual space) as follows. 
For every $u\in H^{-s}(\R^n)$, the map
$\pairing{u}{\cdot}: H^s(\R^n)\rightarrow\R$ is a bounded linear
functional, i.e., an element of $H^{s}(\R^n)^*$. Conversely, according to the Riesz representation theorem, there exists a unique element $w_\phi\in H^s(\R^n)$ for each
element $\phi\in H^{s}(\R^n)^*$ such that
\[\phi(v)=\scalarpr{w_\phi}{v}_s
\] 
for all $v \in H^s(\R^n).$
The last expression can be written as
\begin{equation*}
  \scalarpr{w_\phi}{v}_s=\int_{\R^n}{\widehat w_\phi(\xi) (1+|\xi|^2)^{s/2}}{\widehat v(\xi) (1+|\xi|^2)^{s/2}} d\xi
  =\int_{\R^n}{\widehat v_\phi(\xi)}{\widehat v(\xi)} d\xi,
\end{equation*}
where $\widehat v_\phi:=\widehat w_\phi(\xi) (1+|\xi|^2)^{s}$ is the Fourier transform of $v_{\phi}:=\mathcal F^{-1}(\widehat w_\phi(\xi) (1+|\xi|^2)^{s})$. 
We have $v_\phi\in H^{-s}(\R^n)$, since $\widehat v_\phi (1+|\xi|^2)^{-s/2}=\widehat w_\phi(\xi) (1+|\xi|^2)^{s/2}\in L^2(U)$. By means of the pairing above, we have thus constructed a unique element $v_\phi\in H^{-s}(\R^n)$ corresponding to each $\phi\in H^s(\R^n)^*$. In this sense, we can therefore identify $H^{-s}(\R^n)$ with $H^s(\R^n)^*$ for every $s\in\R$.

The following result concerns the relationship between Sobolev spaces of different indices.
\begin{proposition}
  \label{prop:sobolevdensesubsets}
  For every $s\in\R$ and $\sigma\ge 0$, the space $H^{s+\sigma}(\R^n)$ is a
  dense subset of $H^{s}(\R^n)$.
\end{proposition}
\begin{proof}
We first  show that $H^{s+\sigma}(\R^n)$ is indeed a subset of $H^{s}(\R^n)$ for $\sigma\ge
0$. Suppose that $u\in H^{s+\sigma}(\R^n)$.  Calculating the $\| \cdot \|_s$ norm of $u$, we obtain
\[\|u\|_{s}^2=\int_{\R^n} |\widehat u(\xi)|^2 (1+|\xi|^2)^s d\xi \le
\int_{\R^n} |\widehat u(\xi)|^2 (1+|\xi|^2)^{s+\sigma} d\xi
=\|u\|_{s+\sigma}^2<\infty,
\] 
from which it follows that  $u\in H^{s}(\R^n)$. To check that $H^{s+\sigma}(\R^n)$ is a \emph{dense} subset, it is sufficient to note (see, e.g., \cite{ChoquetBruhat:1982tv})
that $C_0^\infty(\R^n)$ 
(the space of smooth functions with compact support) is dense in both
$H^{s}(\R^n)$ and $H^{s+\sigma}(\R^n)$.
\end{proof}

\subsection*{Convergence results in Sobolev spaces}

One can use this dense inclusion property (Proposition \ref{prop:rescueregularity}) together with the duality properties discussed above to derive certain convergence and closedness-type results for sequences in Sobolev spaces. We first discuss a result of this sort for Sobolev spaces on $\R^n$, and then do the same for Sobolev spaces on $T^1$.

\begin{proposition}
  \label{prop:rescueregularity}
 Choose $s, s_0\in\R$ so that $0\le s_0<s$. Let $(w_m)$ be a bounded sequence in $H^s(\R^n)$ in the sense that there exists a constant $C>0$ so that $\|w_m\|_s\le C$, for all integer $m$. Moreover, suppose that $(w_m)$ converges to some $w\in H^{s_0}(\R^n)$;  i.e., $\|w_m-w\|_{s_0}\rightarrow 0$.  Then, $w$ is contained in $H^{s}(\R^n)$.
\end{proposition}

\begin{proof}
  The boundedness of the sequence implies the existence 
  of a subsequence of $(w_m)$ (which for simplicity we identify with $(w_m)$) which converges weakly. Hence, as a consequence of  the Riesz Representation Theorem and the above dual pairing in \Eqref{eq:pairing}, there exists an element $\widetilde w\in H^s(\R^n)$, so that, for every $Y\in H^{-s}(\R^n)$, 
\begin{equation}
  \label{eq:weakconvergence}
  \pairing {Y}{\widetilde w-w_m}\rightarrow 0 
\end{equation}
We wish to show that $w=\widetilde w$ and hence that $w\in H^{s}(\R^n)$. To do this, we consider an arbitrary $X\in H^{-s_0}(\R^n)$ and the dual pairing
\[\left|\pairing{X}{\widetilde w-w}\right|\le \left|\pairing{X}{\widetilde w-w_m}\right|+\left|\pairing{X}{w-w_m}\right|,\]
where $\widetilde w-w$ is considered as an element of $H^{-s_0}(\R^n)$, and where we have used the triangle inequality. Since $X\in H^{-s_0}(\R^n)\subset H^{-s}(\R^n)$ according to \Propref{prop:sobolevdensesubsets}, we can consider the first term on the right hand side as a pairing between $H^s(\R^n)$ and $H^{-s}(\R^n)$, and hence \Eqref{eq:weakconvergence} implies that this term can be made arbitrarily small by choosing $m$ sufficiently large. The second term is considered as a pairing between $H^{s_0}(\R^n)$ and $H^{-s_0}(\R^n)$ so that \Eqref{eq:pairingestimate} yields
\[\left|\pairing{X}{w-w_m}\right|\le \|X\|_{-s_0} \|{w-w_m}\|_{s_0}.\]
Also this term can be made arbitrarily small by choosing $m$ sufficiently large. Hence, we have found that $\pairing{X}{\widetilde w-w}=0$ for all $X\in H^{-s_0}(\R^n)$. Now, the Riesz representation theorem implies that for every $X\in H^{-s_0}(\R^n)$ there exists precisely one $\widetilde X\in H^{s_0}(\R^n)$ for which
\[0= \pairing{X}{\widetilde w-w}=\scalarpr{\widetilde X}{\widetilde w-w}_{H^{s_0}(\R^n)}.\]
In particular, we may choose $\widetilde X=\widetilde w-w$, which implies that $\widetilde w-w=0$.
\end{proof}

\begin{corollary}
\label{cor:rescueregularityT1}
 Choose non-negative integers $q$ and $q_0$ so that $q_0<q$. Let $(w_m)$ be a bounded sequence in $H^q(T^1)$, in the sense that there exists a constant $C>0$ so that $\|w_m\|_{H^q(T^1)}\le C$, for all integers $m$. Moreover, suppose that $(w_m)$ converges to some $w\in H^{q_0}(T^1)$; i.e., $\|w_m-w\|_{H^{q_0}(T^1)}\rightarrow 0$.  Then, $w$ is contained in $H^{q}(T^1)$.
\end{corollary}

\begin{proof}
We formulate the proof so that it can be easily generalized to general smooth orientable, connected compact Riemannian manifolds $M$ in any dimension $n$. For this paper, the relevant special case is 
$M=T^1$. Let $((U_i,\phi_i))$ be a collection of coordinate charts, i.e., open subsets $U_i\subset M$ and homeomorphisms $\phi_i: V_i\rightarrow U_i$ where $V_i:=\phi_i^{-1}(U_i)$ are open subset of $\R^n$, which cover $M$, i.e., $M=\bigcup_i U_i$. It follows from compactness that we can assume that there are $N$ such coordinate charts. Let $(\tau_i)$ be a subordinate partition of unity. Then we find that $(w_m)$ is a bounded sequence in $H^q(M)$ if and only if for all $i=1,\ldots,N$, we have that $(w_m\circ\phi_i)$ is a bounded sequence in $H^q(V_i)$. Moreover, $\|w_m-w\|_{H^{q_0}(T^1)}\rightarrow 0$ for some $w\in H^{q_0}(M)$ if and only if for all $i=1,\ldots,N$, we have that $\|w_m\circ\phi_i-w\circ\phi_i\|_{H^{q_0}(V_i)}\rightarrow 0$ (since $w\circ\phi_i\in H^{q_0}(V_i)$). Now, the Stein Extension Theorem (Theorem~5.24 in \cite{Adams:1975wi}) implies the existence of \keyword{total extension operators} $E_i$ (Definition~5.17 in \cite{Adams:1975wi}), which are linear maps $E_i$ from functions defined on $V_i$ to functions defined on $\R^n$ with the following property: If $f\in H^r(V_i)$ for any non-negative integer $r$, then
\begin{enumerate}
\item  $\left. (E_i f)\right|_{V_i}=f$ almost everywhere, 

\item $E_if$ is in $H^r(\R^n)$, and there exists a constant $C>0$, so that
  \[\|E_i f\|_{H^r(\R^n)}\le C \|f\|_{H^r(V_i)}.\]
\end{enumerate}
Hence, we find that  $(w_m)$ is a bounded sequence in $H^q(M)$ if and only if for all $i=1,\ldots,N$, we have that $(E_i (w_m\circ\phi_i))$ is a bounded sequence in $H^q(\R^n)$. Moreover, $\|w_m-w\|_{H^{q_0}(T^1)}\rightarrow 0$ for some $w\in H^{q_0}(M)$ if and only if for all $i=1,\ldots,N$, we have that $\|E_i(w_m\circ\phi_i)-E_i(w\circ\phi_i)\|_{H^{q_0}(\R^n)}\rightarrow 0$  (since  $E_i(w\circ\phi_i)\in H^{q_0}(\R^n)$).
It follows from \Propref{prop:rescueregularity}, that $E_i(w\circ\phi_i)\in H^{q}(\R^n)$. Hence, $w\circ\phi_i\in H^{q}(V_i)$. Since this is true for all $i=1,\ldots,N$, it follows that $w\in H^{q}(M)$.
\end{proof}

\end{document}